\documentclass[twoside]{bour}
\usepackage{rivasseau_h}

\usepackage{amsmath,amssymb,amsthm,latexsym}

\usepackage{stmaryrd,wasysym,upgreek,mathrsfs,dsfont}
\usepackage[english]{babel}
\usepackage{graphicx,color}
\usepackage{slashed,subfig}

\usepackage[pdftex]{hyperref}

\newcommand{\eqdef} {\stackrel{\rm def}{=}}

\newcommand{\lp}  {\left(}
\newcommand{\rp}  {\right)}
\newcommand{\Br}{\overline}
\newcommand{\no}{\nonumber}

\newcommand{\bqa}{\begin{eqnarray}}
\newcommand{\eqa}{\end{eqnarray}}
\newcommand{\be}{\begin{equation}}
\newcommand{\ee}{\end{equation}}
\newcommand{\al}{\alpha}

\newcommand{\ga}{\gamma}
\newcommand{\de}{\delta}

\newcommand{\la}{\lambda}

\newcommand{\ta}{\tau}

\newcommand{\ch}{\chi}
\newcommand{\Om}{\Omega}

\newcommand{\Si}{\Sigma}

\newcommand{\Ga}{\Gamma}

\theoremstyle{plain}
\newtheorem{thm}{Theorem}[section]
\newtheorem{lemma}[thm]{Lemma}
\newtheorem{prop}[thm]{Proposition}
\newtheorem{cor}[thm]{Corollary}

\theoremstyle{plain}
\newtheorem{defn}{Definition}[section]
\newtheorem*{rem}{Remark}

\newcommand{\beqa}{\begin{eqnarray}}
\newcommand{\eeqa}{\end{eqnarray}}
\newcommand{\bea}{\begin{eqnarray}}
\newcommand{\eea}{\end{eqnarray}}

\renewcommand{\mathbf}{\boldsymbol}

\newcommand{\encv}{{non-commutative}}

\newcommand{\les}{\leqslant}
\newcommand{\ges}{\geqslant}

\def\ZZ{{\mathchoice {\mathsf{Z \hspace{-0.45em} Z}} {\mathsf{Z 
        \hspace{-0.45em} Z}} {\mathsf{Z \hspace{-0.32em} Z}} 
    {\mathsf{Z \hspace{-0.23em} Z}}}}

\newcommand{\N}{\mathbb{N}}

\newcommand{\R}{\mathbb{R}}
\newcommand{\C}{\mathbb{C}}

\newcommand{\bbbone}{{\mathds{1}}}

\newcommand{\lnat}{\llbracket} 
\newcommand{\rnat}{\rrbracket} 

\newcommand{\defi}{\stackrel{\text{\tiny def}}{=}}
\newcommand{\wed}{\wedge}

\def\lbt{\left(}
  \def\rbt{\right)}
\def\labs{\left |}
  \def\rabs{\right |}

\def\lb{\left \{}
  \def\rb{\right \}}
\def\lsb{\left [}
  \def\rsb{\right ]}

\DeclareMathOperator*{\Tr}{Tr}

\DeclareMathOperator*{\seq}{\simeq}

\newcommand {\tqs}{\mathrel{:}}

\newcommand{\GN}{\text{GN}^{2}_{\Theta}}

\def\xt{\widetilde{x}}

\def\Ot{\widetilde{\Omega}}

\def\ps{\slashed{p}}

\def\xs{\slashed{x}}
\def\ys{\slashed{y}}

\def\xts{\slashed{\widetilde{x}}}
\def\yts{\slashed{\widetilde{y}}}

\def\psib{\bar{\psi}}

\newcommand\scF{{\mathscr F}}

\newcommand\cA{{\mathcal A}}

\newcommand\cF{{\mathcal F}}
\newcommand\cG{{\mathcal G}}

\newcommand\cI{{\mathcal I}}

\newcommand\cL{{\mathcal L}}
\newcommand\cM{{\mathcal M}}

\newcommand\cO{{\mathcal O}}

\newcommand\cR{{\mathcal R}}
\newcommand\cS{{\mathcal S}}

\newcommand\cW{{\mathcal W}}

\newcommand{\noi}{\noindent}

\numberwithin{equation}{section}

\begin{document}

\title{Non-Commutative Renormalization}
\author{Vincent {\sc Rivasseau}\\
Laboratoire de Physique Th\'eorique\\
 B\^at.\ 210\\
Universit\'e Paris XI\\
 F--91405 Orsay Cedex, France}

\maketitle

\begin{abstract}
A new version of scale analysis and renormalization theory has been found
on the non-commutative Moyal space. 
It could be useful for physics beyond the standard model or for standard
physics in strong external field.
The good news is that quantum field theory is better behaved on non-commutative
than on ordinary space: indeed it has no Landau ghost. Noncommutativity might therefore
be an alternative to supersymmetry. We review this rapidly growing subject. 

\end{abstract}


\section{Introduction}

The world as we know it today is made of about 61 different
scales if we use powers of ten\footnote{\label{parochial} Or of about 140 $e$-folds if we want to avoid
any parochialism due to our ten fingers. What is important is to measure distances on a 
logarithmic scale.}. Indeed there is a fundamental length
obtained by combining the three fundamental constants
of physics, Newton's gravitation constant $G$, Planck's
constant $\hbar$ and the speed of light $c$. It is the Planck length $\ell_P = \sqrt{\hbar G /c^3}$, 
whose value is about $1.6 \; 10^{-35}$ meters. Below this length
ordinary space time almost certainly has to be quantized, 
so that the very notion of scale might be modified.
But there is also a maximal observable scale or ``horizon" in the universe, not for fundamental
but for practical reasons. The current distance from the Earth to the edge of the visible universe is about 46 billion light-years in any direction\footnote{The age of the universe is only about 13.7 billion years, so one could believe the observable radius would be 13.7 billion light years.
This gives already a correct order of magnitude, but in our expanding universe spacetime is 
actually curved so that distances have to be measured in comoving coordinates. The light
emitted by matter shortly after the big-bang, that is about 13.7 billion years ago, that reaches us now
corresponds to a present distance of that matter to us that is almost three times bigger, see http://en.wikipedia.org/wiki/Observable\_universe.}. This translates into a comoving radius of the visible universe of about $4.4 \; 10^{26}$ meters, or more fundamentally
$2.7 \; 10^{61}$ Planck lengths. 
Although we do not observe galaxies that far away, 
the WMAP data indicate that the universe is really at least 80\% that big \cite{CSSK}.
The geometric mean between the size of the (observable) universe and the Planck's length stands
therefore around $10^{-4}$ meters, about the size of an (arguably very small) ant. 
In \cite{RV1}, we proposed to call this the ``antropic principle".

Among the roughly sixty scales of the universe, only about ten to eleven were 
relatively well known to ancient Greeks and Romans two thousand years
ago. We have now at least some knowledge of the 45 largest scales from $2\; 10^{-19}$ meters 
(roughly speaking the scale of 1 Tev, observable at the largest particle colliders on earth) up
to the size of the universe. This means that we know about three fourths of all scales. 
But the sixteen scales between $2\; 10^{-19}$ meters 
and the Planck length form the last true \emph{terra incognita} of physics.
Note that this year the LHC accelerator at Cern with maximum energy of about
10 Tev should start opening a window into a new power of ten. 
But that truly special treat also will mark the end of an era. The next 
fifteen scales between $2.10^{-20}$ meters and the Planck length
may remain largely out of direct reach in the foreseeable future,
except for the glimpses which are expected to come from the study
of very energetic but rare cosmic rays. Just as the Palomar mountain telescope remained 
the largest in the world for almost fifty years, we expect the LHC to
remain the machine with highest energy for a rather long time until
truly new technologies emerge\footnote{New 
colliders such as the planned linear $e^+$- $e^-$ international collider might be built soon. They
will be very useful and cleaner than the LHC, but they should remain for a long time 
with lower total energy.}.
Therefore we should try to satisfy our 
understandable curiosity about the \emph{terra incognita} in the coming decades
through more and more sophisticated indirect analysis. Here theoretical and 
mathematical physics have a large part to play because they will help us to better compare and 
recoup many indirect observations, most of them probably coming from astrophysics and cosmology, 
and to make better educated guesses.

I would like now to argue both that quantum field theory and renormalization are
some of the best tools at our disposal for such educated guesses, but also 
that very likely we shall also need some generalization of these concepts. 

Quantum field theory or QFT provides a quantum description of particles and interactions which is 
compatible with special relativity \cite{Pesk}-\cite{ItZu}-\cite{Ram}-\cite{GJ}. It is certainly essential
because it lies right at the frontier of the
\emph{terra incognita}. It is the accurate formalism
at the shortest distances we know,  
between roughly the atomic scale of $10^{-10}$ meters,
at which relativistic corrections to quantum mechanics start playing a significant role\footnote{For instance quantum electrodynamics explains the Lamb shift in the 
hydrogen atom spectrum.},
up to the last known scale of a Tev or $2 \; 10^{-19}$ meters. Over the years
it has evolved into the \emph{standard model} which explains in great detail
most experiments in particle physics and is contradicted by none.
But it suffers from at least two flaws. First it is 
not yet compatible with general relativity, that is Einstein's theory of gravitation.
Second, the standard model incorporates so many different Fermionic matter fields coupled by 
Bosonic gauge fields that it seems more some kind of new Mendeleyev table than a
fundamental theory. For these two reasons QFT and the standard model
are not supposed to remain valid without any changes until the
Planck length where gravitation should be quantized. They could 
in fact become inaccurate much before that scale.

What about renormalization? Nowadays renormalization is 
considered the heart of QFT, and even much more \cite{Salm}-\cite{PoincRen}-\cite{Riv1}.
But initially renormalization was little more than a trick, a quick fix to remove the divergences 
that plagued the computations of quantum electrodynamics. These divergences
were due to summations over exchanges of virtual particles with high momenta. 
Early renormalization theory succeeded in hiding 
these divergences into unobservable \emph{bare} parameters of the theory. In this way
the physical quantities, when expressed in terms of the \emph{renormalized}
parameters at observable scales, no longer showed any divergences. Mathematicians
were especially scornful. But many physicists also were not fully satisfied.
F. Dyson, one of the founding fathers
of that early theory, once told me: ``We believed renormalization would not last 
more than six months, just the time for us to invent something better..."

Surprisingly, renormalization survived and prospered. 
In the mid 50's Landau and others found a key difficulty, called the Landau ghost or triviality
problem, which plagued simple renormalizable QFT such as the $\phi^4_4$ theory
or quantum electrodynamics. Roughly speaking Landau showed that the
infinities supposedly taken out by renormalization were still there, 
because the bare coupling corresponding to
a non zero renormalized coupling became infinite at a very small but finite scale.
Although his argument was not mathematically fully rigorous, many physicists 
proclaimed QFT and renormalization dead and looked for a better theory. 
But in the early 70's, against all odds, they both made a spectacular comeback.
As a double consequence of better experiments but also of better computations,
quantum electrodynamics was demoted of its possibly fundamental status and incorporated 
into the larger electroweak theory of Glashow, Weinberg and Salam. 
This electroweak theory is still a QFT but with a non-Abelian gauge symmetry.
Motivated by this work 't Hooft and Veltman proved that renormalization
could be extended to non-Abelian gauge theories \cite{tHVe}. This difficult technical feat
used the new technique of dimensional renormalization to better respect the gauge symmetry. 
The next and key step was the extraordinary discovery that such non-Abelian gauge theories no longer
have any Landau ghost. This was done first by 't Hooft in some unpublished work, then by 
D. Gross, H. D. Politzer and F. Wilczek \cite{GWil1}-\cite{Pol}. D. Gross and F. Wilczek then used this 
discovery to convincingly formulate a non-Abelian gauge theory
of strong interactions \cite{GWil2}, the ones which govern nuclear forces, which they
called quantum chromodynamics. Remark that in 
every key aspect of this striking recovery, renormalization was no longer 
some kind of trick. It took a life of its own.

But as spectacular as this story might be, something even more important happened
to renormalization  around that time. In the hands of K. Wilson \cite{Wil} and others,
renormalization theory went out of its QFT cradle. Its scope expanded considerably.
Under the alas unfortunate name of renormalization group (RG), 
it was recognized as the right mathematical technique 
to move through the different scales of physics. 
More precisely over the years it became 
a completely general paradigm to study changes of scale,
whether the relevant physical phenomena are classical or quantum, 
and whether they are deterministic or statistical.
This encompasses in particular the full Boltzmann's program to deduce thermodynamics 
from statistical mechanics and potentially much more.
In the  hands of Wilson, Kadanoff, Fisher and followers, RG 
allowed to much better understand phase transitions in statistical mechanics, 
in particular the universality of critical exponents \cite{Zinn}.
The fundamental observation of K. Wilson was
that the change from bare to renormalized actions is too complex  a phenomenon 
to be described in a single step. Just like the
trajectory of a complicated dynamical system, it must be studied
step by step through a \emph{local} evolution equation. To summarize, do not jump over 
many scales at once!

Let us make a comparison between renormalization and geometry.
To describe a manifold, one needs a covering set of maps or atlas with crucial
\emph{transition} regions which must appear on different maps
and which are glued through transition functions.
One can then describe more complicated objects, such as bundles
over a base manifold, through connections which allow to parallel
transport objects in the fibers when one moves over the base.

Renormalization theory is both somewhat similar and somewhat different. It is some kind of geometry
with a very sophisticated infinite dimensional``bundle" part which loosely speaking 
describes the  effective actions. These actions 
flow in some infinite dimensional functional space. But 
at least until now the ``base" part is quite trivial:
it is a simple one-dimensional positive real axis, better viewed in fact as a full real axis if
we use logarithmic scales. We have indeed both positive and negative scales around
a reference scale of observation The negative or small spatial scales are called \emph{ultraviolet}
and the positive or large ones are called \emph{infrared} in reference to the 
origin of the theory in electrodynamics. An elementary step from one scale to the next is
called a renormalization group step. K. Wilson understood that there is an analogy
between this step and the elementary evolution step of a dynamical system. This analogy allowed
him to bring the techniques of classical dynamical systems into renormalization theory.
One can say that he was able to see the classical structure hidden in QFT.

Working in the direction opposite to K. Wilson,
G. Gallavotti and collaborators were able to see the quantum field theory
structure hidden in classical dynamics. For instance they understood secular averages
in celestial mechanics as a kind of renormalization \cite{GBGG}-\cite{GMas}. 
In classical mechanics, small denominators play the role
of high frequencies or ultraviolet divergences
in ordinary RG. The interesting physics consists 
in studying the long time behavior of the classical trajectories,
which is the analog of the infrared or large distance effects in 
statistical mechanics.

At first sight the classical structure discovered by Wilson in QFT and the quantum structure
discovered by Gallavotti and collaborators in classical mechanics are both surprising because 
classical and QFT perturbation theories look very different. Classical perturbation theory,
like the inductive solution of any deterministic equation, is indexed by trees, whether 
QFT perturbation theory is  indexed by more complicated ``Feynman graphs", 
which  contain the famous ``loops" of anti-particles responsible for the 
ultraviolet divergences \footnote{Remember that one can interpret 
antiparticles as going backwards in time.}. But the classical trees 
hidden inside QFT were revealed in many steps,
starting with  Zimmermann (which called them forests...)\cite{Zimm} through 
Gallavotti and many others,
until Kreimer and Connes viewed them as generators of Hopf algebras \cite{Kreimer:1997dp,Connes:2000uq,Connes:2001kx}. Roughly speaking 
the trees were hidden because they are not just subgraphs of the Feynman graphs.
They picture abstract inclusion relations
of the short distance connected components of the graph within the bigger components at larger scales.
Gallavotti and collaborators understood why there is a structure on the  trees which index
the classical Poincar\'e-Lindstedt perturbation series
similar to Zimmermann's forests in quantum field perturbation theory\footnote{In addition Gallavotti
also remarked that antimatter loops in Feynman graphs
can just be erased by an appropriate choice of non-Hermitian field interactions \cite{Gal}.}.

Let us make an additional remark which points to another fundamental
similarity between renormalization group flow
and time evolution. Both seem naturally \emph{oriented} flows.
Microscopic laws are expected to determine macroscopic laws, not the converse.
Time runs from past to future and entropy increases rather than decreases.
This is philosophically at the heart of standard determinism.
A key feature of Wilson's RG is to have defined in a mathematically precise way 
\emph{which} short scale  information should be forgotten through coarse graining: it is the 
part corresponding to the \emph{irrelevant operators} in the action.
But coarse graining is also fundamental
for the second law in statistical mechanics, which is the only law
in classical physics which is ``oriented in time" and also the one 
which can be only understood in terms of change of scales.

Whether this arrow common to RG and to time evolution
is of a cosmological origin remains to be further
investigated. We remark simply here that
in the distant past the big bang has to be explored and understood
on a logarithmic time scale. At the beginning of our universe
important physics is the one at very short distance.
As time passes and the universe evolves, physics at longer distances,
lower temperatures and lower momenta becomes literally visible.
Hence the history of the universe itself can be 
summarized as a giant unfolding of the renormalization group.

This unfolding can then be specialized into many
different technical versions depending on the particular 
physical context, and the particular problem at hand. RG has the potential to provide\
microscopic explanations for many phenomenological theories.
Hence it remains today a very active subject, with several 
important new brands developed in the two last decades at various levels of 
physical precision and of mathematical rigor.
To name just a few of these brands:

- the RG around extended singularities governs the quantum behavior of condensed
matter \cite{FT1}\cite{FT2}\cite{BG}. 
It should also govern the propagation of wave fronts and the long-distance
scattering of particles in Minkowski space.
Extended singularities alter dramatically the behavior of the renormalization group.
For instance because the dimension of the extended singularity of the Fermi surface
equals that of the space itself minus one, hence that of space-time minus \emph{two},
local quartic Fermionic interactions in condensed matter in \emph{any} dimension have
the same power counting than \emph{two} dimensional Fermionic field theories.
This means that condensed matter in \emph{any} dimension is similar 
to \emph{just renormalizable} field theory. Among the main consequences,
there is no critical mean field dimension in condensed matter except at infinity, but
there is a rigorous way to handle non perturbative phase transitions such as the
BCS formation of superconducting pairs through a dynamical $1/N$ expansion \cite{FMRT1}.

- the RG trajectories in dimension 2 between conformal theories
with different central charges have been pioneered in \cite{Zam}. 
Here the theory is less advanced, but again the $c$-theorem is a 
very tantalizing analog of Boltzmann's H-theorem.

- the functional RG of \cite{Wie} governs the behavior of many disordered systems.
It might have wide applications from spin glasses to surfaces.

\medskip

Let us return to our desire to glimpse into
the \emph{terra incognita} from currently known physics. We
are in the uncomfortable situations of salmons returning to their birthplace, since
we are trying to run against the RG flow. Many different bare actions lead to
the same effective physics, so that we may be lost in a maze.
However the region of \emph{terra incognita} closest to us is still far from the Planck scale.
In that region we can expect that any non renormalizable terms in the action 
generated at the Planck scale have been washed out by the RG flow
and renormalizable theories should still dominate 
physics. Hence renormalizability remains a guiding principle
to lead us  into the maze of speculations
at the entrance of \emph{terra incognita}. Of course we should also be alert and ready to incorporate
possible modifications of QFT as we progress towards the Planck scale, since we know 
that quantization of gravity at that scale will not happen through standard field theory.

String theory \cite{GSW} is currently the leading candidate for such a quantum theory of gravitation.
Tantalizingly the spectrum of massless particles of the closed string contains particles
up to spin 2, hence contains a candidate for the graviton. Open strings 
only contain spin one massless particles such as gauge Bosons. 
Since closed strings must form out of open strings through interactions,
it has been widely argued that string theory provides an
explanation for the existence of quantum gravity as a necessary complement
to gauge theories. This remains the biggest success of the theory up to now. 
It is also remarkable that string theory (more precisely membrane theory) allows 
some microscopic derivations of the Beckenstein-Hawking formula
for blackhole entropy \cite{StroVafa}. 

String theory also predicts two crucial features which are unobserved up to now, 
supersymmetry and six or seven new Kaluza-Klein dimensions of space time at short distance. 
Although no superpartner of any real particle has been found yet, there are some
indirect indications of supersymmetry, such as the careful study of the flows of the
running non-Abelian standard model gauge couplings\footnote{The three couplings join better at a single
very high scale if supersymmetry is included in the picture. Of course sceptics can remark that 
this argument requires to continue these flows deep within \emph{terra incognita}, where new physics
could occur.}. Extra dimensions might also be welcome, especially if they are significantly larger than the Planck scale, because they might provide an explanation 
for the puzzling weakness of gravitation with respect to other forces.
Roughly speaking gravitation could be weak because in string theory it propagates very naturally into 
such extra dimensions in contrast with other interactions which may remain stuck to our ordinary 
four dimensional universe or ``brane".

But there are several difficulties with string theory which cast some doubt on its usefulness
to guide us into the first scales of \emph{terra incognita}. First the theory
is really a very bold stroke to quantize gravity at the Planck scale, very far from current observations.
This giant leap runs directly against the step by step philosophy of the RG.
Second the mathematical structure of string theory is complicated 
up to the point where it may become depressing. For instance great effort
is needed to put the string theory at two loops on some rigorous footing \cite{HokPhong},
and three loops seem almost hopeless.
Third, there was for some time the hope that string theory and the phenomenology at lower energies derived
from it might be unique. This hope has now vanished with the discovery of a very complicated \emph{landscape} of different possible string vacua and associated long distance 
phenomenologies.

In view of these difficulties some physicists have started to openly criticize 
what they consider a disproportionate amount of intellectual resources 
devoted to the study of string theory compared to other alternatives \cite{Smolin}.

I do not share these critics. I think in particular that string theory has been
very successful as a brain storming tool. It has lead already to 
many spectacular insights into pure mathematics and geometry.
But my personal bet would be that if somewhere in the mountains near the Planck scale
string theory might be useful, or even correct, we should also search for other 
complementary and more reliable principles to guide us in the maze of waterways
at the entrance of \emph{terra incognita}. If these other complementary principles
turn out to be compatible with string theory at higher scales, so much the better.

It is a rather natural remark that since gravity alters the very geometry of ordinary space,
any quantum theory of gravity should quantize ordinary space,
not just the phase space of mechanics, as quantum mechanics does.
Hence at some point at or before the Planck scale we should expect 
the algebra of ordinary coordinates or observables 
to be generalized to a non commutative algebra.
Alain Connes, Michel Dubois-Violette, Ali Chamseddine and others have forcefully advocated 
that the \emph{classical} Lagrangian of the current standard
model arises much more naturally on simple non-commutative geometries
than on ordinary commutative Minkowsky space. We refer to Alain's lecture here for these arguments.
They remain in the line of Einstein's classical unification of Maxwell's electrodynamics 
equations through the introduction of a new four dimensional space-time. The next logical step seems to 
find the analog of quantum electrodynamics. It should be quantum field theory on 
non-commutative geometry, or NCQFT. The idea of NCQFT goes back at least to 
Snyders \cite{Snyders}.

A second line of argument ends at the same conclusion. String theorists 
realized in the late 90's that NCQFT is an effective theory of strings
\cite{a.connes98noncom,Seiberg1999vs}. Roughly this 
is because in addition to the symmetric tensor $g_{\mu\nu}$ the spectrum 
of the closed string also contains an antisymmetric tensor $B_{\mu\nu}$. There is no reason
for this antisymmetric tensor not to freeze at some lower scale into a classical field,
just as $g_{\mu\nu}$ is supposed to freeze into the classical metric of
Einstein's general relativity. But such a freeze of $B_{\mu\nu}$ 
precisely induces an effective non commutative geometry. 
In the simplest case of flat Riemannian metric and trivial constant antisymmetric tensor,
the geometry is simply of the Moyal type; it reduces to 
a constant anticommutator between (Euclidean) space-time coordinates. 
This made NCQFT popular among string theorists. A good review of these ideas 
can be found in \cite{DouNe}

These two lines of arguments, starting at both ends of \emph{terra incognita}
converge to the same conclusion: there should be an intermediate regime
between QFT and string theory where NCQFT is the right formalism. 
The breaking of locality and the appearance of cyclic-symmetric rather than
fully symmetric interactions in NCQFT is fully consistent with this intermediate
status of NCQFT between fields and strings. The ribbon graphs 
of NCQFT may be interpreted either as ``thicker particle world-lines"
or as ``simplified open strings world-sheets"
in which only the ends of strings appear but not yet their internal
oscillations. 
However until recently a big stumbling block remained. The simplest NCQFT on Moyal space,
such as $\phi^{\star 4}_4$, were found not to be renormalizable because of a surprising phenomenon called \emph{uv/ir mixing}. Roughly speaking this $\phi^{\star 4}_4$ theory still has infinitely many
ultraviolet divergent graphs but fewer than the ordinary $\phi^{4}_4$ theory. The new ``ultraviolet
convergent" graphs, such as the non-planar tadpole \includegraphics[scale=.7]{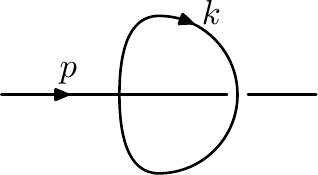} generate completely unexpected 
infrared divergences which are not of the renormalizable 
type \cite{MiRaSe}.

However three years ago the solution out of this riddle was found.
H. Grosse and R. Wulkenhaar in a brilliant series of papers
discovered how to renormalize $\phi^{\star 4}_4$ \cite{GrWu03-1,GrWu03-2,c}. 
This ``revolution" happened quietly without mediatic fanfare, but it might turn out to 
develop into a good Ariane's thread at the entrance of the maze.
Indeed remember the argument of Wilson: renormalizable theories
are the building blocks of physics because they are the ones who
survive RG flows...

It is always very interesting to develop a new brand of RG, but 
\emph{that} new brand on non commutative Moyal space is especially exciting.
Indeed it changes the very definition of scales in a new and non trivial way. 
Therefore it may ultimately change our view of locality 
and causality, hence our very view of the deterministic 
relationship from small to large distances. It is fair to say that the same is true
of string theory, where $T$-dualities also change small into large distances
and vice-versa. But in contrast with string theory, this new brand of
NCQFT is mathematically tractable, not at one or two loops, but as we shall see below,
at any  number of loops and probably even non-perturbatively!
This just means that we can do complicated computations in these NCQFT's
with much more ease and confidence than in string theory.

The goal of these lectures is to present this new set of burgeoning ideas.

We start with a blitz introduction to standard renormalization group
concepts in QFT: functional integration and Feynman graphs.
The system of Feynman graphs of the $\phi^4_4$ theory provide the simplest example
to play and experiment with the idea of renormalization. It is straightforward to analyze the
basic scaling behavior of high energy subgraphs within graphs 
of lower energy. In this way one discovers relatively  easily 
the most important physical effect under change
of the observation scale, namely the flow of the coupling constant. It leads
immediately to the fundamental difficulty associated to the short distance behavior
of the theory, namely the Landau ghost or triviality problem.
That ghost disappears in the ``asymptotically free" non-Abelian gauge theories \cite{GWil1}-\cite{Pol}. 
With hindsight this result might perhaps be viewed in a not so distant future as the first 
glimpse of NCQFT...

Grosse and Wulkenhaar realized that previous studies of NCQFT had used the wrong propagator! 
Moyal interactions were noticed to obey a certain Langmann-Szabo duality\cite{LaSz},
which exchanges space and momentum variables.
Grosse and Wulkenhaar realized that the propagator should be modified to also 
respect this symmetry \cite{c}. 
This means that NCQFT on Moyal spaces has to be based on the Mehler kernel,
which governs propagation in a harmonic potential,
rather than on the heat kernel, which governs 
ordinary propagation in commutative space.
Grosse and Wulkenhaar were able to compute for the first time the Mehler kernel in the
\emph{matrix base} which transforms the Moyal product 
into a matrix product. This is a real {\it tour de force}! The matrix based Mehler kernel
is quasi-diagonal, and they were able to use their computation to prove 
perturbative renormalizability of the theory, up to some estimates 
which were finally proven in \cite{Rivasseau2005bh}. 

By matching correctly propagator and interaction to respect symmetries, 
Grosse and Wulkenhaar were following one of the main successful thread of 
quantum field theory. Their renormalizability result is in the direct footsteps of 
't Hooft and Veltman, who did the same for non Abelian gauge theories thirty years before.
However I have often heard two main critics raised, which I would like to answer here. 

The first critic is that it is no wonder that adding a harmonic potential gets rid 
of the infrared problem. It is naive because the harmonic potential 
is the only partner of the Laplacian under LS duality. No other infrared 
regulator would make the theory renormalizable. 
The theory has infinitely many degrees of freedom, and
infinitely many divergent graphs, so the new BPHZ theorem 
obtained by Grosse and Wulkenhaar is completely non-trivial.
In fact now that the RG flow corresponding to these theories is better understood, 
we understand the former uv/ir mixing just as an ordinary anomaly 
which signaled a missing marginal term in the Lagrangian under that RG flow.

The second and most serious critic is that since the
Mehler kernel is not translation invariant, the Grosse and Wulkenhaar ideas will
 never be able to describe any mainstream physics in which there should be no preferred origin. 
This is just \emph{wrong} but for a more subtle reason.
We have shown that the Grosse-Wulkenhaar method can be extended to renormalize
theories such as the Langmann-Szabo-Zarembo $\bar\phi\star\phi \star\bar\phi\star\phi$ 
model \cite{Langmann2003if,Langmann2003cg,Langmann2002ai} 
in four dimensions or the Gross-Neveu model in two dimensions. 
In these theories the ordinary Mehler kernel is replaced by a related kernel which governs propagation 
of charged particles in a constant background field. This kernel, which we now propose to call
the \emph{covariant Mehler kernel}\footnote{Initially we called such NCQFT theories \emph{critical}, but it was pointed to us that this word may create
confusion with critical phenomena, so we suggest now to call them
\emph{covariant}.}, is still not translation
invariant because it depends on non translation-invariant gauge choice. 
It oscillates rather than decays when particles move away from a preferred origin.
But in such theories physical observables, which are gauge invariant, do not feel 
that preferred origin. That's why translation invariant phenomena can be described! 

We proposed to call the whole new class of NCQFT theories built either on the Mehler kernel 
or on its covariant generalizations \emph{vulcanized} (may be we should have spelled Wulkenized?)
because renormalizability means that 
their structure resist under change of scale \footnote{\label{vulca} Vulcanization is a technological operation
which adds sulphur to natural rubber to improve its mechanical properties and its resistance to
temperature change, and temperature is a scale in imaginary time...}. 

These newly discovered vulcanized theories or NCVQFT and their associated RG flows
absolutely deserve a thorough and systematic investigation, 
not only because they may be relevant for physics beyond the standard 
model, but also (although this is often less emphasized) because they 
may provide explanation of non-trivial 
effective physics in our ordinary standard world 
whenever strong background gauge fields are present. Many examples
come to mind, from various aspects of the quantum Hall effect to 
the behavior of two dimensional charged polymers under magnetic fields or even to
quark confinement. In such cases appropriate generalizations of the 
vulcanized RG may be the right tool to show how 
the correct effective non-local interactions emerge out of local interactions. 

At the Laboratoire de physique th\'eorique at Orsay we have embarked 
on such a systematic investigation of NCVQFTs and of their RG flows. 
This program is also actively pursued elsewhere.
Let us review briefly the main recent results and open problems.

\begin{itemize}

\item{\bf Multiscale Analysis}

The initial Grosse-Wulkenhaar breakthrough used sharp cutoffs
in matrix space, which like sharp cutoffs in ordinary direct and momentum space
are not so well suited to rigorous bounds and multiscale analysis.
By replacing these cutoffs by smoother cutoffs which cut directly the Mehler
parameter into slices, we could derive rigorously the estimates that were only numerically checked
in \cite{c} hence close the last gaps in the BPHZ theorem for vulcanized 
non commutative $\phi^{\star 4}_4$
\cite{Rivasseau2005bh}. We could
also replace the somewhat cumbersome recursive use of the Polchinski equation \cite{Polch}
by more direct and explicit bounds in a multiscale analysis.

\item{\bf Direct Space}

Although non translation invariant propagators and non local vertices are unfamiliar,
the direct space representation of NCVQFT remains closer to our ordinary intuition than the matrix base. 
Using direct space methods, we have provided a new proof of the BPHZ theorem for vulcanized 
non commutative $\phi^{\star 4}_4$ \cite{xphi4-05}.
We have also extended the Grosse-Wulkenhaar results to the $\bar\phi\star\phi \star\bar\phi\star\phi$ 
LSZ model \cite{Langmann2003if}. Our proof relies on a multiscale analysis analogous to 
\cite{Rivasseau2005bh} but in direct space. It allows a more transparent understanding of
the \emph{Moyality} of the counterterms for \emph{planar} subgraphs 
at higher scales when seen through external propagators at lower scales.
This is the exact analog of the \emph{locality} in ordinary QFT of general subgraphs 
at higher scales when seen through external propagators at lower scales. Such propagators do not 
distinguish short distance details, and ordinary locality could be summarized as the obvious remark
that from far enough away any object looks roughly like a point. But
\emph{Moyality}  could be summarized as a more surprising fact: viewed from lower RG
scales\footnote{These scales being defined in the new RG sense, we no longer say ``from far away".
Although I hate to criticize, I feel a duty here to warn the reader that 
often cited previous ``proofs of Moyality" such as \cite{Chepelev2000hm,CheRoi} 
should be dismissed. The main theorem in \cite{Chepelev2000hm}, whose proof 
never appeared, is simply wrong; and even more importantly 
the analysis in \cite{CheRoi} does not lead to any BPHZ theorem nor 
to any sensible RG flow. This is because using the old definition of 
RG scales it misses vulcanization.}, planar higher scale effects, which are the 
only ones large enough to require renormalization,
look like Moyal products.

\item{\bf Fermionic theories}

To enlarge the class of renormalizable non-commutative field theories and
to attack the quantum Hall effect problem it is essential
to extend the results of Grosse-Wulkenhaar
to Fermionic theories. Vulcanized Fermionic propagators have been computed
and their scaling properties established, both in matrix base and 
direct space, in \cite{toolbox05}. They seem to be necessarily of the covariant type. 

The simplest Fermionic NCVQFT theory, corresponding to the two-dimensional 
ordinary Gross-Neveu model, was then proved 
renormalizable to all orders in \cite{RenNCGN05}. 
This was done using the $x$-space version which seems also the most promising for a complete 
non-perturbative construction, using Pauli's principle to control the
apparent (fake) divergences of perturbation theory. 


\item{\bf Ghost Hunting} 

Grosse and Wulkenhaar made the first non trivial 
one loop RG computation in NCVQFT in \cite{GrWu04-2}.
Although they did not word it initially in this way, their result 
means that at this order there is no Landau ghost in NCVQFT!
A non trivial fixed point of the renormalization group develops at
high energy, where the Grosse-Wulkenhaar parameter $\Omega$ tends to 
the \emph{self-dual point} $\Omega=1$,
so that Langmann-Szabo duality become exact, and the beta function vanishes. 
This stops the growth of the bare coupling constant in the 
ultraviolet regime, hence kills the ghost. So after all NCVQFT is not
only as good as QFT with respect to renormalization, it is definitely 
better! This vindicates, although in a totally unexpected way, the initial 
intuition of Snyders \cite{Snyders}, who like many after him was at least partly motivated
by the hope to escape the divergences in QFT which were considered ugly.
Remark however that the ghost is not killed because of asymptotic freedom.
Both the bare and the renormalized coupling are non zero. They can
be made both small if the renormalized $\Omega$ is not too small, in which case  
perturbation theory is expected to remain valid all along the complete RG 
trajectory. It is only in the singular limit $\Omega_{ren} \to 0$ that the ghost begins
to reappear. 

For mathematical physicists who like me 
came from the constructive field theory program, the Landau ghost has always been
a big frustration. Remember that because non Abelian gauge theories are very complicated
and lead to confinement in the infrared regime, there is no good four dimensional
rigorous field theory without unnatural cutoffs up to now\footnote{We have only renormalizable constructive theories for two dimensional Fermionic theories \cite{GK1}-\cite{FMRS1} and for the infrared side of $\phi^4_4$\cite{GK2}-\cite{FMRS2}.}. 
I was therefore from the start very excited by the possibility to build 
non perturbatively the $\phi^{\star 4}_4$ theory as the first such rigorous
four dimensional field theory without unnatural cutoffs, even if it lives on the Moyal space 
which is not the one initially expected, and does not obey the usual axioms of ordinary QFT.

For that happy scenario to happen, two main non trivial steps are needed. The first
one is to extend the vanishing of the beta function at the self-dual point
$\Omega=1$ to all orders of perturbation theory. This has been done in 
\cite{DisertoriRivasseau2006,beta2-06}, using the matrix version of the theory.
First the result was checked by brute force computation
at two and three loops. Then we devised a general method for all orders. 
It relies on Ward identities inspired by those of similar 
theories with quartic interactions in which the beta function vanishes \cite{MdC,BM1,BM2}. 
However the relation of these Ward identities to the underlying LS symmetry remains unclear
and we would also like to develop again an $x$-space version of that 
result to understand better its relation to the LS symmetry.

The second step is to extend in a proper way constructive methods 
such as cluster and Mayer expansions to build non perturbatively 
the connected functions of NCVQFT in a single RG slice.
Typically we would like a theorem of Borel summability \cite{Sok} in the coupling constant 
for these functions which has to be \emph{uniform} in the slice index.
This is in progress. A construction of the model and of its full RG trajectory
would then presumably follow from a multiscale analysis 
similar to that of \cite{Abd2}.

\item{\bf $\phi^{\star 3}_6$ and Kontsevich model}

The noncommutative $\phi^{\star 3}$ model in 6 dimensions has been shown to be renormalizable, asymptotically free, and solvable genus by genus by mapping it to the Kontsevich model, in \cite{Grosse2005ig,Grosse2006qv,GrStei}.
The running coupling constant has also been computed exactly, and found
to decrease more rapidly than predicted by the one-loop beta function.
That model however is not expected to have a non-perturbative definition because
it should be unstable at large $\phi$.

\item{\bf Gauge theories}

A very important and difficult goal is to properly vulcanize gauge theories
such as Yang-Mills in four dimensional Moyal space or Chern-Simons on the two
dimensional Moyal plane plus one additional ordinary commutative time direction.
We do not need to look at complicated gauge groups since the $U(1)$ pure gauge theory 
is non trivial and interacting on non commutative geometry even without matter fields. 
What is not obvious is to find a proper compromise between gauge and Langmann-Szabo 
symmetries which still has a well-defined perturbation theory  around a computable vacuum
after gauge invariance has been fixed through appropriate Faddeev-Popov or BRS procedures. 
We should judge success in my opinion by one main criterion, namely
renormalizability. Recently de Goursac, Wallet and Wulkenhaar computed the
non commutative action for gauge fields which can be induced 
through integration of a scalar renormalizable matter field minimally coupled to the 
gauge field \cite{GourWW}; the result exhibits both gauge symmetry and LS covariance,
hence vulcanization, but the vacuum looks non trivial so that to check whether the associated 
perturbative expansion is really renormalizable seems difficult.

Dimensional regularization and renormalization better respect gauge symmetries
and they were the key to the initial 'tHooft-Veltman proof of renormalizability 
of ordinary gauge theories. 
Therefore no matter what the final word will be on 
NCV gauge theories, it should be useful to have the corresponding tools 
ready at hand in the non commutative context\footnote{The Connes-Kreimer works also use abundantly dimensional regularization and renormalization, and this is another motivation.}. 
This requires several steps, the first of which is

\item{\bf Parametric Representation}

In this compact representation, 
direct space or momentum variables have been integrated out for each Feynman amplitude.
The result is expressed as integrals over the heat kernel parameters of each propagator,
and the integrands are the topological polynomials of the graph\footnote{Mathematicians
call these polynomials Kirchoff polynomials, 
and physicist call them 
Symanzik polynomials in the quantum field 
theory context.}. These integrals
can then be shown analytic in the dimension $D$ of space-time for $\Re D$ small enough. They 
are in fact meromorphic in the complex plane, and ultraviolet divergences 
can be extracted through appropriate inductive contour integrations.

The same program can be accomplished in NCVQFT because the Mehler kernel is still quadratic
in space variables\footnote{This is true provided ``hypermomenta" are introduced to Fourier 
transform the space conservation at vertices which in Moyal space is the LS dual to ordinary 
momentum conservation.}. The corresponding topological hyperbolic
polynomials are richer than in ordinary field theory since
they are invariants of the ribbon graph which for instance
contain information about the genus of the surface on which these graphs live.
They can be computed both for ordinary NCVQFT \cite{gurauhypersyman}
and in the more difficult case of covariant theories such as
the LSZ model \cite{RivTan}.

\item{\bf Dimensional Regularization and Renormalization}

From the parametric representation the corresponding regularization and minimal dimensional
renormalization scheme should follow for NCVQFTs. However
appropriate factorization of the leading terms of 
the new hyperbolic polynomials under rescaling of the parameters of any subgraph is required.
This is indeed the analog in the parameter representation of 
the ``Moyality" of the counterterms in direct space. 
This program is under way \cite{GurTan}.

\item{\bf Quantum Hall Effect} 

NCQFT and in particular the non commutative Chern Simons theory has been 
recognized as effective theory of the quantum hall effect
already for some time \cite{Suss}-\cite{Poly}-\cite{HellRaam}.
We also refer to the lectures of V. Pasquier and of A. Polychronakos in this volume.
But the discovery of the vulcanized RG holds promises for a better explanation
of how these effective actions are generated from the microscopic level.

In this case there is an interesting reversal of
the initial Grosse-Wulkenhaar problematic. In the
$\phi^{\star 4}_4$ theory the vertex is given
a priori by the Moyal structure, and it is LS invariant. The challenge was
to find the right propagator 
which makes the theory renormalizable, and it turned out to
have LS duality. 

Now to explain the (fractional) quantum Hall effect, which is a bulk effect
whose understanding requires electron interactions, we can almost invert this logic. The propagator
is known since it corresponds to non-relativistic electrons in two dimensions
in a constant magnetic field. It has LS duality. But the effective 
theory should be anionic hence not local. 
Here again we can argue that among all possible non-local interactions, a few renormalization
group steps should select the only ones which form a renormalizable theory with the corresponding
propagator. In the commutative case (i.e. zero magnetic field)
local interactions such as those of the Hubbard model are just 
renormalizable in any dimension because of the extended nature of the Fermi-surface singularity.
Since the non-commutative electron propagator (i.e. in non zero magnetic field)
looks very similar to the Grosse-Wulkenhaar propagator (it is in fact a 
generalization of the Langmann-Szabo-Zarembo propagator)
we can conjecture that the renormalizable interaction corresponding to this propagator
should be given by a Moyal product. That's why we hope
that non-commutative field theory and a suitable generalization of the Grosse-Wulkenhaar RG
might be the correct framework for a 
microscopic {\it ab initio} understanding 
of the fractional quantum Hall effect which is currently lacking.

\item{\bf Charged Polymers in Magnetic Field}

Just like the heat kernel governs random motion,
the covariant Mahler kernel governs random motion of charged particles 
in presence of a magnetic field. Ordinary polymers can be studied as random walk
with a local self repelling or self avoiding interaction. They can be treated by 
QFT techniques using the $N=0$ component limit or the supersymmetry trick
to erase the unwanted vacuum graphs. Many results, such as various exact 
critical exponents in two dimensions, approximate ones in three dimensions,
and infrared asymptotic freedom in four dimensions have been computed
for self-avoiding polymers through renormalization group techniques.
In the same way we expect that charged polymers under magnetic field should be studied
through the new non commutative vulcanized RG. The relevant interactions
again should be of the Moyal rather than of the local type, and there is no reason that
the replica trick could not be extended in this context. Ordinary observables such as
$N$ point functions would be only translation \emph{covariant}, but translation invariant physical
observables such as density-density correlations
should be recovered out of gauge invariant observables. In this way it 
might be possible to deduce new scaling properties of 
these systems and their exact critical exponents through the generalizations of the
techniques used in the ordinary commutative case \cite{Duplantier}.

More generally we hope that the conformal invariant two dimensional theories, the RG flows between them
and the $c$ theorem of Zamolodchikov \cite{Zam} should have 
appropriate \emph{magnetic generalizations} which should involve vulcanized flows
and Moyal interactions.

\item{\bf Quark Confinement}

It is less clear that NCVQFT gauge theories might shed light on confinement,
but this is also possible.

Even for regular commutative field theory such as non-Abelian gauge theory,
the strong coupling or non-perturbative regimes may be studied fruitfully through
their non-commutative (i.e. non local) counterparts. 
This point of view is forcefully suggested in \cite{Seiberg1999vs}, where a mapping is proposed
between ordinary and non-commutative gauge fields which do not preserve the gauge groups
but preserve the gauge equivalent classes. 
Let us further remark that the effective physics of confinement should be
governed by a non-local interaction, as is the case in effective strings or bags models.
The great advantage of NCVQFT over the initial 
matrix model approach of 'tHooft \cite{thooft} is that in the latter the planar graphs dominate
because a gauge group $SU(N)$ with $N$ large is introduced in an \emph{ad hoc} way instead of
the physical  $SU(2)$ or $SU(3)$, whether in the former case, 
there is potentially a perfectly physical explanation 
for the planar limit, since it should just emerge naturally out of 
a renormalization group effect. We would like the
large $N$ \emph{matrix} limit in NCVQFT's to parallel the large 
$N$ \emph{vector} limit which allows to understand the formation of
Cooper pairs in supraconductivity \cite{FMRT1}. In that case $N$ is not arbitrary but
is roughly the number of effective quasi particles or sectors around the extended Fermi 
surface singularity at the superconducting transition temperature.
This number is automatically very large if this temperature is very low. 
This is why we called this phenomenon a \emph{dynamical} large $N$ \emph{vector} limit.
NCVQFTs provides us with the first clear example of a 
\emph{dynamical} large $N$ \emph{matrix} limit. We hope therefore that it should be ultimately useful 
to understand bound states in ordinary commutative non-Abelian gauge theories, hence quark confinement.

\item{\bf Quantum Gravity}

Although ordinary renormalizable QFTs seem more or less to have 
NCVQFT analogs on the Moyal space, there is no renormalizable commutative field theory
for spin 2 particles, so that the NCVQFTs alone should 
not allow quantization of gravity. However
quantum gravity might enter the picture of NCVQFTs at a later and more advanced stage.
Since quantum gravity appears in closed strings, it may have
something to do with doubling the ribbons of some NCQFT 
in an appropriate way. But because there is no reason
not to quantize the antisymmetric tensor $B$ which defines the non commutative geometry
as well as the symmetric one $g$ which defines the metric, we should clearly no longer 
limit ourselves to Moyal spaces. A first step towards a non-commutative
approach to quantum gravity along these lines should be to search for
the proper analog of vulcanization in more general non-commutative 
geometries. It might for instance describe physics in the vicinity of 
a charged rotating black hole generating a strong magnetic field. However 
we have to admit that any theory of quantum gravity will probably remain highly 
conjectural for many decades or even centuries...

\end{itemize}

We would like to conclude this introduction on a slightly mind-provocative question:
could non-commutativity be an attractive alternative to supersymmetry?

In the version of the standard model developped by Alain Connes and followers
\cite{CCM} there is some non commutative geometry but restricted to a very simple 
internal space. This model when fed with the spectral action
principle reproduces in astonishing detail all the standard model terms. Furthermore
it has some natural unification scale (without requiring a bigger non-Abelian gauge group 
and proton decay!). When prolonged  through \textit{ordinary}
commutative renormalization group flows on ordinary $\R^4$ from that 
unification scale back to the
Tev or Gev scales, it  postdicts within a few percent the top quark mass
and predicts the expected Higgs mass. Hence it seems a good starting point for 
understanding the standard model, just waiting for some additional fine tuning.

Now one of the strongest argument in favor of the existence of (still unobserved)
\emph{supersymmetry} is that it tames ultraviolet flows by adding
loops of superpartners to the ordinary loops. In particular a main argument for supersymmetry is
that it makes the three flows of the standard model 
$U(1)$, $SU(2)$ and $SU(3)$  couplings better
converge at a single unification scale (see \cite{BJLL} 
and references therein for a discussion of this subtle question).
The taming of loops by superpartners is also very important
to improve the ultraviolet behavior of supergravity and ultimately of superstrings.



But we have now a new way to tame ultraviolet flows, namely
non-commutativity of space-time! The mechanism which killed
the Landau ghost could become therefore a substitute for supersymmetry,
especially if superpartners are not found at the LHC.

If at some energy scale in the presumed ``desert" (that is somewhere
between  the Tev and the Planck scale) non-commutativity escapes
the internal space of A. Connes and invades  
ordinary space-time itself, it might manifest itself first in the form
of a tiny non-zero commutator between pairs of space time variables.
From that scale up  towards grand unification and Planck scale, we should presumably use 
the non-commutative scale decomposition and the non-commutative renormalization group 
reviewed below rather than the ordinary one. Although we don't know fully yet
how non-Abelian gauge theories will behave in this respect, it may provide the neessary
fine tuning of the Connes model. Just like for $\phi^4_4$, the flows 
should become milder and may grind to a halt. 

In short the lack of Landau ghosts in non-commutative field theory discussed below
means that non-commutative geometry might be an attractive alternative to supersymmetry
to tame ultraviolet flows without introducing new particles.

\medskip
\noindent{\bf Acknowledgments}
\medskip

I would like to warmly thank all the collaborators who contributed in various ways to the elaboration
of this material, in particular M. Disertori, R. Gurau, J. Magnen, A. Tanasa,
F. Vignes-Tourneret, J.C. Wallet and 
R. Wulkenhaar. Special thanks are due to F. Vignes-Tourneret since this review is largely 
based on our common recent review \cite{RivTour}, with introduction and
sections added on commutative renormalization, ghost hunting and the parametric representation.
I would like also to sincerely apologize to the many people whose work in this area would be worth of citation but has not been cited here: this is because of my lack of time or competence 
but not out of bad will.

\section{Commutative Renormalization, a Blitz Review}

This section is a summary of \cite{Riv2} which we include for self-containedness.

\subsection{Functional integral}

In QFT, particle number is not conserved. Cross sections in scattering experiments
contain the physical information of the theory. 
They are the matrix elements of the diffusion matrix $\cS$. 
Under suitable conditions they are expressed in terms of 
the Green functions $G_{N}$ 
of the theory through so-called ``reduction formulae" 

Green's functions are time ordered vacuum expectation values
of the field $\phi$, which is operator valued and acts on the Fock space:
\be  G_{N}(z_{1},...,z_{N})  =  <\psi_{0}, 
T[\phi(z_{1})...\phi (z_{N})]\psi_{0}> \,  .
\ee
Here $\psi_{0}$ is the vacuum state and the
$T$-product orders 
$\phi(z_{1})...\phi (z_{N})$ according to times.

Consider a Lagrangian field theory, and split 
the total Lagrangian as the sum of a free plus an
interacting piece, $\cL=  \cL_{0} + \cL_{int}$. 
The Gell-Mann-Low formula expresses
the Green functions as vacuum expectation values of a similar product
of free fields with an $e^{i\cL_{int}}$ insertion:
\be G_{N}(z_{1},...,z_{N})  = \frac{ <\psi_{0}, 
T\biggl[ \phi(z_{1})...\phi (z_{N})e^{i\int dx \cL_{int}(\phi(x))}\biggr]
 \psi_{0}>}     {<\psi_{0}, 
T(e^{i\int dx \cL_{int}(\phi(x))})\psi_{0}>    } .    
\ee

In the functional integral formalism proposed by Feynman \cite{FH},
the Gell-Mann-Low formula is 
replaced by a functional integral in terms of an (ill-defined) 
``integral over histories"
which is formally the product of Lebesgue measures over all space time. 
The corresponding formula is the Feynman-Kac formula:
\be G_{N}(z_{1},...,z_{N}) =
  = \frac { \int  \prod\limits_{j}\phi(z_{j}) e^{i\int \cL(\phi(x))
  dx} D\phi}  
{\int   e^{i\int \cL(\phi(x)) dx} D\phi }    .   \label{funct}
\ee

The integrand in (\ref{funct}) contains now 
the full Lagrangian $\cL=\cL_{0}+\cL_{int}$ instead of the interacting one. 
This is  interesting to expose symmetries of the theory which may not be separate 
symmetries of the free and interacting Lagrangians, for instance 
gauge symmetries. Perturbation theory
and the Feynman rules can still be derived as explained in the next subsection. 
But (\ref{funct}) is also well adapted to constrained quantization
and to the study of non-perturbative effects. 
Finally there is a deep analogy between the Feynman-Kac formula and the 
formula which expresses correlation functions in classical statistical mechanics. For instance, the 
correlation functions for a lattice Ising model are given by

\be \bigl<   \prod_{i=1}^{n} \sigma_{x_{i}} \bigr>
  = \frac { \sum\limits_{ \{ \sigma_{x} = \pm 1 \} } e^{- L(\sigma)}
 \prod\limits_{i}\sigma_{x_{i}}} 
{\sum\limits_{ \{ \sigma_{x} = \pm 1 \} } e^{- L(\sigma) }}   ,
\label{ising}    \ee
where $x$ labels the discrete sites of the lattice, the sum
is over configurations $\{ \sigma_{x} = \pm 1 \}$ which associate
a ``spin'' with value +1 or -1 to each such site and 
$L(\sigma)$ contains usually nearest neighbor interactions 
and possibly a magnetic field h: 
\be L(\sigma) = \sum_{x, y \ {\rm nearest\  neighbors}} J \sigma_{x} 
\sigma_{y} + \sum_{x} h \sigma_{x} .\ee

By analytically continuing (\ref{funct}) to imaginary time, or 
Euclidean space, it is possible to complete the analogy with (\ref{ising}), 
hence to establish a firm contact with statistical mechanics
\cite{Zinn,ItDr,Par}.

This idea also allows 
to give much better meaning to the  
path integral, at least for a free bosonic field. Indeed the
free Euclidean measure  can be defined easily as a Gaussian measure, because
in Euclidean space $L_{0}$ is a quadratic form of positive type\footnote{However 
the functional space
that supports this measure is not in general a space of smooth functions,
but rather of distributions. This was already true for functional integrals
such as those of Brownian motion, which are supported by continuous but
not differentiable paths. Therefore ``functional integrals'' in quantum
field theory should more appropriately be 
called ``distributional integrals''.}.

The Green functions continued to Euclidean points are called the Schwinger 
functions of the model, and are given by the Euclidean Feynman-Kac formula:
\be S_{N}(z_{1},...,z_{N}) =
Z^{-1} \int  \prod_{j=1}^{N} \phi(z_{j}) e^{-\int \cL (\phi(x)) dx} 
D\phi   
\ee
\be Z= \int   e^{-\int \cL (\phi(x)) dx} 
D\phi  .\ee

The simplest interacting field theory is the theory of a one component scalar 
bosonic field $\phi$ with quartic interaction $\lambda\phi^{4}$ ($\phi^{3}$ which is 
simpler is unstable).  In ${\R}^{d}$ it is called the $\phi^{4}_{d}$ 
model. For $d=2,3$ the model is 
superrenormalizable and has been built non perturbatively by constructive 
field theory. For $d=4$ it is just renormalizable, and it provides the simplest 
pedagogical introduction to perturbative renormalization theory. But 
because of the Landau ghost Landau ghost or triviality problem explained in subsection 
\ref{Landaughost}, the model presumably does not exist as a true interacting
theory at the non perturbative level. Its non commutative version should exist 
on the Moyal plane, see section \ref{Noghost}.

Formally the Schwinger functions of $\phi^{4}_{d}$ are the moments of the 
measure:
\be  d\nu  = \frac {1}{Z} e^{-\frac{\lambda}{4!}\int \phi^{4} -(m^{2} / 2)
\int \phi^{2} - (a/2) \int (\partial _{\mu } \phi \partial ^{\mu }\phi  )
} D\phi , \label{mesurenu}\ee
where
\begin{itemize}
\item $\lambda$ is the coupling constant, usually assumed positive or complex 
with positive real part; remark the convenient 1/4! factor to take into account
the symmetry of permutation of all fields at a local vertex. In the non commutative
version of the theory permutation symmetry becomes the more restricted cyclic symmetry and 
it is convenient to change the 1/4! factor to 1/4.

\item $m$ is the mass, which fixes an energy scale for the theory;

\item $a$ is the wave function constant. It can be set to 1 by a rescaling of the field.

\item $Z$ is a normalization factor which makes (\ref{mesurenu}) a probability 
measure;

\item $D\phi$ is a formal (mathematically ill-defined) product 
$\prod\limits_{x \in {\R}^{d} }d\phi(x)$
of Lebesgue measures at every point of ${\R}^{d}$.
\end{itemize}

The Gaussian part of the measure is
\be  d\mu(\phi)  =  \frac{1}{Z_0} e^{-(m^{2} / 2)
\int \phi^{2} - (a/2) \int (\partial _{\mu } \phi \partial ^{\mu }\phi  )
} D\phi . \label{mesuremu}
\ee
where $Z_0$ is again the normalization factor which makes
(\ref{mesuremu}) a probability  measure.

More precisely if we consider the translation invariant 
propagator $C(x,y) \equiv C(x-y)$ (with slight abuse of notation), whose 
Fourier transform is

\be  C(p) =  \frac{1}{(2\pi)^{d}}  \frac{1}{p^{2}+m^{2}} , \ee
we can use Minlos theorem and the general theory of Gaussian processes  
to define $d\mu(\phi)$ as the centered Gaussian measure on the
Schwartz space of tempered distributions $S'({\R}^{d})$
whose covariance is $C$. A Gaussian measure is uniquely defined 
by its moments, or the integral of polynomials of fields. 
Explicitly this integral is zero for a monomial of odd degree, and for $n=2p$
even it is equal to 

\be   \int  \phi  (x_1 ) ...\phi  (x_n ) d\mu(\phi) = \sum_{\ga} 
\prod_{\ell \in \ga} C(x_{i_\ell},x_{j_\ell}) ,
\ee
where the sum runs over all the $2p!! = (2p-1)(2p-3)...5.3.1$
pairings $\ga$ of the $2p$ arguments into $p$ disjoint
pairs $\ell= (i_\ell,j_\ell)$.

Note that since for $d\ge 2$, $C(p)$ is not integrable,
$C(x,y)$ must be understood as a distribution. It is therefore
convenient to also use regularized kernels, for instance
\be     C_{\kappa}(p) =  \frac{1}{(2\pi)^{d}}
\frac{e^{-\kappa (p^{2}+m^{2})}}{p^{2}+m^{2}} 
= \int_{\kappa}^{\infty} e^{-\alpha (p^{2}+m^{2})} d\alpha \ee
whose Fourier transform $  C_{\kappa}(x,y)$ is now a smooth function
and not a distribution:
\be     C_{\kappa}(x,y) = \int_{\kappa}^{\infty} e^{-\alpha m^{2}- (x-y)^{2}/4\alpha} 
\frac{d\alpha}{\alpha^{D/2}} \ee
$\alpha^{-D/2} e^{- (x-y)^{2}/4\alpha}$
is the \emph{heat kernel} and therefore this $\alpha$-representation has also an
interpretation in terms of Brownian motion:
\begin{equation}
C_{\kappa}(x,y) = \int_{\kappa}^{\infty} d\alpha \exp (-m^2 \alpha)\, P(x,y;\alpha)
\end{equation} 
where $P(x,y;\alpha)=(4 \pi \alpha)^{-d/2}
\exp (-{\vert x-y\vert}^2/4\alpha)$ is the Gaussian probability
distribution of a Brownian path going from $x$ to $y$ in time $\alpha$.

Such a regulator $\kappa$ is called
an ultraviolet cutoff, and we have (in the distribution sense)
$\lim_{\kappa \to 0} C_{\kappa}(x,y)= C(x,y)$. Remark that due to the
non zero $m^2$ mass term, the kernel $C_{\kappa}(x,y)$ decays
exponentially at large $\vert x-y\vert$ with rate $m$. 
For some constant $K$ and $d>2$ we have:
\be   \vert C_{\kappa}  (x,y) \vert \le K 
\kappa ^{1-d/2} e^{- m \vert x-y \vert } .  \label{decay}
\ee

It is a standard useful construction to build from the 
Schwinger functions the connected Schwinger functions, given by:

\be  C_{N} (z_{1},...,z_{N}) = \sum_{P_{1}\cup ... \cup P_k = \{1,...,N\}; 
\, P_{i} \cap P_j =0}    (-1)^{k+1}  (k-1)! \prod_{i=1}^{k}  S_{p_i} 
(z_{j_1},...,z_{j_{p_i}}) , \ee
where the sum is performed over all distinct partitions of $\{1,...,N\}$ 
into $k$ 
subsets $P_1,...,P_k$, $P_i$ being made of $p_i$ elements called 
$j_1,...,j_{p_i}$. For instance in the $\phi^4$ theory, where
all odd Schwinger functions vanish due to the unbroken $\phi \to -\phi$ symmetry,
the connected 4-point function is simply:
\bqa C_{4} (z_{1},...,z_{4})  &=&  
S_{4} (z_{1},...,z_{4}) - S_{2}(z_{1},z_{2})
S_{2}(z_{3},z_{4}) \no\\
&& \quad \quad - S_{2}(z_{1},z_{3})S_{2}(z_{2},z_{4}) 
-S_{2}(z_{1},z_{4})S_{2}(z_{2},z_{3})  . \eqa

\subsection{Feynman Rules}

The full interacting measure may now be defined as the multiplication of
the Gaussian measure $d\mu(\phi)$ by the interaction factor:

\be  d\nu  =  \frac{1}{Z} e^{-\frac{\lambda}{4!} \int \phi^{4}(x) dx} d\mu(\phi)
\label{mesurenui}\ee
and the Schwinger functions are the normalized moments of this measure:
\be  S_{N} (z_1,...,z_N )  = \int     
\phi(z_1)...\phi(z_N) d\nu (\phi) . \label{schwimom}
\ee
Expanding the exponential as a power series in the coupling constant $\lambda$, 
one obtains a formal expansion for the Schwinger functions:

\be  S_{N} (z_1,...,z_N )  = 
\frac{1}{Z}  \sum_{n=0}^{\infty} \frac{(-\lambda)^{n}}{n!} \int
\bigl[ \int \frac{\phi^{4}(x) dx }{ 4!} \bigr]^{n}  
\phi(z_1)...\phi(z_N)  d\mu (\phi) \label{schwi}
\ee
It is now possible to perform explicitly the functional integral of
the corresponding polynomial.
The result gives at any order $n$ a sum over $(4n+N-1)!!$
``Wick contractions schemes $\cW$",  i.e. ways of pairing together 
$4n+N$ fields into $2n+N/2$ pairs. 
At order $n$ the result of this perturbation scheme is therefore simply
the sum over all these schemes $\cW$ of the spatial integrals
over $x_1,...,x_n$ of the integrand
$\prod_{\ell \in \cW} C(x_{i_\ell},x_{j_\ell}) $ 
times the factor 
$\frac{1}{n!}(\frac{-\lambda}{4!})^{n}$. These integrals are then functions (in 
fact distributions) of the external positions $z_{1},...,z_N$.
But they may diverge 
either because they are integrals over all of ${\R}^{4}$ (no volume cutoff) 
or because of the singularities in the propagator $C$ at coinciding points.

Labeling the $n$ dummy integration variables in (\ref{schwi})
as $x_1,...,x_n$, we draw a line $\ell$ for each contraction of two fields.
Each position $x_1,...,x_n$ is then associated to a four-legged vertex
and each external source $z_i$ to a one-legged vertex, as shown in 
Figure \ref{graph1}.

\begin{figure}
\centerline{\includegraphics[width=12cm]{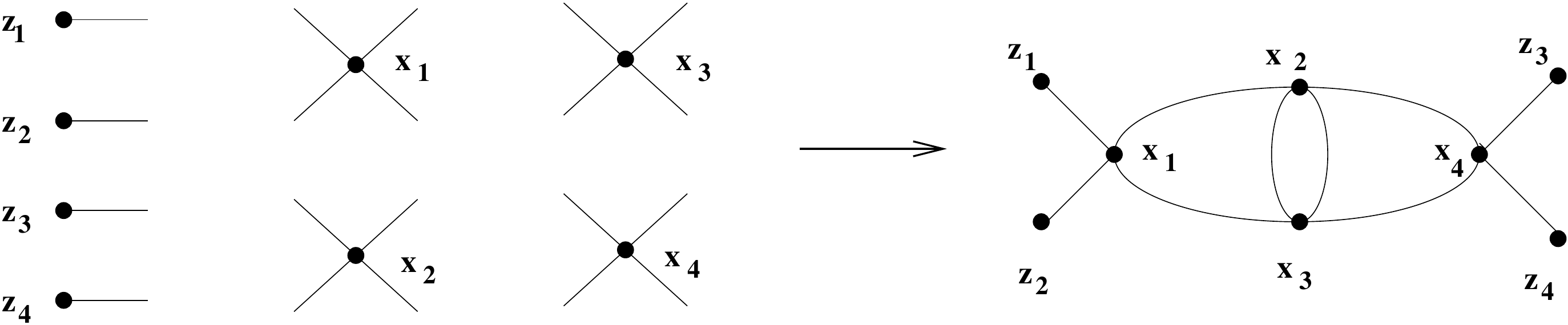}}
\caption{A possible contraction scheme with $n=N=4$.}
\label{graph1}
\end{figure}

For practical computations, it is obviously more convenient to gather all the 
contractions which lead to the same drawing, hence to the same 
integral. This leads to the notion of Feynman graphs.
To any such graph is associated a contribution or amplitude, 
which is the sum of the contributions associated with the corresponding set of 
Wick contractions. The ``Feynman rules" summarize how to compute this 
amplitude with its correct combinatoric factor.

We always use the following notations for a graph $G$:

\begin{itemize}
\item $n(G)$ or simply $n$ 
is the number of internal vertices of $G$, or the 
order of the graph.
\item $l(G)$ or $l$ is the number of internal lines of $G$, i.e. lines 
hooked at both ends to an internal vertex of $G$.
\item $N(G)$ or $N$ is the number of external vertices of $G$; 
it corresponds 
to the order of the Schwinger function one is looking at. 
When $N=0$ the graph 
is a vacuum graph, otherwise it is called an $N$-point graph.
\item $c(G)$ or $c$ is the number of connected components of $G$,
\item $L(G)$ or $L$ is the number of independent loops of G.
\end{itemize}

For a $regular$ $\phi^{4}$ graph, i.e. a graph which has 
no line hooked at both 
ends to external vertices, we have the relations:
\be     l(G) = 2n(G) - N(G)/2 , \label{externu} \ee
\be      L(G)  = l(G) - n(G) + c(G) =  n(G) +1 -N(G)/2 . \label{externu1} \ee
where in the last equality we assume connectedness of $G$, hence $c(G)=1$.

A $subgraph$ $F$ of a graph $G$ is a subset 
of internal lines of $G$, together
with the corresponding attached vertices. 
Lines in the subset defining $F$ are the internal lines 
of $F$, and their number is simply $l(F)$, as before. 
Similarly all the vertices of $G$ hooked to at least one of these 
internal lines of F are called the internal vertices of $F$ and considered 
to be in $F$; their number by definition is $n(F)$. 
Finally a good convention is to call external half-line of $F$
every half-line of $G$ which is not in $F$ but which is hooked to a vertex of 
$F$; it is then the number of such  external half-lines which we call $N(F)$. 
With these conventions one has for $\phi^{4}$ subgraphs the same 
relation (\ref{externu}) as for regular $\phi^{4}$ graphs.

To compute the amplitude associated to a $\phi^{4}$ graph, 
we have to add the contributions of the corresponding contraction schemes. 
This is summarized by the ``Feynman rules":
\begin{itemize}

\item To each line $\ell$ with end vertices at positions $x_\ell$ and $y_\ell$, 
associate a propagator $C(x_j,y_j)$.
\item To each internal vertex, associate $(-\lambda)/4!$.
\item Count all the contraction schemes giving this diagram. The number 
should be of the form $(4!)^n  n!/S(G)$ 
where $S(G) $ is an integer called the 
symmetry factor of the diagram.  
The $4!$ represents the permutation of the 
fields hooked to an internal vertex.
\item Multiply all these factors, divide by $n!$ and sum over the position 
of all internal vertices.
\end{itemize}

The  formula for the bare amplitude of a graph is therefore, as a distribution 
in $z_1,....z_N$:
\be  A_{G}(z_1,...,z_N) \equiv  \int \prod_{i=1}^{n}
 dx_i  \prod_{\ell \in G}  C(x_\ell,y_\ell) . \ee
This is the ``direct'' or ``$x$-space'' representation of a Feynman integral.
As stated above, this integral suffers of possible divergences. But 
the corresponding quantities with both volume cutoff and ultraviolet cutoff 
$\kappa$ are well defined. They are:

\be  A_{G,\Lambda}^{\kappa}(z_1,...,z_N) \equiv
\int_{\Lambda^{n}} \prod_{i=1}^{n}
dx_i  \prod_{\ell \in G}  C_{\kappa}(x_\ell,y_\ell) . \ee
The integrand is indeed bounded and the integration 
domain is a compact box $\Lambda$.

The $unnormalized$ Schwinger functions are therefore formally given 
by the sum over all graphs with the right number of external lines of the 
corresponding Feynman amplitudes:

\be ZS_{N} =   \sum_{\phi^{4}{\rm \ graphs \ } G {\rm\ with \ } N(G)=N} 
\frac{(-\lambda)^{n(G)}}{ S(G)} A_G
           . \label{series}\ee
$Z$ itself, the normalization, is given by the sum of all vacuum amplitudes:
\be Z =   \sum_{\phi^{4}{\rm \ graphs \ } G {\rm\ with \ } N(G)=0} 
\frac{(-\lambda)^{n(G)}}{S(G)} A_G. \ee

Let us remark that since the total number of Feynman graphs
is $(4n+N)!!$, taking into account Stirling's formula and the 
symmetry factor $1/n!$ from the exponential we expect perturbation
theory at large order to behave as $K^n n!$ for some constant $K$. 
Indeed at order $n$ 
the amplitude of a Feynman graph is a 4n-dimensional integral.
It is reasonable to expect that in average it should behave
as $c^n$ for some constant $c$. But this means that one should expect
zero radius of convergence for the series (\ref{series}).
This is not too surprising.
Even the one-dimensional integral 
\be F(g) = \int_{-\infty}^{+\infty}
e^{-x^2/2 - \lambda x^4 /4!} dx 
\ee 
is well-defined only for $\lambda \ge 0$.
We cannot hope infinite dimensional functional
integrals of the same kind to 
behave better than this one dimensional 
integral. In mathematically precise terms,
$F$ is not analytic near $\lambda=0$, but only Borel summable \cite{Sok}.
Borel summability
is therefore the best we can hope for the $\phi^4$ theory,
and it has indeed been proved
for the theory in dimensions 2 and 3 \cite{EMS,MS}.

From translation invariance, we do not expect $A_{G,\Lambda}^{\kappa}$ to have 
a limit as $\Lambda \to \infty$ if there are vacuum subgraphs in $G$.
But we can remark that an amplitude factorizes 
as the product of the amplitudes 
of its connected components.

With simple combinatoric verification 
at the level of contraction schemes we can 
factorize the sum over all vacuum graphs in the expansion of 
unnormalized Schwinger functions, hence get for the normalized functions
a formula analog to (\ref{series}):
\be 
S_N =  \sum_{\substack{\scriptstyle \phi^{4}{\rm  \ graphs \ } G {\rm \ with \ }
 N(G)=N\\
 \scriptstyle G {\rm\ without \ any \ vacuum \ subgraph }} }
\frac{(-\lambda)^{n(G)}}{S(G)} A_G . \label{normalfey} 
\ee
Now in (\ref{normalfey}) 
it is possible to pass to the thermodynamic limit (in the 
sense of formal power series) because using the exponential decrease of the 
propagator, each individual graph has a limit at fixed external arguments. 
There is of course no need to divide by the volume for that because each 
connected component in 
(\ref{normalfey})  is tied to at least one external source, and 
they provide the necessary breaking of translation invariance. 

Finally one can find the perturbative expansions for the 
connected Schwinger functions and 
the vertex functions. As expected, the connected Schwinger functions are given 
by sums over connected amplitudes:

\be  C_N =  \sum_{\phi^{4}{\rm \ connected \ graphs \ } G {\rm \  with \ } 
N(G)=N} \frac{(-\lambda)^{n(G)}}{ S(G)} A_{G} 
\ee
and the vertex functions are the sums of the $amputated$ amplitudes for 
proper graphs, also called one-particle-irreducible. 
They are the graphs which remain 
connected even after removal of any given internal line. 
The amputated amplitudes are defined in momentum space by 
omitting the Fourier transform of the propagators 
of the external lines. It is 
therefore convenient to write these amplitudes in the so-called momentum 
representation:

\be  \Gamma_N (z_1,...,z_N)=  \sum_{\phi^{4}{\rm \ proper \ graphs \ } 
G {\rm \  with \ } 
N(G)=N} \frac{(-\lambda)^{n(G)}}{S(G)} A_{G}^T (z_1,...,z_N) , \ee
\be A_{G}^T (z_1,...,z_N) \equiv  \frac{1}{(2\pi)^{dN/2}}
\int dp_1...dp_N  e^{i\sum p_iz_i}  A_G (p_1,...,p_N ) ,  \ee
\be  A_G (p_1,...,p_N ) = \int  \prod_{\ell \ {\rm internal \ line \ of\ }G} 
\frac{d^d p_\ell}{p_\ell^2 + m^2}
\prod_{v \in G} \delta ( \sum_\ell \epsilon_{v,\ell} \;  p_\ell  )  .
\label{momrep}\ee
Remark in (\ref{momrep}) the $\delta$ functions which 
ensure momentum conservation at 
each internal vertex $v$; the sum inside is over both internal and external 
momenta; each internal line is oriented in an arbitrary way and each 
external line is oriented towards 
the inside of the graph. The incidence 
matrix  
$\epsilon(v,\ell)$ is 1 if the line $\ell$ arrives at $v$, -1 if it starts from $v$ 
and 0 otherwise. Remark also that there is an overall momentum conservation 
rule $\delta(p_1 + ... + p_N)$ hidden in (\ref{momrep}). The drawback of 
the momentum representation lies in the necessity for practical 
computations to eliminate the $\delta$ functions by a ``momentum routing" 
prescription, and there is no canonical choice for that. Although this is 
rarely explicitly explained in the quantum field theory literature, such a choice of a 
momentum routing is equivalent to the choice of a particular spanning tree
of the graph.

\subsection{Scale Analysis and Renormalization}

In order to analyze the ultraviolet or short distance limit according to the renormalization group method,
we can cut the propagator $C$ into slices $C_i$ so that $C= \sum_{i=0}^{\infty}C_i$. 
This can be done conveniently within 
the parametric representation, since $\alpha $ in this representation roughly corresponds to
$1/p^2$. So we can define the propagator within a slice 
as 
\be C_{0} =  \int_{1}^{\infty} e^{-m^{2}\alpha -  
\frac{\vert x-y \vert^{2}}{ 4\alpha }} \frac{d\alpha}{\alpha^{d/2}}\ , \ \ 
C_{i} =  \int_{M^{-2i}}^{M^{-2(i-1)}} e^{-m^{2}\alpha -  
\frac{\vert x-y \vert^{2}}{ 4\alpha }}  \frac{d\alpha}{\alpha^{d/2}} \ \ {\rm for}\  i\ge 1 .
\label{slicing}
\ee
where $M$ is a fixed number, for instance 10, or 2, or $e$ (see footnote 
\ref{parochial} in the Introduction).
We can intuitively imagine $C_i$ as the piece of the field oscillating
with Fourier momenta essentially of size $M^{i}$.
In fact it is easy to prove the bound (for $d>2$)
\be 
\vert C_{i} (x,y) \vert \le K. M^{(d-2)i} 
e^{-  M^{i} \vert x-y\vert }
\label{boundslicing}
\ee
where $K$ is some constant.

Now the full propagator with ultraviolet
cutoff $M^{\rho}$, $\rho$ being a large integer, may be viewed as a sum
of slices:
\be C_{\le \rho} =  \sum_{i=0}^{\rho}  C_{i} \ee

Then the basic renormalization group step is made of two main operations:

\begin{itemize}
\item A functional integration 
\item The computation of a logarithm
\end{itemize}

Indeed decomposing a covariance in a Gaussian process
corresponds to a decomposition of the field into
independent Gaussian random variables $\phi^i$, each distributed with a
measure $d\mu_i$ of covariance $C_i$. 
Let us introduce
\be \Phi_{i} = \sum_{j=0}^{i} \phi_j .
\ee
This is the ``low-momentum" field for all frequencies
lower than $i$. The RG idea is that starting from
scale $\rho$ and performing $\rho -i$ steps,  
one arrives at an effective action for the remaining field 
$\Phi_{i}$. Then, writing $\Phi_i = \phi_i + \Phi_{i-1}$, one 
splits the field into a ``fluctuation" field
$\phi_i$ and a ``background" field $\Phi_{i-1}$. 
The first step, functional integration, is performed solely
on the fluctuation field, so it computes
\be Z_{i-1} (\Phi_{i-1}) = 
\int d\mu_i (\phi_i) e^{- S_i (\phi_i + \Phi_{i-1}) } .
\label{rengrou1}\ee
Then the second step rewrites this quantity as the exponential 
of an effective action, hence simply computes
\be S_{i-1}  (\Phi_{i-1}) = - \log [ Z_{i-1} (\Phi_{i-1}) ]
\label{rengrou2}\ee
Now $Z_{i-1} =e^{- S_{i-1}}$ and one can iterate!
The flow from the initial bare action $S=S_{\rho}$ for the full field to
an effective renormalized action $S_0$ for the last ``slowly varying" 
component $\phi_0$ of the field is similar to the
flow of a dynamical system. Its evolution
is decomposed into a sequence of discrete 
steps from $S_i$ to $S_{i-1}$.

This renormalization group strategy can be best understood on the system of Feynman graphs
which represent the perturbative expansion of the theory. The first step,
functional integration over fluctuation fields, means that we have to consider 
subgraphs with all their internal lines in higher slices than any of their external lines.
The second step, taking the logarithm, means that we have to consider only
\emph{connected} such subgraphs. We call such connected subgraphs \emph{quasi-local}.
Renormalizability is then a non trivial result that combines
locality and power counting for these quasi-local subgraphs. 

Locality  simply means that \emph{quasi-local} subgraphs $S$
look \emph{local} when seen through their external lines. Indeed since they are connected and since
their internal lines have scale say $\ge i$,
all the internal vertices are roughly at distance $M^{-i}$.
But the external lines have scales $\le i-1$, which only distinguish details larger than $M^{-(i-1)}$.
Therefore they  cannot distinguish the internal vertices of $S$ one from the other. 
Hence quasi-local subgraphs look like
``fat dots" when seen through their external lines, see Figure \ref{graph2}. 
Obviously this locality principle is completely independent of dimension.

\begin{figure}
\centerline{\includegraphics[scale=0.4]{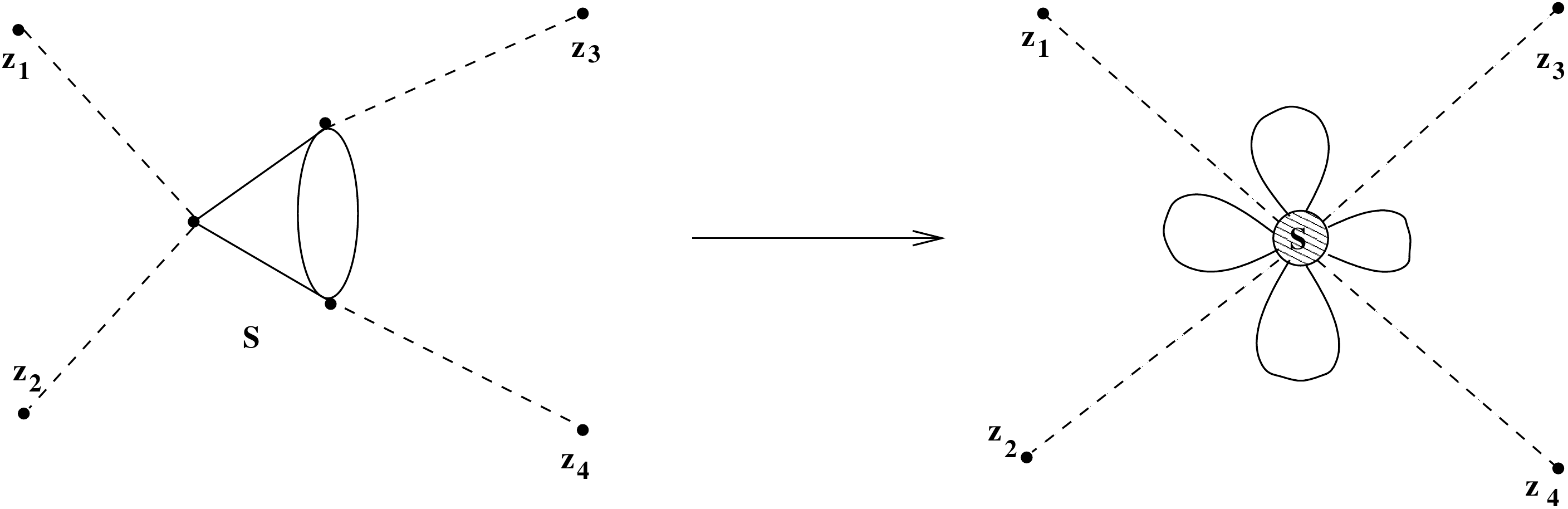}}
\caption{A high energy subgraph {\bf S} seen from lower energies looks quasi-local.}
\label{graph2}
\end{figure}

Power counting is a rough estimate which compares the size of a fat dot such as $S$ in Figure \ref{graph2} with $N$ external legs to
the  coupling constant that would be in front of an \emph{exactly local} $\int \phi^N (x) dx$ interaction term if it were in the Lagrangian. To simplify we now assume that the internal scales 
are all equal to $i$, the external scales are 
$O(1)$, and we do not care about constants and so on, but only about the dependence in
$i$ as $i$ gets large. 
We must first save one internal position such as the barycentre of the fat dot or the position of a particular internal vertex to represent the $\int dx$ integration in $\int \phi^N (x) dx$.
Then we must integrate
over the positions of all internal vertices of the subgraph \emph{save that one}. 
This brings about a weight $M^{-di(n-1)}$, because since $S$ is connected 
we can use the decay of the internal lines
to evaluate these $n-1$ integrals. Finally we should not forget the prefactor $M^{(D-2)li}$
coming from (\ref{boundslicing}), for the $l$ internal lines. Multiplying 
these two factors and using relation (\ref{externu})-(\ref{externu1})
we obtain that the "coupling constant" or factor in front of the fat dot
is of order $M^{-di(n-1) +2i(2n-N/2)}= M^{\omega (G)}$, if we define
the superficial degree of  divergence of a $\phi^{4}_{d}$ connected graph as:
\be   \omega (G) = (d-4)n(G) + d - \frac{d-2}{2} N(G) .  \ee
So power counting, in contrast with locality, depends on the space-time dimension.

Let us return to the concrete example of Figure \ref{graph2}. 
A 4-point subgraph made of three vertices and four internal lines at a high slice $i$
index. If we suppose the four external dashed lines have much lower index, say of order unity,
the subgraph looks almost local, like a fat dot at this unit scale. We have to save one vertex
integration for the position of the fat dot. Hence
the coupling constant of this fat dot is 
made of two vertex integrations and the four weights of the internal lines (in order not to forget these
internal line factors we kept internal lines apparent
as four tadpoles attached to the fat dot in the right of Figure \ref{graph2}).
In dimension 4 this total weight turns out to be independent of the scale.

At lower scales propagators can branch either through the initial bare coupling or through any
such fat dot in all possible ways because of the combinatorial rules of functional integration.
Hence they feel effectively a new coupling which is the sum of the bare 
coupling plus all the fat dot corrections coming from higher scales.
To compute these new couplings only graphs with $\omega (G) \ge 0$,
which are called primitively divergent, really matter
because their weight does not decrease as the gap $i$ increases.

- If $d=2$, we find $  \omega (G) = 2-2n  $, so the only primitively divergent
graphs have $n=1$, and $N=0$ or $N=2$. The only divergence
is due to the ``tadpole'' loop   $\int   \frac{d^2 p}{(p^2 + m^2)}$
which is logarithmically divergent.

- If $d=3$, we find $ \omega (G) = 3- n -N/2  $, so the only primitively divergent
graphs have $n\le 3 $, $N=0$, or $n \le 2$ and $N=2$. 
Such a theory with only a finite number of ``primitively divergent''
subgraphs is called superrenormalizable.

- If $d=4$, $\omega (G) = 4-N$.
Every two point graph is quadratically divergent
and every four point graph is logarithmically divergent.
This is in agreement with the superficial degree of 
these graphs being respectively 2 and 0. 
The couplings that do not decay with $i$ all correspond to terms that 
were already present in the Lagrangian, namely
$\int \phi^4$, $\int \phi^2$ and $\int (\nabla \phi ).(\nabla \phi )$\footnote{Because the graphs 
with $N=2$ are quadratically divergent we must Taylor expand the quasi local fat dots until we get 
convergent effects. Using parity and rotational symmetry, this generates only a
logarithmically divergent $\int (\nabla \phi ).(\nabla \phi )$ term beyond the 
quadratically divergent $\int \phi^2$. Furthermore this term starts only at $n=2$
or two loops, because the first tadpole graph at $N=2$, $n=1$ is \emph{exactly} local.}. 
Hence the structure of the Lagrangian resists under change of scale,
although the values of the coefficients can change.
The theory is called just renormalizable.

- Finally for $d >4$ we have infinitely many primitively divergent graphs
with arbitrarily large number of external legs, and the theory is called
non-renormalizable, because fat dots with $N$ larger than 4 are important and they 
correspond to new couplings generated by the renormalization group  which are not present
in the initial bare Lagrangian.

To summarize:

\begin{itemize}

\item Locality means that quasi-local subgraphs 
look local when seen through their external lines. It holds in any dimension.

\item  Power counting gives the rough size of the new couplings associated to these subgraphs
as a function of their number $N$ of external legs, of their order $n$
and of the dimension of space time $d$.

\item  Renormalizability (in the ultraviolet regime)
holds if the structure of the Lagrangian resists under change of scale,
although the values of the coefficients or coupling constants may change. 
For $\phi^4$ it occurs if $d\le 4$, with $d=4$ the most interesting case.

\end{itemize}

\subsection{The BPHZ Theorem}

The BPHZ theorem is both a brilliant historic
piece of mathematical physics which gives precise mathematical meaning to
the notion of renormalizability, using the mathematics
of formal power series, but it is also ultimately a bad way to understand and express 
renormalization. Let us try to explain both statements.

For the massive Euclidean $\phi_4^4$ theory
we could for instance state the following normalization 
conditions on the connected functions in momentum space at zero momenta:
\be C^{4} (0,0,0,0) = -\lambda_{ren} ,        \ee
\be C^{2} (p^{2}=0)  = \frac{1}{ m^{2}_{ren}} , \ee
\be  \frac{d}{dp^{2}} C^{2} \vert _{p^{2}=0}  
= -\frac{a_{ren}}{m^{4}_{ren}} . \ee
Usually one puts $a_{ren}=1$ by rescaling the field $\phi$.

Using the inversion theorem on formal power series for any \emph{fixed ultraviolet cutoff $\kappa$}
it is possible to reexpress any formal power series in $\lambda_{bare}$ with bare propagators
$1/(a_{bare}p^2 +m^2_{bare})$
for any Schwinger functions as a formal power series 
in $\lambda_{ren}$ with renormalized propagators
$1/(a_{ren}p^2 +m^2_{ren})$.
The BPHZ theorem then states that that formal perturbative formal power series  has finite coefficients order by order when the ultraviolet cutoff $\kappa$ is lifted. The first proof by Hepp
relied on the inductive Bogoliubov's recursion scheme. Then a completely explicit
expression for the coefficients of the renormalized series was written by
Zimmermann and many followers. The coefficients of that renormalized series
can be written as sums of renormalized Feynman amplitudes. They are similar
to Feynman integrals but with additional subtractions indexed by 
Zimmermann's forests. Returning to an inductive rather than explicit
scheme, Polchinski remarked that it is possible to also deduce the BPHZ theorem from 
a renormalization group equation and inductive bounds which does not decompose each order of 
perturbation theory into Feynman graphs \cite{Polch}. This method was
clarified and applied by C. Kopper and coworkers, see \cite{kopper}.

The solution of the difficult ``overlapping" divergence problem 
through Bogoliubov's or Polchinski's
recursions and Zimmermann's forests becomes particularly
clear in the parametric representation using Hepp's sectors. A Hepp sector is simply
a complete ordering of
the $\alpha$ parameters for all the lines of the graph. 
In each sector there is a different classification of forests
into packets so that each packet
gives a finite integral \cite{BLam}\cite{CR}.

But from the physical point of view we cannot conceal the fact
that purely perturbative renormalization theory is not very satisfying. 
At least two facts hint at 
a better theory which lies behind:

- The forest formula seems unnecessarily complicated,
with too many terms. 
For instance in any given 
Hepp sector only one particular packet of forests is 
really necessary to make the renormalized amplitude finite, the
one which corresponds to the quasi-local divergent subgraphs of \emph{that} sector. 
The other packets seem useless, a little bit like ``junk DNA". They are there just because they are necessary for other sectors. 
This does not look optimal.

- The theory makes renormalized amplitudes finite, but at tremendous cost!
The size of some of these renormalized amplitudes becomes 
unreasonably large as the size of the graph increases. 
This phenomenon is called the ``renormalon problem".
For instance it is easy to check that 
the renormalized amplitude (at 0 external momenta)
of the graphs $P_n$ with 6 external legs and $n+2$ internal 
vertices in Figure \ref{graph8} becomes as large as $c^n n!$ when $n \to \infty$.
Indeed at large $q$ the renormalized
amplitude $A_{G_2}^R$ in Figure \ref{oneloop} grows like $\log \vert
q\vert$. Therefore the chain of $n$ such graphs in Figure \ref{graph8} behaves as
$[\log \vert q\vert]^n$, and the total amplitude of $P_n$
behaves as 
\be  \int 
[\log \vert q \vert]^n  \frac{d^4 q}{[q^2 + m^2 ]^3} \simeq_{n \to\infty} 
c^n n!
\ee
So after renormalization some families of graphs acquire so large values
that they cannot be resumed! Physically this is just as bad as if infinities were still there.
\begin{figure}
\centerline{\includegraphics[width=10cm]{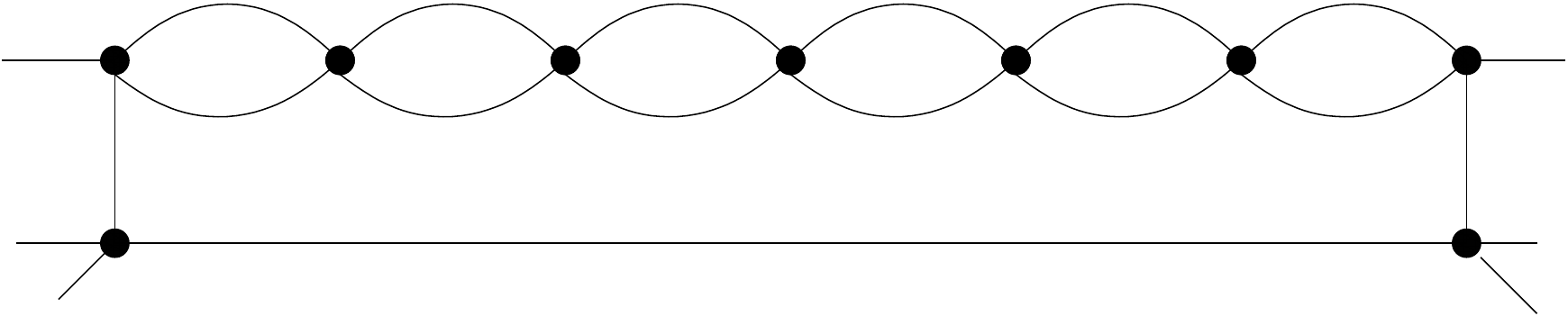}}
\caption{A family of graphs $P_n$ producing a renormalon.}
\label{graph8}
\end{figure}
These two hints are in fact linked. As their name
indicates, renormalons are due to
renormalization. Families of completely convergent graphs such as the graphs
$Q_n$ of Figure \ref{graph9}, are 
bounded by $c^n$, and produce no renormalons. 
\begin{figure}
\centerline{\includegraphics[width=10cm]{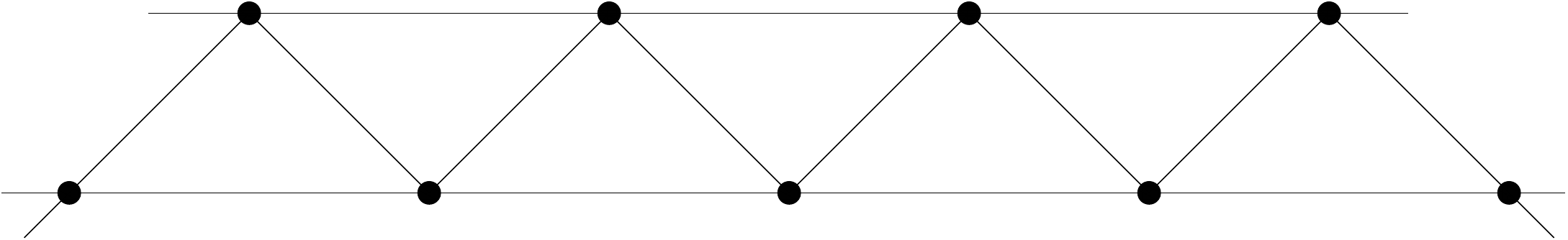}}
\caption{A family of convergent graphs $Q_n$, that do 
not produce any renormalon.}
\label{graph9}
\end{figure}

Studying more carefully
renormalization in the $\alpha$ parametric representation
one can check that renormalons are solely due
to the forests packets that we compared to ``junk DNA".
Renormalons are due to subtractions that are not necessary to ensure convergence, just like 
the strange $\log \vert q\vert$ growth of $A_{G_0}^R$ at large $q$
is solely due to the counterterm in the region where this counterterm
is not necessary to make the amplitude finite.

We can therefore conclude that subtractions are not organized
in an optimal way by the Bogoliubov recursion.
What is wrong 
from a physical point of view in the BPHZ theorem
is to use the size of the graph as the relevant parameter to
organize Bogoliubov's induction. It is rather
the size of the line momenta that should be used to
better organize the renormalization subtractions. 

This leads to the point of view advocated in \cite{Riv1}: neither the bare
nor the renormalized series are optimal. Perturbation should be organized
as a power series in an infinite set of effective expansions, which are
related through the RG flow equation. In the end exactly the same contributions are resumed
than in the bare or in the renormalized series, but they are regrouped in a much better
way.

\subsection{The Landau ghost and Asymptotic Freedom}
\label{Landaughost}

In the case of $\phi^4_4$ only the flow of the coupling constants
really matters, because the flow of $m$ and of $a$ for different reasons 
are not very important in the ultraviolet limit:

- the flow of $m$ is governed at leading order by the tadpole.
The bare mass $m^2_i$ corresponding to a finite positive physical mass $m^2_{ren}$
is negative and grows as $\lambda M^{2i}$ with the slice index $i$. But
since $p^2$ in the $i$-th slice
is also of order $M^{2i}$ but without the $\lambda$, as long as the coupling $\lambda$ 
remains small it remains much larger than $m^2_i$. Hence the mass term plays no significant role
in the higher slices. It was remarked in \cite{Riv1} that because  there are no
overlapping problem associated to 1PI two point subgraphs, there is in fact no inconvenience
to use the full renormalized $m_{ren}$ all the way from the bare to renormalized scales,
with subtractions on 1PI two point subgraphs independent of their scale.

- the flow of $a$ is also not very important. Indeed it really starts at two loops
because the tadpole is exactly local. So this flow is in fact bounded, and generates no renormalons. 
In fact as again remarked in \cite{Riv1} for theories of the $\phi^4_4$
type one might as well use the bare value $a_{bare}$ all the way 
from bare to renormalized scales and perform no second Taylor subtraction on any 
1PI two point subgraphs,.

But the physics of $\phi^4_4$ in the ultraviolet limit really depends of the flow
of $\lambda$. By a simple second order
computation there are only 2 connected graphs with $n=2$ and $N=4$
pictured in Figure \ref{oneloop}. They govern at leading order the flow of the coupling constant.

\begin{figure}
\centerline{\includegraphics[width=10cm]{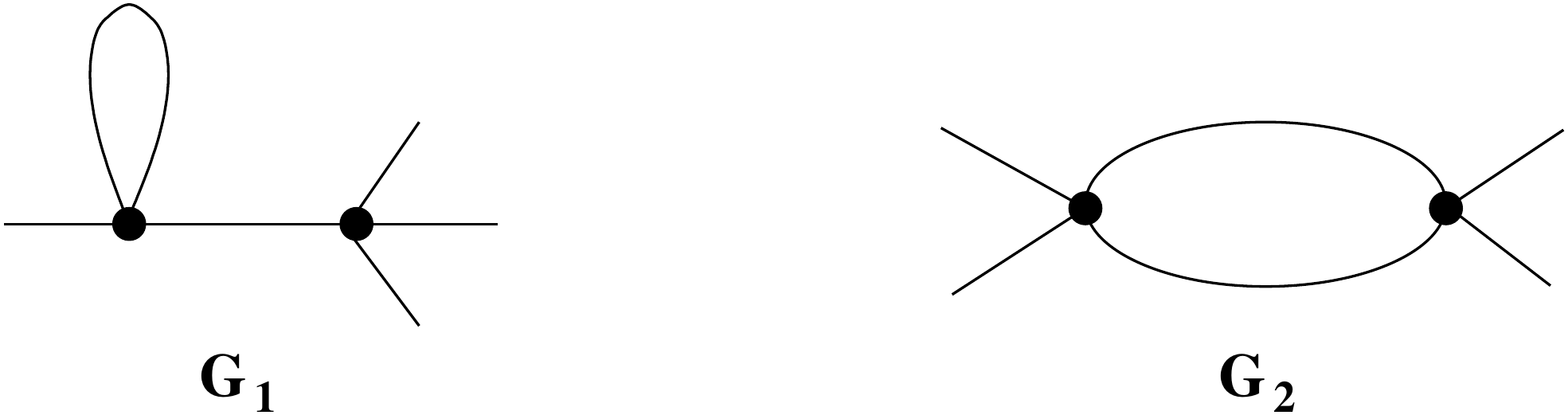}}
\caption{The $\phi^4$ connected graphs with $n=2$, $N=4$.}
\label{oneloop}
\end{figure}

In the commutative $\phi^4_4$ theory the graph $G_1$
does not contribute to the coupling constant flow. This 
can be seen in many ways, for instance after mass renormalization
the graph $G_1$ vanishes exactly because it contains a tadpole
which is not quasi-local but \emph{exactly} local. One can also remark
that the graph is one particle reducible. In ordinary translation-invariant,
hence momentum-conserving theories, one-particle-reducible
quasi-local graphs never contribute significantly to RG flows.
Indeed they become very small when the gap $i$ between internal 
and external scales grows. This is because
by momentum conservation the momentum of any one-particle-reducible line $\ell$ has to be the sum
of a finite set of external momenta on one of its sides.
But a finite sum of small momenta remains small and this clashes directly with the fact that
$\ell$ being internal its momentum should grow as the gap $i$ grows. 
Remark that this is no longer be true in non commutative vulcanized
$\phi^{\star 4}_4$, because that theory is not translation invariant, and that's
why it will ultimately escape the Landau ghost curse.

So in $\phi^4_4$ the flow is intimately linked to the sign
of the graph $G_2$ of Figure \ref{oneloop}. More precisely, we
find that at second order the relation between $\lambda_i$ and $\lambda_{i-1}$
is
\bqa \lambda_{i-1} &\simeq& \lambda_i - \beta \lambda_{i}^2 \label{flow1}
\eqa
(remember the minus sign in the exponential of the action),
where $\beta$ is a constant, namely the asymptotic value of
$\sum_{j, j' / \inf(j,j') =i} \int d^4y C_{j}(x,y)  C_{j'}(x,y) $
when $i \to \infty$. Clearly this constant is positive.
So for the normal stable $\phi_4^4$ theory, the relation
(\ref{flow1}) inverts into
\be \lambda_i \simeq \lambda_{i-1} + \beta  \lambda_{i-1}^2 ,
\ee 
so that fixing the renormalized
coupling seems to lead at finite $i$ to a large, diverging bare coupling,
incompatible with perturbation theory. This is the Landau ghost problem, which affects both the 
$\phi^4_4$ theory and electrodynamics.
Equivalently if one keeps $\lambda_i$ finite as $i$ gets large, $\lambda_0=\lambda_{ren}$
tends to zero and the final effective theory is ``trivial" which means it is a free
theory without interaction, in contradiction with the physical observation
e.g. of a coupling constant of about $1/137$ in electrodynamics.

But in non-Abelian gauge theories an extra
minus sign is created by the algebra of the
Lie brackets. This surprising discovery has deep 
consequences. The flow relation becomes approximately
\be \lambda_i \simeq \lambda_{i-1} - \beta \lambda_i \lambda_{i-1} ,
\ee 
with $\beta > 0 $, or, dividing by $\lambda_i \lambda_{i-1}$,
\be 1/\lambda_i \simeq 1/\lambda_{i-1} + \beta  ,
\ee 
with solution $\lambda_i \simeq \frac{\lambda_0}{1 + \lambda_0 \beta i}$. 
A more precise computation to third order in fact
leads to 
\be \lambda_i \simeq \frac{\lambda_0}{1 + \lambda_0 (\beta i + 
\gamma \log i + O(1)) }. \ee 
Such a theory is called asymptotically free
(in the ultraviolet limit) because the effective coupling
tends to 0 with the cutoff 
for a finite fixed small renormalized coupling. Physically
the interaction is turned off at small distances. This
theory is in agreement with scattering experiments 
which see a collection of almost free particles
(quarks and gluons) inside the hadrons at very high energy.
This was the main initial argument to adopt
quantum chromodynamics, a non-Abelian gauge
theory with $SU(3)$ gauge group, as the theory of strong interactions \cite{GWil2}. 

Remark that in such asymptotically free theories which form the backbone of today's standard model,
the running coupling constants remain bounded between far ultraviolet ``bare" scales 
and the lower energy scale where
renormalized couplings are measured. Ironically the point of view on early 
renormalization theory as a trick to hide the ultraviolet divergences of QFT
into infinite unobservable bare parameters could not turn out to be more wrong
than in the standard model. Indeed the bare coupling constants
tend to 0 with the ultraviolet cutoff, and what can be farther from infinity than 0?


\section{Non-commutative field theory}

\subsection{Field theory on Moyal space}
\label{sec:new-divergences}

The recent progresses concerning the renormalization of non-commutative field theory have been 
obtained on a very simple non-commutative space namely the Moyal space. From the point of view of 
quantum field theory, it is certainly the most studied space. Let us start with its precise definition.

\subsubsection{The Moyal space \texorpdfstring{${\mathbb R}^{D}_{\theta}$}{}}

Let us define $E=\lb x^{\mu},\,\mu\in\lnat 1,D\rnat\rb$ and $\C\langle E\rangle$ the free algebra 
generated by $E$. Let $\Theta$ a $D\times D$ non-degenerate skew-symmetric matrix (which requires 
$D$ even) and $I$ the ideal of $\C\langle E\rangle$ generated by the elements $x^{\mu}x^{\nu}-x^{\nu}
x^{\mu}-\imath\Theta^{\mu\nu}$. The Moyal algebra $\cA_{\Theta}$ is the quotient $\C\langle E\rangle/I$. 
Each element in $\cA_{\Theta}$ is a formal power series in the $x^{\mu}$'s for which the relation $\lsb x^
{\mu},x^{\nu}\rsb=\imath\Theta^{\mu\nu}$ holds.

Usually, one puts the matrix $\Theta$ into its canonical form :
\begin{eqnarray}
  \Theta= 
  \begin{pmatrix}
    \begin{matrix} 0 &\theta_{1} \\ 
      \hspace{-.5em} -\theta_{1}&0
    \end{matrix}    &&     (0)
    \\ 
    &\ddots&\\
    (0)&&
    \begin{matrix}0&\theta_{D/2}\\
      \hspace{-.5em}-\theta_{D/2}&0
    \end{matrix}
  \end{pmatrix}.\label{eq:Thetamatrixbase}
\end{eqnarray}
Sometimes one even set $\theta=\theta_{1}=\dotsm =\theta_{D/2}$. The preceeding algebraic definition 
whereas short and precise may be too abstract to perform real computations. One then needs a more 
analytical definition. A representation of the algebra $\cA_{\Theta}$ is given by some set of functions on 
$\R^{d}$ equipped with a non-commutative product: the \emph{Groenwald-Moyal} product. What follows 
is based on \cite{Gracia-Bondia1987kw}.

\paragraph{The Algebra $\cA_{\Theta}$}
\label{sec:lalgebre-ca_theta}

The Moyal algebra $\cA_{\Theta}$ is the linear space of smooth and rapidly decreasing functions $\cS
(\R^{D})$ equipped with the \encv{} product defined by: $\forall f,g\in\cS_{D}\defi\cS(\R^{D})$,
\begin{align}
  (f\star_{\Theta} g)(x)=&\int_{\R^D} \frac{d^{D}k}{(2\pi)^{D}}d^{D}y\, f(x+{\textstyle\frac 12}\Theta\cdot
  k)g(x+y)e^{\imath k\cdot y}\\
  =&\frac{1}{\pi^{D}\labs\det\Theta\rabs}\int_{\R^D} d^{D}yd^{D}z\,f(x+y)
  g(x+z)e^{-2\imath y\Theta^{-1}z}\; .
  \label{eq:moyal-def}
\end{align}
This algebra may be considered as the  ``functions on the Moyal space $\R^{D}_{\theta}$''. In the 
following we will write $f\star g$ instead of $f\star_{\Theta}g$ and use : $\forall f,g\in\cS_{D}$, $\forall j\in
\lnat 1,2N\rnat$,
\begin{align}
  (\scF f)(x)=&\int f(t)e^{-\imath tx}dt  
\end{align}
for the Fourier transform and
\begin{align}
  (f\diamond g)(x)=&\int f(x-t)g(t)e^{2\imath x\Theta^{-1}t}dt  
\end{align}
for the twisted convolution. As on $\R^{D}$, the Fourier transform exchange product and convolution:
\begin{align}
    \scF(f\star g)=&\scF(f)\diamond\scF(g)\label{eq:prodtoconv}\\
    \scF(f\diamond g)=&\scF(f)\star\scF(g)\label{eq:convtoprod}.
\end{align}
One also shows that the Moyal product and the twisted convolution are \textbf{associative}:
\begin{align}
  ((f\diamond g)\diamond h)(x)=&\int f(x-t-s)g(s)h(t)e^{2\imath(x\Theta^{-1}t+(x-t)\Theta^{-1}s)}ds\,dt\\
  =&\int f(u-v)g(v-t)h(t)e^{2\imath(x\Theta^{-1}v-t\Theta^{-1}v)}dt\,dv\notag\\
  =&(f\diamond(g\diamond h))(x).
\end{align}
Using \eqref{eq:convtoprod}, we show the associativity of the $\star$-product. The complex conjugation 
is \textbf{involutive} in $\cA_{\Theta}$
\begin{align}
  \overline{f\star_{\Theta}g}=&\bar{g}\star_{\Theta}\bar{f}.\label{eq:Moyal-involution}  
\end{align}
One also have
\begin{align}
  f\star_{\Theta}g=&g\star_{-\Theta}f.\label{eq:Moyal-commutation}  
\end{align}
\begin{prop}[Trace]\label{prop:trace}
  For all $f,g\in\cS_{D}$,
  \begin{align}
    \int dx\,(f\star g)(x)=&\int dx\,f(x)g(x)=\int dx\,(g\star f)(x)\; .\label{eq:Moyal-trace}
  \end{align}
\end{prop}
\begin{proof}
  \begin{align}
    \int dx\,(f\star g)(x)=&\scF(f\star g)(0)=(\scF f\diamond\scF g)(0)\\
    =&\int\scF f(-t)\scF g(t)dt=(\scF f\ast\scF g)(0)=\scF(fg)(0)\notag\\
    =&\int f(x)g(x)dx\notag
  \end{align}
where $\ast$ is the ordinary convolution.
\end{proof}

In the following sections, we will need lemma \ref{lem:Moyal-prods} to compute the interaction terms for 
the $\Phi^{\star 4}_{4}$ and Gross-Neveu models. We write $x\wed y\defi 2x\Theta^{-1}y$.
\begin{lemma}\label{lem:Moyal-prods}
  For all $j\in\lnat 1,2n+1\rnat$, let $f_{j}\in\cA_{\Theta}$. Then
  \begin{align}
    \lbt f_{1}\star_{\Theta}
    \dotsb\star_{\Theta}f_{2n}\rbt(x)=&\frac{1}{\pi^{2D}\det^{2}\Theta}\int\prod_{j=1}^{2n}
    dx_{j}f_{j}(x_{j})\,e^{-\imath
      x\wed\sum_{i=1}^{2n}(-1)^{i+1}x_{i}}\,e^{-\imath\varphi_{2n}},\\
    \lbt f_{1}
    \star_{\Theta}\dotsb\star_{\Theta}f_{2n+1}\rbt(x)=&\frac{1}{\pi^{D}\det\Theta}\int\prod_{j=1}^{2n+1}
    dx_{j}f_{j}(x_{j})\,\delta\Big(x-\sum_{i=1}^{2n+1}(-1)^{i+1}x_{i}\Big)\,e^{-\imath\varphi_{2n+1}},\\
  \forall p\in\N,\,\varphi_{p}=&\sum_{i<j=1}^{p}(-1)^{i+j+1}x_{i}\wed x_{j}.
  \end{align}
\end{lemma}
\begin{cor}\label{cor:int-Moyal}
  For all $j\in\lnat 1,2n+1\rnat$, let $f_{j}\in\cA_{\Theta}$. Then
  \begin{align}
    \int dx\,\lbt
    f_{1}\star_{\Theta}\dotsb\star_{\Theta}f_{2n}\rbt(x)=&\frac{1}{\pi^{D}\det\Theta}
    \int\prod_{j=1}^{2n}
    dx_{j}f_{j}(x_{j})\,\,\delta
    \Big(\sum_{i=1}^{2n}(-1)^{i+1}x_{i}\Big)e^{-\imath\varphi_{2n}},\label{eq:int-Moyal-even}\\
    \int dx\,\lbt f_{1}\star_{\Theta}
    \dotsb\star_{\Theta}f_{2n+1}\rbt(x)=&\frac{1}{\pi^{D}\det\Theta}\int\prod_{j=1}^{2n+1}
    dx_{j}f_{j}(x_{j})\,e^{-\imath\varphi_{2n+1}},\\
    \forall p\in\N,\,\varphi_{p}=&\sum_{i<j=1}^{p}(-1)^{i+j+1}x_{i}\wed x_{j}.
  \end{align}
\end{cor}

The cyclicity of the product, inherited from proposition \ref{prop:trace} implies: $\forall f,g,h\in\cS_{D}$,
\begin{align}
  \langle f\star g,h\rangle=&\langle f,g\star h\rangle=\langle g,h\star f\rangle
\end{align}
and allows to extend the Moyal algebra by duality into an algebra of tempered distributions.

\paragraph{Extension by Duality}
\label{sec:extens-par-dual}

Let us first consider the product of a tempered distribution with a Schwartz-class function. Let $T\in\cS'_
{D}$ and $h\in\cS_{D}$. We define $\langle T,h\rangle\defi T(h)$ and $\langle T^\ast,h\rangle =\overline
{\langle T,\overline{h}\rangle}$.
\begin{defn}\label{defn:Tf}
  Let $T\in\cS'_{D}$, $f,h\in\cS_{D}$, we define $T\star f$ and $f\star T$ by
  \begin{align}
    \langle T\star f,h\rangle=&\langle T,f\star h\rangle,\\
    \langle f\star T,h\rangle=&\langle T,h\star f\rangle.
  \end{align}
\end{defn}
For example, the identity $\bbbone$ as an element of $\cS'_{D}$ is the unity for the $\star$-product: $
\forall f,h\in\cS_{D}$,
\begin{align}
  \langle\bbbone\star f,h\rangle=&\langle\bbbone,f\star h\rangle\\
  =&\int(f\star h)(x)dx=\int f(x)h(x)dx\notag\\
  =&\langle f,h\rangle.\notag
\end{align}
We are now ready to define the linear space $\cM$ as the intersection of two sub-spaces $\cM_{L}$ and 
$\cM_{R}$ of $\cS'_{D}$.
\begin{defn}[Multipliers algebra]\label{defn:M}
  \begin{align}
    \cM_{L}=&\lb S\in\cS'_{D}\tqs\forall f\in\cS_{D},\,S\star f\in\cS_{D}\rb,\\
    \cM_{R}=&\lb R\in\cS'_{D}\tqs\forall f\in\cS_{D},\,f\star R\in\cS_{D}\rb,\\
    \cM=&\cM_{L}\cap\cM_{R}.
  \end{align}
\end{defn}
One can show that $\cM$ is an associative $\ast$-algebra. It contains, among others, the identity, the 
polynomials, the $\delta$ distribution and its derivatives. Then the relation
\begin{align}
  \lsb x^{\mu},x^{\nu}\rsb=&\imath\Theta^{\mu\nu}, 
\end{align}
often given as a definition of the Moyal space, holds in $\cM$ (but not in $\cA_{\Theta}$).

\subsubsection{ The \texorpdfstring{$\Phi^{\star 4}$}{phi4}-theory on \texorpdfstring{${\mathbb R}^{4}_{\theta}
$}{Moyal space}}

The simplest \encv{} model one may consider is the $\Phi^{\star 4}$-theory on the four-dimensional Moyal 
space. Its Lagrangian is the usual (commutative) one where the pointwise product is replaced by the 
Moyal one:
\begin{align}
  S[\phi] =&\int d^4x \Big( -\frac{1}{2} \partial_\mu \phi
\star \partial^\mu \phi  + \frac{1}{2} m^2
\,\phi \star \phi
+ \frac{\lambda}{4} \phi \star \phi \star \phi \star
\phi\Big)(x).\label{eq:phi4naif}
\end{align}
Thanks to the formula \eqref{eq:moyal-def}, this action can be explicitly computed. 
The interaction part is 
given by the corollary \ref{cor:int-Moyal}:
\begin{align}
  \int dx\, \phi^{\star
    4}(x)=&\int\prod_{i=1}^{4}dx_{i}\,\phi(x_{i})\,\delta(x_{1}-x_{2}+x_{3}-x_{4})e^{\imath\varphi},
  \label{eq:interaction-phi4}\\
  \varphi=&\sum_{i<j=1}^{4}(-1)^{i+j+1}x_{i}\wed x_{j}.\nonumber
\end{align}
The most obvious characteristic of the Moyal product is its non-locality. But its 
non-commutativity implies that the 
vertex of the model \eqref{eq:phi4naif} is only invariant under cyclic permutation of the fields. This 
restricted invariance incites to represent the associated Feynman graphs with ribbon propagators. One can 
then make a clear distinction between planar and non-planar graphs. This will be detailed in section \ref
{sec:multi-scale-analysisMatrix}.

Thanks to the delta function in \eqref{eq:interaction-phi4}, the oscillation may be written in different ways:
\begin{subequations}
  \begin{align}
    \delta(x_{1}-x_{2}+x_{3}-x_{4})e^{\imath\varphi}=&\delta(x_{1}-x_{2}+x_{3}-x_{4})e^{\imath
      x_{1}\wed x_{2}+\imath x_{3}\wed x_{4}}\\
    =&\delta(x_{1}-x_{2}+x_{3}-x_{4})e^{\imath
      x_{4}\wed x_{1}+\imath x_{2}\wed x_{3}}\\
    =&\delta(x_{1}-x_{2}+x_{3}-x_{4})\exp\imath(x_{1}-x_{2})\wed(x_{2}-x_{3}).\label{eq:oscill-trans}
  \end{align} 
\end{subequations}
The interaction is real and positive\footnote{Another way to prove it is from \eqref{eq:Moyal-involution}, $
\overline{\phi^{\star 4}}=\phi^{\star 4}$.}:
\begin{align}
  &\int\prod_{i=1}^{4}dx_{i}\phi(x_{i})\,\delta(x_{1}-x_{2}+x_{3}-x_{4})e^{\imath\varphi}\label{eq:int-
positive}\\
  =&\int dk\lbt\int dxdy\,\phi(x)\phi(y)e^{\imath k(x-y)+\imath x\wed y}\rbt^{\!\!2}\in\R_{+}.\notag
\end{align}
It is also translation invariant as shows equation \eqref{eq:oscill-trans}.

The property \ref{prop:trace} implies that the propagator is the usual one: $\hat{C}(p)=1/(p^{2}+m^{2})$.

\subsubsection{UV/IR mixing}
\label{sec:uvir-mixing}

In the article \cite{Filk1996dm}, Filk computed the Feynman rules 
corresponding to \eqref{eq:phi4naif}. He showed that 
the planar amplitudes equal the commutative ones whereas the non-planar ones give rise to 
oscillations coupling the internal and external legs. Hence contrary perhaps
to overoptimistic initial expectations, non commutative geometry alone does not eliminate
the ultraviolet divergences of QFT. Since there are infinitely many planar graphs with four external legs,
the model \eqref{eq:phi4naif} might at best be just renormalizable in the
ultraviolet regime, as ordinary $\phi^4_4$.

In fact it is not.
Minwalla, Van Raamsdonk 
and Seiberg  discovered that the model \eqref{eq:phi4naif} exhibits a new type of divergences making it non-renormalizable \cite{MiRaSe}.  A typical example is the non-planar tadpole:
\begin{align} 
\raisebox{-0.4\totalheight}{\includegraphics{oberfig-1.pdf}}&=\frac{\lambda}{12} 
\int \frac{d^4k}{(2\pi)^4} 
\frac{e^{ip_{\mu} k_{\nu}
      \Theta^{\mu\nu} }}{k^2 + m^2}\notag\\ 
  &=  \frac{\lambda}{48\pi^2}  \sqrt{\frac{m^2}{(\Theta p)^2}}  K_1(\sqrt{m^2
    (\Theta p)^2})\seq_{p\to 0}p^{-2}.
\end{align}
If $p\neq 0$, this amplitude is finite but, for small $p$, it diverges like $p^{-2}$. In other words, if we put 
an ultraviolet cut-off $\Lambda$ to the $k$-integral, the two 
limits $\Lambda\to\infty$ and $p\to 0$ do not 
commute. This is the UV/IR mixing phenomena. A chain of non-planar tadpoles, inserted in bigger 
graphs, makes divergent any function (with six points or more for example). But this divergence is not 
local and can't be absorbed in a mass redefinition. This is what makes the model non-renormalizable. 
We will see in sections \ref{sec:prop-et-renorm} and \ref{sec:direct-space} that the UV/IR mixing results 
in a coupling of the different scales of the theory. We will also note that we should 
distinguish different 
types of mixing.

The UV/IR mixing was studied by several groups. First, Chepelev and Roiban \cite{Chepelev2000hm} 
gave a power counting for different scalar models. They were able to identify the divergent graphs and to 
classify the divergences of the theories thanks to the topological data of the graphs. Then V.~Gayral \cite{Gayral2004cs} showed that UV/IR mixing is present on all isospectral deformations (they consist in 
curved generalizations of the Moyal space and of the \encv{} torus). For this, he considered a scalar 
model \eqref{eq:phi4naif} and discovered contributions to the effective action which diverge when the 
external momenta vanish. The UV/IR mixing is then a general characteristic of the \encv{} theories, at 
least on these deformations.

\subsection{The Grosse-Wulkenhaar breakthrough}

The situation remained unchanged until H.~Grosse and R.~Wulkenhaar discovered a way to define a 
renormalizable \encv{} model. We will detail their result in 
section \ref{sec:multi-scale-analysisMatrix} but 
the main message is the following. By adding an harmonic term to the Lagrangian \eqref{eq:phi4naif},
\begin{align}
    S[\phi] =&\int d^4x \Big( -\frac{1}{2} \partial_\mu \phi
\star \partial^\mu \phi +\frac{\Omega^2}{2} (\tilde{x}_\mu \phi )\star (\tilde{x}^\mu \phi ) + \frac{1}{2} m^2
\,\phi \star \phi
+ \frac{\lambda}{4} \phi \star \phi \star \phi \star
\phi\Big)(x)\label{eq:phi4Omega}
\end{align}
where $\xt=2\Theta^{-1} x$ and the metric is Euclidean, the model, in four dimensions, is renormalizable 
at all orders of perturbation \cite{c}. We will see in section \ref{sec:direct-space} 
that this additional term 
give rise to an infrared cut-off and allows to decouple the different scales of the theory. The new model 
(\ref{eq:phi4Omega}), which we call vulcanized $\Phi^{\star 4}_{4}$, does not exhibit any mixing. This result is very important because it opens the way towards other \encv{} field theories. Remember that we call \emph{vulcanization} the procedure consisting in adding a new term to a 
Lagrangian of a \encv{} theory in order to make it renormalizable, see footnote \ref{vulca}.\\

The propagator $C$ of this $\Phi^{4}$ theory is the kernel of the inverse operator $-\Delta+\Omega^{2}
\xt^{2}+m^{2}$. It is known as the Mehler kernel \cite{simon79funct,toolbox05}:
\begin{equation}
  \label{eq:Mehler}
  C(x,y)=\frac{\Omega^{2}}{\theta^{2}\pi^{2}}\int_{0}^{\infty}\frac{dt}{\sinh^{2}(2\Ot
  t)}\,e^{-\frac{\Ot}{2}\coth( 2\Ot t)(x-y)^{2}-\frac{\Ot}{2}\tanh(2\Ot t)(x+y)^{2}-m^{2}t}.
\end{equation}
Langmann and Szabo remarked that the quartic interaction with Moyal product is invariant under a 
duality transformation. It is a symmetry between momentum and direct space. The interaction part of the 
model \eqref{eq:phi4Omega} is (see equation \eqref{eq:int-Moyal-even})
\begin{align}
  S_{\text{int}}[\phi]=&\int d^{4}x\,\frac{\lambda}{4}(\phi\star\phi\star\phi\star\phi)(x)\\
  =&\int\prod_{a=1}^{4}d^{4}x_{a}\,\phi(x_{a})\,V(x_{1},x_{2},x_{3},x_{4})\label{eq:Vx}\\
  =&\int\prod_{a=1}^{4}\frac{d^{4}p_{a}}{(2\pi)^{4}}\,\hat{\phi}(p_{a})\,\hat{V}(p_{1},p_{2},p_{3},p_{4})\label
{eq:Vp}
  \intertext{with}
  V(x_{1},x_{2},x_{3},x_{4})=&\frac{\lambda}{4}\frac{1}{\pi^{4}\det\Theta}\delta(x_{1}-x_{2}+x_{3}-x_{4})
\cos(2(\Theta^{-1})_{\mu\nu}(x_{1}^{\mu}x_{2}^{\nu}+x_{3}^{\mu}x_{4}^{\nu}))\notag\\
  \hat{V}(p_{1},p_{2},p_{3},p_{4})=&\frac{\lambda}{4}(2\pi)^{4}\delta(p_{1}-p_{2}+p_{3}-p_{4})\cos(\frac 12
\Theta^{\mu\nu}(p_{1,\mu}p_{2,\nu}+p_{3,\mu}p_{4,\nu}))\notag
\end{align}
where we used a \emph{cyclic} Fourier transform: $\hat{\phi}(p_{a})=\int dx\,e^{(-1)^{a}\imath p_{a}x_{a}}
\phi(x_{a})$. The transformation
\begin{align}
  \hat{\phi}(p)\leftrightarrow\pi^{2}\sqrt{|\det\Theta|}\,\phi(x),&\qquad p_{\mu}\leftrightarrow\xt_{\mu}  
\end{align}
exchanges \eqref{eq:Vx} and \eqref{eq:Vp}. In addition, the free part of the model \eqref{eq:phi4naif} isn't 
covariant under this duality. The vulcanization adds a term to the Lagrangian which restores the 
symmetry. The theory \eqref{eq:phi4Omega} is then covariant under the Langmann-Szabo duality:
\begin{align}
  S[\phi;m,\lambda,\Omega]\mapsto&\Omega^{2}\,S[\phi;\frac{m}{\Omega},\frac{\lambda}{\Omega^{2}},
\frac{1}{\Omega}].  
\end{align}
By symmetry, the parameter $\Omega$ is confined in $\lsb 0,1\rsb$. Let us note that for $\Omega=1$, 
the model is invariant.

\paragraph{}
\label{parag:interpHarm}
The interpretation of that harmonic term is not yet clear. But the vulcanization procedure already allowed 
to prove the renormalizability of several other models on Moyal spaces such that $\Phi^{\star 4}_{2}$ \cite
{GrWu03-2}, $\phi^{3}_{2,4}$ \cite{Grosse2005ig,Grosse2006qv} and the LSZ models \cite
{Langmann2003if,Langmann2003cg,Langmann2002ai}. These last ones are of the type
\begin{align}
    S[\phi] =&\int d^nx \Big( \frac{1}{2} \bar{\phi}\star(-\partial_\mu+\xt_{\mu}+m)^{2}\phi
+ \frac{\lambda}{4} \bar{\phi} \star \phi \star \bar{\phi} \star\phi\Big)(x).\label{eq:LSZintro}
\end{align}
By comparison with \eqref{eq:phi4Omega}, one notes that here the additional term is formally equivalent 
to a fixed magnetic background. Therefore such a model is invariant under
magnetic translations which combine a translation and a phase shift on the field. 
This model is invariant under the above duality and is exactly soluble. 
Let us remark that the complex interaction in \eqref
{eq:LSZintro} makes the Langmann-Szabo duality more natural. It doesn't need a cyclic Fourier 
transform. The $\phi^{\star 3}$ model at $\Omega=1$ also 
exhibits a soluble structure \cite{Grosse2005ig,Grosse2006qv,GrStei}.

\subsection{The non-commutative Gross-Neveu model}
\label{sec:non-comm-gross}

Apart from the $\Phi^{\star 4}_{4}$, the modified Bosonic LSZ model \cite{xphi4-05}
and supersymmetric theories, we now know several renormalizable \encv{} field
theories. Nevertheless they either are super-renormalizable ($\Phi^{\star 4}_{2}$
\cite{GrWu03-2}) or (and) studied at a special point in the parameter space
where they are solvable ($\Phi^{\star 3}_{2},\Phi^{\star 3}_{4}$
\cite{Grosse2005ig,Grosse2006qv}, the LSZ models
\cite{Langmann2003if,Langmann2003cg,Langmann2002ai}). Although only
logarithmically divergent for parity reasons, the \encv{} Gross-Neveu model is a
just renormalizable quantum field theory as $\Phi^{\star 4}_{4}$. One of its main
interesting features is that it can be interpreted as a non-local
Fermionic field theory in a constant magnetic background. Then apart from
strengthening the ``vulcanization'' procedure to get renormalizable \encv{} field
theories, the Gross-Neveu model may also be useful for the study of the
quantum Hall effect. It is also a good first candidate for a constructive
study \cite{Riv1} of a \encv{} field theory as Fermionic models are usually
easier to construct. Moreover its commutative counterpart being asymptotically
free and exhibiting dynamical mass generation
\cite{Mitter1974cy,Gross1974jv,KMR}, a study of the physics of this model
would be interesting.\\

The \encv{} Gross-Neveu model ($\GN$) is a Fermionic quartically interacting quantum field theory on 
the Moyal plane $\R^{2}_{\theta}$. The skew-symmetric matrix $\Theta$ is
\begin{align}
  \Theta=&
  \begin{pmatrix}
    0&-\theta\\\theta&0
  \end{pmatrix}.
\end{align}
The action is
\begin{align}\label{eq:actfunctGN}
  S[\psib,\psi]=&\int
  dx\lbt\psib\lbt-\imath\slashed{\partial}+\Omega\xts+m+\mu\,\gamma_{5}\rbt\psi+V_{\text{o}}(\psib,\psi)
  +V_{\text{no}}(\psib,\psi)\rbt(x)
\end{align}
where $\xt=2\Theta^{-1}x$, $\gamma_{5}=\imath\gamma^{0}\gamma^{1}$ and $V=V_{\text{o}}+V_{\text
{no}}$ is the interaction part given hereafter. The $\mu$-term appears at two-loop order. We use a 
Euclidean metric and the Feynman convention $\slashed{a}=\gamma^{\mu}a_{\mu}$. The $\gamma^{0}
$ and $\gamma^{1}$ matrices form a two-dimensional representation of the Clifford algebra $\{\gamma^
{\mu},\gamma^{\nu}\}=-2\delta^{\mu\nu}$. Let us remark that the $\gamma^{\mu}$'s are then skew-
Hermitian: $\gamma^{\mu\dagger}=-\gamma^{\mu}$.

\paragraph{Propagator}
The propagator corresponding to the action \eqref{eq:actfunctGN} is given by the following lemma:
\begin{lemma}[Propagator \cite{toolbox05}]\label{xpropa1GN}
  The propagator of the Gross-Neveu model is
  \begin{align}
    C(x,y)=&\int d\mu_{C}(\psib,\psi)\,\psi(x)\psib(y)=\lbt-\imath\slashed{\partial}+\Omega\xts+m\rbt^{-1}
(x,y)\\
    =&\ \int_{0}^{\infty}dt\, C(t;x,y),\notag\\
    C(t;x,y)=&\ -\frac{\Omega}{\theta\pi}\frac{e^{-tm^{2}}}{\sinh(2\Ot t)}\,
    e^{-\frac{\Ot}{2}\coth(2\Ot t)(x-y)^{2}+\imath\Omega x\wed y}\\
    &\times\lb\imath\Ot\coth(2\Ot t)(\xs-\ys)+\Omega(\xts-\yts)-m\rb
    e^{-2\imath\Omega t\gamma\Theta^{-1}\gamma}\notag
  \end{align}
  with $\Ot=\frac{2\Omega}{\theta}$ et $x\wed y=2x\Theta^{-1}y$.\\
We also have $e^{-2\imath\Omega t\gamma\Theta^{-1}\gamma}=\cosh(2\Ot t)\mathds{1}_{2}-\imath\frac
{\theta}{2}\sinh(2\Ot
  t)\gamma\Theta^{-1}\gamma$.
\end{lemma}
If we want to study a $N$-\emph{color} model, we can consider a propagator diagonal in these color 
indices.

\paragraph{Interactions}
Concerning the interaction part $V$, recall that (see corollary \ref{cor:int-Moyal}) 
for any $ f_{1}$, $f_{2}$, $f_{3}$, $f_{4}$ in $\cA_{\Theta}$,
\begin{align}
  \int dx\,\lbt f_{1}\star f_{2}\star f_{3}\star
  f_{4}\rbt(x)=&\frac{1}{\pi^{2}\det\Theta}\int\prod_{j=1}^{4}dx_{j}f_{j}(x_{j})\,
  \delta(x_{1}-x_{2}+x_{3}-x_{4})e^{-\imath\varphi},\label{eq:interaction-GN}\\
  \varphi=&\sum_{i<j=1}^{4}(-1)^{i+j+1}x_{i}\wed x_{j}.
\end{align}
This product is non-local and only invariant under cyclic 
permutations of the fields. Then, contrary to the 
commutative Gross-Neveu model, for which there exits only one spinorial interaction, the $\GN$ model 
has, at least, six different interactions: the \emph{orientable} ones
\begin{subequations}\label{eq:int-orient}
  \begin{align}
    V_{\text{o}}=\phantom{+}&\frac{\lambda_{1}}{4}\int
    dx\,\lbt\psib\star\psi\star\psib\star\psi\rbt(x)\label{eq:int-o-1}\\
    +&\frac{\lambda_{2}}{4}\int
    dx\,\lbt\psib\star\gamma^{\mu}\psi\star\psib\star\gamma_{\mu}\psi\rbt(x)\label{eq:int-o-2}\\
    +&\frac{\lambda_{3}}{4}\int
    dx\,\lbt\psib\star\gamma_{5}\psi\star\psib\star\gamma_{5}\psi\rbt(x),&\label{eq:int-o-3}
  \end{align}
\end{subequations}
where $\psi$'s and $\psib$'s alternate and the \emph{non-orientable} ones
\begin{subequations}\label{eq:int-nonorient}
  \begin{align}
    V_{\text{no}}=\phantom{+}&\frac{\lambda_{4}}{4}\int
    dx\,\lbt\psi\star\psib\star\psib\star\psi\rbt(x)\label{eq:int-no-1}&\\
    +&\frac{\lambda_{5}}{4}\int
    dx\,\lbt\psi\star\gamma^{\mu}\psib\star\psib\star\gamma_{\mu}\psi\rbt(x)\label{eq:int-no-2}\\
    +&\frac{\lambda_{6}}{4}\int
    dx\,\lbt\psi\star\gamma_{5}\psib\star\psib\star\gamma_{5}\psi\rbt(x).\label{eq:int-no-3}
  \end{align}
\end{subequations}
All these interactions have the same $x$ kernel thanks to the equation \eqref{eq:interaction-GN}. The 
reason for which we call these interactions orientable or not will be clear in section \ref{sec:direct-space}.

\section{Multi-scale analysis in the matrix basis}
\label{sec:multi-scale-analysisMatrix}

The matrix basis is a basis for Schwartz-class functions. In this basis, the Moyal product becomes a 
simple matrix product. Each field is then represented by an infinite matrix \cite{Gracia-Bondia1987kw,GrWu03-2,vignes-tourneret06:PhD}.

\subsection{A dynamical matrix model}
\label{sec:phi4-matrixbase}

\subsubsection{From the direct space to the matrix basis}
\label{sec:de-lespace-direct}

In the matrix basis, the action~\eqref{eq:phi4Omega} takes the form:
\begin{align}
  S[\phi]=&(2\pi)^{D/2}\sqrt{\det\Theta}\Big(\frac 12\phi\Delta\phi+\frac{\lambda}{4}\Tr\phi^{4}\Big)\label
{eq:SPhi4matrix}
\end{align}
where $\phi=\phi_{mn},\,m,n\in\N^{D/2}$ and
\begin{align}
    \Delta_{mn,kl}=&\sum_{i=1}^{D/2}\Big(\mu_{0}^{2}+\frac{2}{\theta}(m_{i}
    +n_{i}+1)\Big)\delta_{ml}\delta_{nk} \ -\frac{2}{\theta}
    (1-\Omega^{2})\label{eq:formequadMatrixPhi4}\\
    &\hspace{-1.5cm}\Big(\sqrt{(m_{i}+1)(n_{i}+1)}\,
    \delta_{m_{i}+1,l_{i}}\delta_{n_{i}+1,k_{i}}+\sqrt{m_{i}n_{i}}\,\delta_{m_{i}-1,l_{i}}
    \delta_{n_{i}-1,k_{i}}\Big)\prod_{j\neq i}\delta_{m_{j}l_{j}}\delta_{n_{j}k_{j}}.\notag
\end{align}
The (four-dimensional) matrix $\Delta$ represents the quadratic part of the Lagragian. The first difficulty 
to study the matrix model \eqref{eq:SPhi4matrix} is the computation of its propagator $G$ defined as the 
inverse of $\Delta$ :
\begin{align}
  \sum_{r,s\in\N^{D/2}}\Delta_{mn;rs}G_{sr;kl}
  =\sum_{r,s\in\N^{D/2}}G_{mn;rs}\Delta_{sr;kl}=\delta_{ml}\delta_{nk}.
\end{align}

Fortunately, the action is invariant under $SO(2)^{D/2}$ thanks to the form \eqref{eq:Thetamatrixbase} of 
the $\Theta$ matrix. It implies a conservation law
\begin{align}
  \Delta_{mn,kl}=&0\iff m+k\neq n+l.\label{eq:conservationindices}
\end{align}
The result is \cite{c,GrWu03-2}
\begin{align}
  \label{eq:propaPhimatrix}
  G_{m, m+h; l + h, l} 
&= \frac{\theta}{8\Omega} \int_0^1 d\alpha\,  
\dfrac{(1-\alpha)^{\frac{\mu_0^2 \theta}{8 \Omega}+(\frac{D}{4}-1)}}{  
(1 + C\alpha )^{\frac{D}{2}}} \prod_{s=1}^{\frac{D}{2}} 
G^{(\alpha)}_{m^s, m^s+h^s; l^s + h^s, l^s},
\\
 G^{(\alpha)}_{m, m+h; l + h, l}
&= \left(\frac{\sqrt{1-\alpha}}{1+C \alpha} 
\right)^{m+l+h} \sum_{u=\max(0,-h)}^{\min(m,l)}
   {\mathcal A}(m,l,h,u)
\left( \frac{C \alpha (1+\Omega)}{\sqrt{1-\alpha}(1-\Omega)} 
\right)^{m+l-2u},\notag
\end{align}
where ${\mathcal A}(m,l,h,u)=\sqrt{\binom{m}{m-u}
\binom{m+h}{m-u}\binom{l}{l-u}\binom{l+h}{l-u}}$ and $C$ is a function in $\Omega$ : $C(\Omega)=
\frac{(1-\Omega)^2}{4\Omega}$. The main advantage of the matrix basis is that it simplifies the 
interaction part: $\phi^{\star 4}$ becomes $\Tr\phi^{4}$. But the propagator becomes very complicated.\\

Let us remark that the matrix model \eqref{eq:SPhi4matrix} is \emph{dynamical}: its quadratic part is not 
trivial. Usually, matrix models are \emph{local}.
\begin{defn}
A matrix model is called
{\bf local} if $G_{mn;kl}=G(m,n)\delta_{ml}\delta_{nk}$ and {\bf non-local} otherwise.
\end{defn}
In the matrix theories, the Feynman graphs are ribbon graphs. The propagator $G_{mn;kl}$ is then 
represented by the Figure \ref{fig:propamatrix}.
\begin{figure}[htb]
  \centering
  \includegraphics{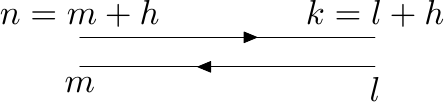}
  \caption{Matrix Propagator}
  \label{fig:propamatrix}
\end{figure}
In a local matrix model, the propagator preserves the index values 
along the trajectories (simple lines).

\subsubsection{Topology of ribbon graphs}
\label{sec:topologie-des-graphes}

The power counting of a matrix model depends on the topological data of its graphs. The figure \ref
{fig:ribbon-examples} gives two examples of ribbon graphs.
\begin{figure}[htbp]
  \centering 
  \subfloat[Planar]{{\label{fig:ribongraph1}}\includegraphics[scale=1]{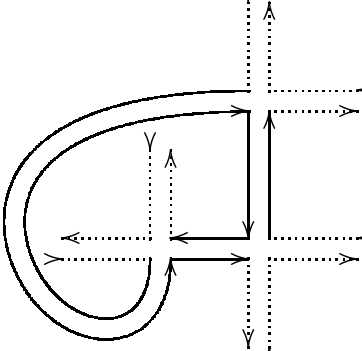}}\qquad
  \subfloat[Non-planar]{\label{fig:ribongraph2}\includegraphics[scale=1]{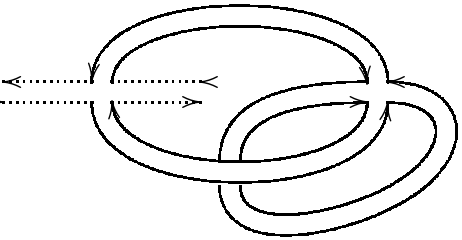}}
  \caption{Ribbon Graphs}
  \label{fig:ribbon-examples}
\end{figure}
Each ribbon graph may be drawn on a two-dimensional manifold. Actually each graph defines a
surface on which it is drawn. Let a graph $G$ with $V$ vertices, $I$ internal propagators (double lines) 
and $F$ faces (made of simple lines). The Euler characteristic
\begin{align}
  \chi=&2-2g=V-I+F\label{eq:Eulercar}
\end{align}
gives the genus $g$ of the manifold. One can make this clear by passing to the \textbf{dual graph}. The 
dual of a given graph $G$ is obtained by exchanging faces and vertices. The dual graphs of the 
$\Phi^{\star 4}$ 
theory are tesselations of the surfaces on which they are drawn. Moreover each direct face broken 
by external legs becomes, in the dual graph, a \textbf{puncture}. If among the $F$ faces of a graph, $B$ 
are broken, this graph may be drawn on a surface of genus $g=1-\frac 12(V-I+F)$ with $B$ punctures. 
The figure \ref{fig:topo-ruban} gives the topological data of the graphs of the figure \ref{fig:ribbon-examples}.
\begin{figure}[hbtp]
  \centering
  \begin{minipage}[c]{3cm}
    \centering
    \includegraphics[width=3cm]{gt1.pdf}
  \end{minipage}%
  \begin{minipage}[c]{2cm}
    \centering
    $ \Longrightarrow$
  \end{minipage}%
  \begin{minipage}[c]{2.6cm}
    \centering
    \includegraphics[width=2.6cm]{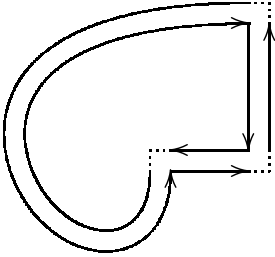}
  \end{minipage} \hspace{1cm}
   $\left .
     \begin{array}{c}
       $V=3$\\
       $I=3$\\
       $F=2$\\
       $B=2$
     \end{array}\rb
     \Longrightarrow\ g=0$\\
     \vspace{1cm}
     \begin{minipage}[c]{4cm}
       \centering
       \includegraphics[width=4cm]{gt3.pdf}
     \end{minipage}%
     \begin{minipage}[c]{2cm}
       \centering
       $\Longrightarrow$
     \end{minipage}%
     \begin{minipage}[c]{3cm}
       \centering
       \includegraphics[width=3cm]{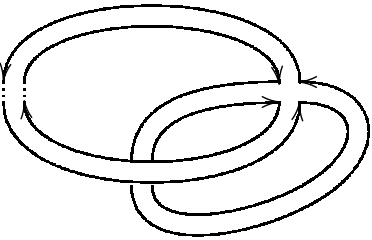}
     \end{minipage} \hspace{1cm}
     $\left .
       \begin{array}{c}
         $V=2$\\
         $I=3$\\
         $F=1$\\
         $B=1$
       \end{array}\rb
       \Longrightarrow\ g=1$
       \caption[Topological Data]{Topological Data of Ribbon Graphs}
       \label{fig:topo-ruban}
\end{figure}

\subsection{Multi-scale analysis}
\label{sec:analyse-multi-echell-matrix}

In \cite{Rivasseau2005bh}, a multi-scale analysis was introduced to complete the rigorous
study of the power counting of the \encv{} $\Phi^{\star 4}$ theory.

\subsubsection{Bounds on the propagator}
\label{sec:bornes-sur-le}

Let $G$ a ribbon graph of the $\Phi^{ \star 4}_{4}$ theory with $N$ external legs, $V$ vertices, $I$ internal lines 
and $F$ faces. Its genus is then $g=1 - \frac 12(V-I+F)$. Four indices $\{m,n;k,l\} \in\N^2$ are associated 
to each internal line of the graph and two indices to each external line, that is to say $4I+2N =8V$ 
indices. But, at each vertex, the left index of a ribbon equals the right one of the neighbor ribbon. This 
gives rise to $4V$ independent identifications which allows to write each index in terms of a set $
\mathcal{I}$ made of $4V$ indices, four per vertex, for example the left index of each half-ribbon.\\

The graph amplitude is then
\begin{align}  
  A_{G} = \sum_{\mathcal{I}} \prod_{\delta \in G}
  G_{m_{\delta}(\mathcal{I}),n_{\delta}(\mathcal{I});k_{\delta}(\mathcal{I}),l_{\delta}(\mathcal{I})}\;
    \delta_{m_{\delta}-l_{\delta},n_{\delta}-k_{\delta}} \;,
\label{IG}
\end{align}
where the four indices of the propagator $G$ of the line $\delta$ are function of $\mathcal{I}$ and written
\\
$\{m_{\delta}(\mathcal{I}),n_{\delta}(\mathcal{I}); k_{\delta}(\mathcal{I}),l_{\delta}(\mathcal{I})\} $. We 
decompose each propagator, given by \eqref{eq:propaPhimatrix}:
\begin{align}  
G = \sum_{i=0}^{\infty}G^i\qquad \text{thanks to }\int_{0}^{1}d\alpha=\sum_{i=1}^{\infty}\int_{M^{-2i}}^{M^
{-2(i-1)}}d\alpha,\;M>1.
\end{align}
We have an associated decomposition for each amplitude
\begin{align}  
A_G &= \sum_{\mu} A_{G,\mu}\;,
\\
A_{G,\mu} &= \sum_{\mathcal{I}} \prod_{\delta \in G} G^{i_{\delta}}_{
m_{\delta}(\mathcal{I}),n_{\delta}(\mathcal{I});
k_{\delta}(\mathcal{I}),l_{\delta}(\mathcal{I})}  \;
\delta_{m_{\delta}(\mathcal{I})-l_{\delta}(\mathcal{I}),
n_{\delta}(\mathcal{I})-k_{\delta}(\mathcal{I})} \;,
\label{IGmu}
\end{align}
where $\mu=\{i_{\delta}\}$ runs over the all possible assignments of a positive integer $i_{\delta}$ to 
each line $\delta$. We proved the following four propositions:
\begin{prop}
\label{thm-th1}
For $M$ large enough, there exists a constant $K$ such that, for $\Omega\in [0.5,1]$, we have the 
uniform bound
\begin{equation} 
    \label{th1}
    G^i_{m,m+h;l+h,l}\les 
    KM^{-2i} e^{-\frac{\Omega}{3}M^{-2i}\|m+l+h\|}.
  \end{equation}
\end{prop}
\begin{prop}
\label{thm-th2}
For $M$ large enough, there exists two constants $K$ and $K_{1}$ such that, for $\Omega \in [0.5,1]$, 
we have the uniform bound
\begin{align} 
&   G^i_{m,m+h;l+h,l}
\nonumber
\\*
& \les K M^{-2i} e^{-\frac{\Omega}{4}M^{-2i}\|m+l+h\|} 
\prod_{s=1}^{\frac{D}{2}} \min\left( 1, 
\left(
    \frac{K_{1}\min(m^s,l^s,m^s+h^s,l^s+h^s)}{M^{2i}}
\right)^{\!\!\frac{|m^s-l^s|}{2}} \right).
\label{th2}
\end{align}
\end{prop}
This bound allows to prove that the only diverging graphs have either a constant index along the 
trajectories or a total jump of $2$.
\begin{prop}\label{prop:bound3}
For $M$ large enough, there exists a constant $K$ such that, for $\Omega \in [\frac 23,1]$, we have the 
uniform bound
\begin{align}
\sum_{l =-m}^p G^i_{m,p-l,p,m+l} &\les 
K M^{-2i} \,e^{-\frac{\Omega}{4} M^{-2i} (\|p\|+\|m\|) }\;.
\label{thsum}
\end{align}
\end{prop}
This bound shows that the propagator is almost local in the following sense: with $m$ fixed, the sum 
over $l$ doesn't cost anything (see Figure \ref{fig:propamatrix}). Nevertheless the sums we'll have to 
perform are entangled (a given index may enter different propagators) so that we need the following 
proposition.
\begin{prop}\label{prop:bound4}
For $M$ large enough, there exists a constant $K$ such that, for $\Omega \in [\frac 23,1]$, we have the 
uniform bound
  \begin{equation} \label{thsummax}
   \sum_{l=-m}^\infty\max_{p\ges\max(l,0)}G^i_{m,p-l;p,m+l}
\les  KM^{-2i}e^{-\frac{\Omega}{36}M^{-2i}\|m\|} \;.   
  \end{equation}
\end{prop}
We refer to \cite{Rivasseau2005bh} for the proofs of these four propositions.

\subsubsection{Power counting}
\label{sec:vari-indep}

About half of the $4V$ indices initially associated to a graph is determined by the external indices and 
the delta functions in (\ref{IG}). The other indices are summation indices. The power counting consists in 
finding which sums cost $M^{2i}$ and which cost $\mathcal{O}(1)$ thanks to (\ref{thsum}). The $M^{2i}$ 
factor comes from (\ref{th1}) after a summation over an index\footnote{Recall that each index is in fact 
made of two indices, one for each symplectic pair of $\R^{4}_{\theta}$.} $m\in\N^2$,
\begin{align}
\sum_{m^1,m^2=0}^\infty e^{- c M^{-2i}(m^1+m^2)} = \frac{1}{(1-e^{- c
    M^{-2i}})^2} = \frac{M^{4i}}{c^2} (1 + \mathcal{O}(M^{-2i})).
\label{summ1m2}
\end{align}

We first use the delta functions as much as possible to reduce the set $\mathcal{I}$ to a true minimal set 
$\cI'$ of independent indices. For this, it is convenient to use the dual graphs where the resolution of the 
delta functions is equivalent to a usual momentum routing.

The dual graph is made of the same propagators than the direct graph except the position of their 
indices. Whereas in the original graph we have $G_{mn;kl}=\raisebox{-1ex}{\includegraphics[scale=.6]
{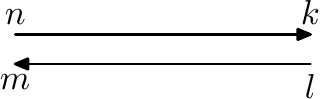}}$, the position of the indices in a dual propagator is
\begin{align}
G_{mn;kl} = \raisebox{-1ex}{\includegraphics[scale=.6]{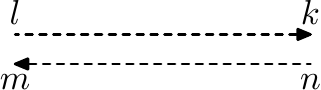}}\;.
\label{dual-assign}
\end{align}
The conservation $\delta_{l-m,-(n-k)}$ in (\ref{IG}) implies that the difference $l-m$ is conserved along 
the propagator. These differences behave like \emph{angular momenta} and the conservation of 
the differences $\ell=l-m$ and $-\ell=n-k$ is nothing else than the conservation of the angular 
momentum thanks to the symmetry $SO(2)\times SO(2)$ of the action \eqref{eq:SPhi4matrix}:
\begin{align}
\raisebox{-1.5ex}{\includegraphics[scale=1]{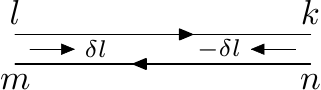}}&\qquad
l=m+\ell\;,~~ n=k+(-\ell).
\label{angmom}
\end{align}
The cyclicity of the vertices implies the vanishing of the sum of the angular momenta entering a vertex. 
Thus the angular momentum in the dual graph behaves exactly like the usual momentum in ordinary 
Feynman graphs.\\

We know that the number of independent momenta is exactly the number $L'$ ($=I-V'+1$ for a 
connected graph) of loops in the dual graph. Each index at a (dual) vertex is then given by a unique 
\emph{reference index} and a sum of momenta. If the dual vertex under consideration is an external one, 
we choose an external index for the reference index. The reference indices in the dual graph 
correspond to the loop indices in the direct graph. The number of summation indices is then $V'-B + L' = I 
+ (1-B)$ where $B\ges 0$ is the number of broken faces of the direct graph or the number of external 
vertices in the dual graph.\\

By using a well-chosen order on the lines, an optimized tree and a $L^{1}-L^{\infty}$ bound, one can 
prove that the summation over the angular momenta does not cost anything thanks to (\ref{thsum}). 
Recall that a connected component is a subgraph for which all internal lines have indices greater than 
all its external ones. The power counting is then:
\begin{align}
  A_G\les& K'^{V}\sum_{\mu}\prod_{i,k}M^{\omega(G^{i}_{k})}\\
  \text{with }\omega(G^{i}_{k})=&4(V'_{i,k}-B_{i,k})-2I_{i,k}=4(F_{i,k}-B_{i,k})-2I_{i,k}\\
  =&(4-N_{i,k})-4(2g_{i,k}+B_{i,k}-1)\notag
\end{align}
where $N_{i,k}$, $V_{i,k}$, $I_{i,k}=2V_{i,k}-\frac{N_{i,k}}{2}$, $F_{i,k}$ and $B_{i,k}$ are respectively the 
numbers of external legs, of vertices, of (internal) propagators, of faces and broken faces of the 
connected component $G^{i}_{k}$ ; $g_{i,k}= 1 - \frac{1}{2} (V_{i,k}-I_{i,k}+F_{i,k})$ is its genus. We have
\begin{thm}
\label{pc-slice}
The sum over the scales attributions $\mu$ converges if $\forall i,k,\, \omega(G^{i}_{k}) <0$.
\end{thm}
We recover the power counting obtained in \cite{GrWu03-1}.

From this point on, renormalizability of $\Phi^{\star 4}_4$ can proceed (however remark that it remains 
limited to $\Omega\in [0.5,1]$ by the technical estimates such as (\ref{th1}); this limitation
is overcome in the direct space method below).

The multiscale analysis allows to define the so-called effective expansion,
in between the bare and the renormalized expansion, which is optimal, 
both for physical and for constructive purposes \cite{Riv1}.
In this effective expansion only the subcontributions with all \textit{internal} scales higher 
than all \textit{external} scales have to be renormalized by counterterms 
of the form of the initial Lagrangian.

In fact only planar such subcontributions with a single external face
must be renormalized by such counterterms. This follows
simply from the Grosse-Wulkenhaar moves defined in \cite{GrWu03-1}. 
These moves 
translate the external legs along the outer border of the planar graph, up to irrelevant corrections, 
until they all merge together into a term of the proper Moyal form, which is then absorbed in the 
effective constants definition. 
This requires only the estimates (\ref{th1})-(\ref{thsummax}), which were checked 
numerically in \cite{GrWu03-1}.

In this way the relevant and marginal counterterms can be shown to be of the Moyal type,
namely renormalize the parameters $\lambda$, $m$ and $\Omega$\footnote{The wave function 
renormalization
i.e. renormalization of the $\partial_\mu\phi\star \partial^\mu\phi$ term can be absorbed in a rescaling
of the field, called ``field strength renormalization.''}.

Notice that in the multiscale analysis there is no need for the relatively complicated use of 
Polchinski's equation \cite{Polch} made in \cite{GrWu03-1}. Polchinski's method, 
although undoubtedly very elegant for proving perturbative renormalizability
does not seem directly suited to constructive purposes, even in the case of simple 
Fermionic models such as the commutative Gross Neveu model, see e.g. \cite{DiRi}.

The BPHZ theorem itself for the renormalized expansion follows from finiteness
of the effective expansion by developing the
counterterms still hidden in the effective couplings. Its own finiteness can be checked 
e.g. through the standard classification of forests \cite{Riv1}. Let us however
recall once again that in our opinion the effective expansion, not the renormalized one
is the more fundamental object, both to describe the physics and to attack deeper mathematical 
problems, such as those of constructive theory \cite{Riv1,Riv2}.

\section{Hunting the Landau Ghost}

\label{Noghost}

The matrix base simplifies very much at $\Omega =1$, where the matrix propagator becomes diagonal,
i.e. conserves exactly indices. This property has been used for the general proof that the
beta function of the theory vanishes in the ultraviolet regime 
\cite{beta2-06}. At the moment this is the only
concrete result that shows that NCVQFT is definitely \emph{better} behaved than QFT. It 
also opens the perspective of a full non-perturbative construction of the model.

We summarize now the sequence of three papers 
\cite{GrWu04-2}-\cite{DisertoriRivasseau2006}-\cite{beta2-06} 
which lead to this exciting result,
using the simpler notations of \cite{beta2-06}.

\subsection{One Loop}

The propagator in the matrix base at $\Omega=1$ is
\be \label{propafixed}
C_{mn;kl} = G_{mn} \delta_{ml}\delta_{nk} \ ; \ 
G_{mn}= \frac{1}{A+m+n}\  ,
\ee
where $A= 2+ \mu^2 /4$, $m,n\in \mathbb{N}^2$ ($\mu$ being the mass)
and we use the notations
\be
\de_{ml} = \de_{m_1l_1} \de_{m_2l_2}\ , \qquad m+m = m_1 + m_2 + n_1 + n_2 \ .
\ee

We focus on the complex $\bar\phi \star \phi \star \bar\phi \star \phi $ theory, since the result
for the real case is similar \cite{DisertoriRivasseau2006}.

The generating functional is:
\bea
&&Z(\eta,\bar{\eta})=\int d\phi d\bar{\phi}~e^{-S(\bar{\phi},\phi)+F(\bar{\eta},\eta,;\bar{\phi},\phi)}\nonumber\\
&&F(\bar{\eta},\eta;\bar{\phi},\phi)=  \bar{\phi}\eta+\bar{\eta}\phi \nonumber\\
&&S(\bar{\phi},\phi)=\bar{\phi}X\phi+\phi X\bar\phi+A\bar{\phi}\phi+
\frac{\lambda}{2}\phi\bar{\phi}\phi\bar{\phi}
\eea
where traces are implicit and the matrix $X_{m n}$ stands for $m\delta_{m n}$. $S$ is the action and $F$ the external sources. 

We denote $\Gamma^4(0,0,0,0)$ the amputated one particle irreducible four point function  
and $\Sigma(0,0)$ the amputated one particle irreducible two point function 
with external indices set to zero. The wave function renormalization is
$\partial_L \Sigma = \partial_R \Sigma = \Sigma (1,0) - \Sigma (0,0)$ \cite{DisertoriRivasseau2006},
and the corresponding field strength renormalization is 
$Z=(1-\partial_{L}\Sigma(0,0))= (1-\partial_{R}\Sigma(0,0))$
The main result to prove is that \emph{after field strength renormalization}\footnote{We recall
that in the ordinary commutative $\phi^4_4$ field theory there is no one loop
wave-function renormalization, hence the Landau ghost can be seen
directly on the four point function renormalization at one loop.}
the effective coupling is asymptotically constant, hence:

\begin{thm}\label{killghost}
The equation:
\bea\label{beautiful}
\Gamma^{4}(0,0,0,0)=\lambda Z^2
\eea
holds up to irrelevant terms to {\bf all} orders of perturbation, 
either as a bare equation with fixed ultraviolet cutoff,
or as an equation for the renormalized theory. In the latter case $\lambda $ should still be understood 
as the bare constant, but reexpressed as a series in powers of $\lambda_{ren}$.
\end{thm}

The field strength renormalization at one loop is
\be  
Z = 1 - a   \lambda  
\ee
where we can keep in $a$ only the coefficient of the logarithmic divergence, as the rest
does not contribute but to finite irrelevant corrections.

\begin{figure}[!htbp]
\label{tadpoleupdown}
\centerline{\includegraphics[width=8cm]{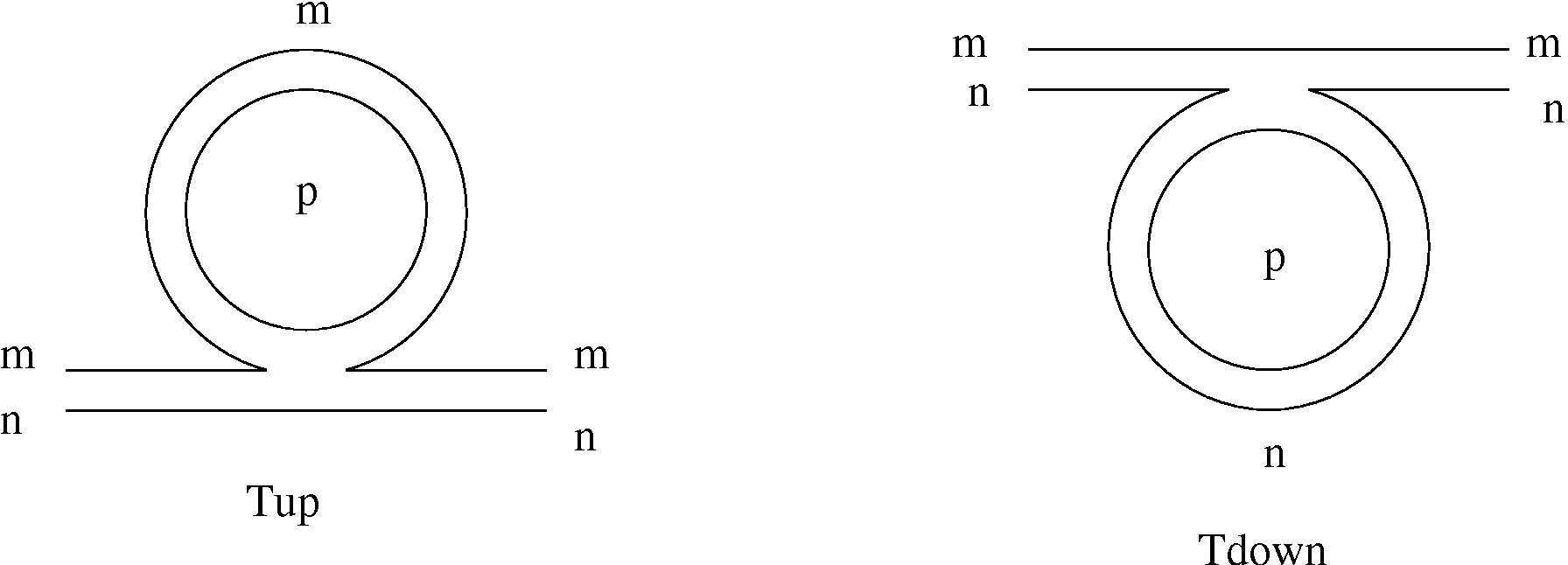}}
\caption{Two Point Graphs at one Loops: the up and down tadpoles}
\label{2p1tad}
\end{figure}

To compute $a$ we should add the wave function renormalization for the
two tadpoles  $T\ up$ and $T\ down$ of Figure \ref{2p1tad}. These two graphs
have both a coupling constant $-\lambda/2$, and a combinatorial factor 2
for choosing to which leg of the vertex the external $\bar \phi $ contracts. 
Then the logarithmic divergence of $T\ up$ is 
\be
\sum_{p} (\frac{1}{m+p +A} - \frac{1}{p +A}) = - \sum_{p} [\frac{m}{(m+p +A)(p +A)}]
\ee
so it corresponds to the renormalization of 
the coefficient of the $m$ factor in $G_{m,n}$ in \ref{propafixed}, with
logarithmic divergence $\lambda \sum_{p} [\frac{1}{p^2}] $. Similarly
the logarithmic divergence of $T\ down$ gives the same renormalization but for  
the $n$ factor in $G_{m,n}$ in \ref{propafixed}.

Altogether we find therefore that
\be
a = + \sum_{p} [\frac{1}{p^2}] 
\ee

In the real case we have a combinatoric factor $4$ instead of 2,
but the coupling constant is $\lambda/4$, so $a$ is the same. 

\begin{figure}[!htbp]\label{fouroneloop}
\centerline{\includegraphics[width=2cm,angle=-90]{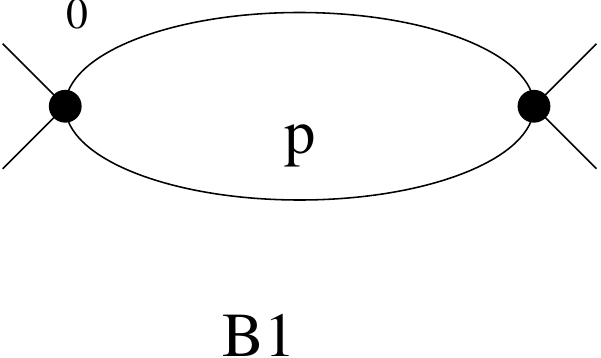}}
\caption{Four Point Graph at one Loop}

\end{figure}
The four point function perturbative expansion at one loop is
\be  
\Gamma_4 (0,0,0,0)  = - \lambda  
[1 -a' \lambda  ] .
\ee
Only the graph B1 of Figure \ref{fouroneloop}) contributes to $a'$. It has a prefactor
$\frac{1}{2!} (\lambda/2)^2$ and a combinatoric factor $2^4$ for contractions,
since there is a factor 2 to choose whether the bubble is "vertical or horizontal"
ie if the horizontal bubble of Figure \ref{fouroneloop}) is of $\bar\phi\star  \phi\   \star \bar\phi\star \phi  $
or of $\bar\phi\star  \phi\   \star \bar\phi\star \phi  $ type, then a factor 2 to choose to which vertex the first external; $\bar\phi$ contracts, then a factor 2 for the leg to which it contracts
in that vertex and finally another factor 2 for the leg to which the other external $\bar \phi$ contracts.

The corresponding sum gives
\be
a'=  (2^4  \lambda /8)  \sum_{p} \tfrac{1}{p^2}  = 2a \qquad (B1) \ .
\ee
so that at one loop equation \ref{beautiful} holds.
In the real case we have a combinatoric factor $4^3$ instead of $2^4$,
but the coupling constant is $\lambda/4$, so $a$ is the same and \ref{beautiful} holds.

\begin{figure}
\centerline{\includegraphics[width=5cm,angle=-90]{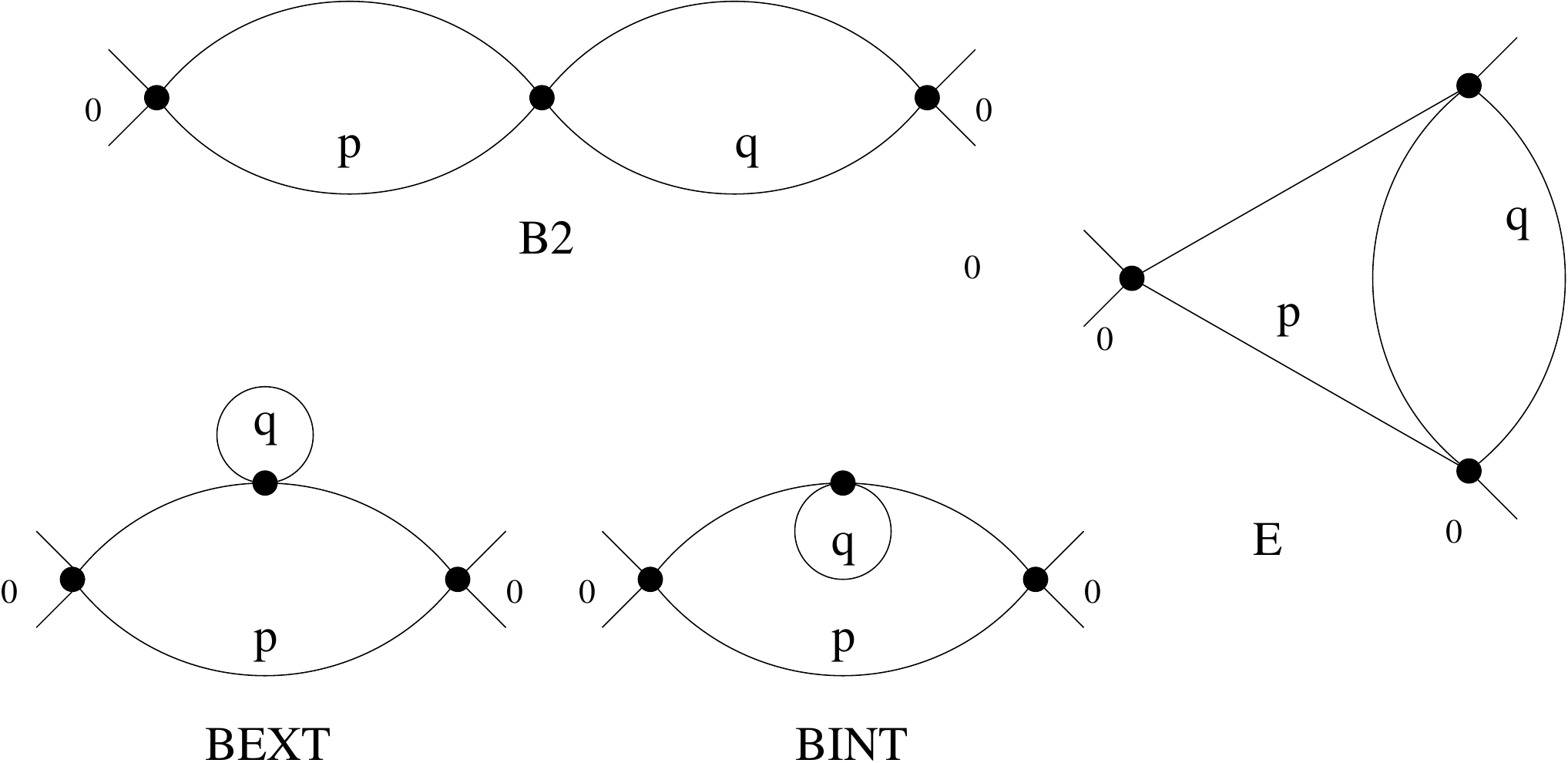}}
\caption{Four Point Graphs at two Loops}
\label{4p12l}\end{figure}

\subsection{Two and Three Loops}

This computation was extended to two and three loops in \cite{DisertoriRivasseau2006}. 
The results were given in the form of tables for the discrete divergent sums
and combinatoric weights of all  planar regular graphs which appears at two and three loops in 
$\Ga_4$ and $Z$. The equation \ref{beautiful} holds again, both in the real and complex cases.

Here we simply reproduce the list of contributing Feynman graphs. Indeed it is interesting to notice
that although at large order there are less planar regular graphs than the general graphs
of the commutative theory, the effect is opposite at small orders.

\begin{figure}
\label{twopointtwoloops}
\centerline{\includegraphics[width=7cm,angle=-90]{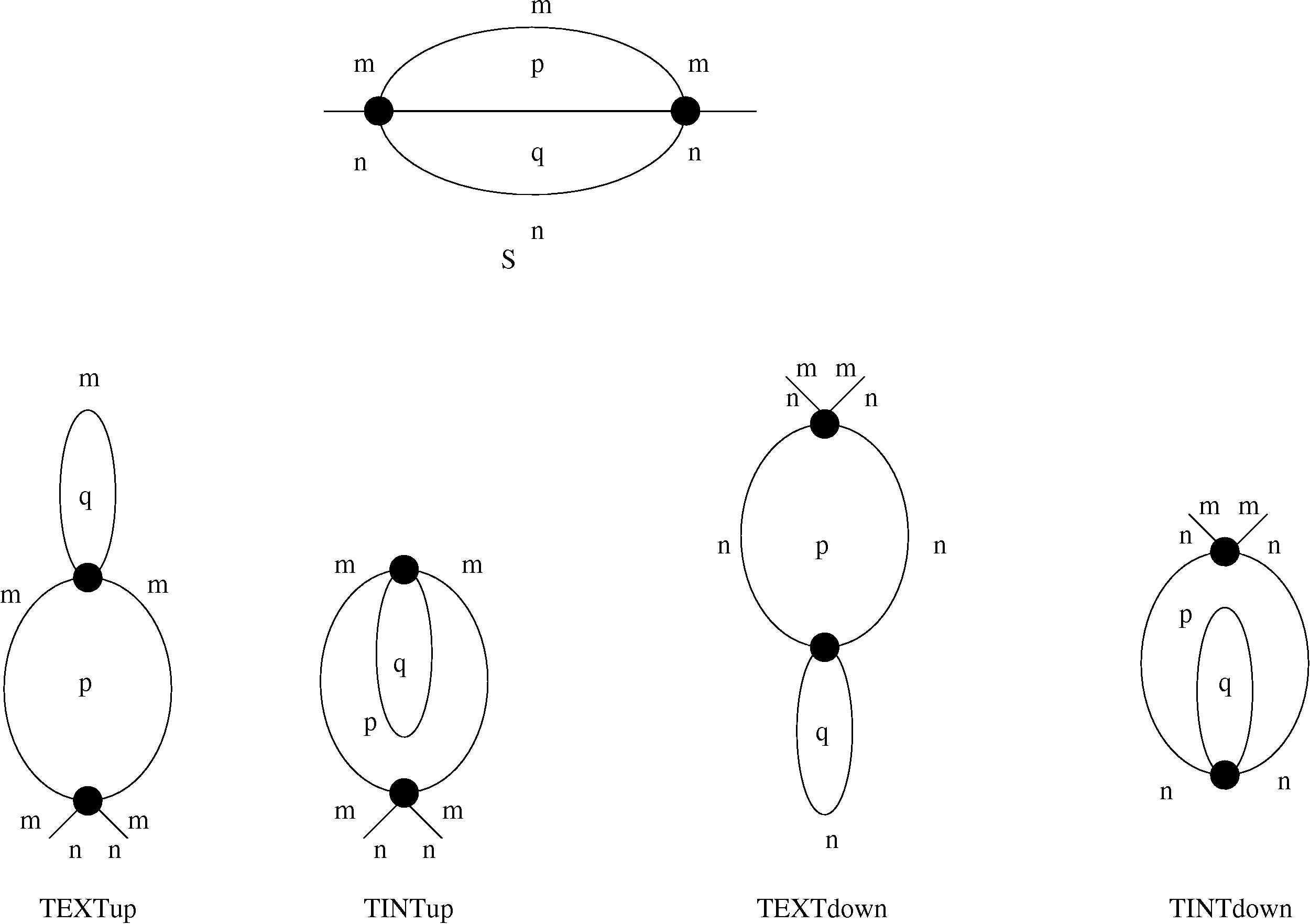}}
\caption{Two Point Graphs at Two Loops}
\label{2p2l}
\end{figure}

\begin{figure}\label{twopointthreeloops}
\centerline{\includegraphics[width=11cm]{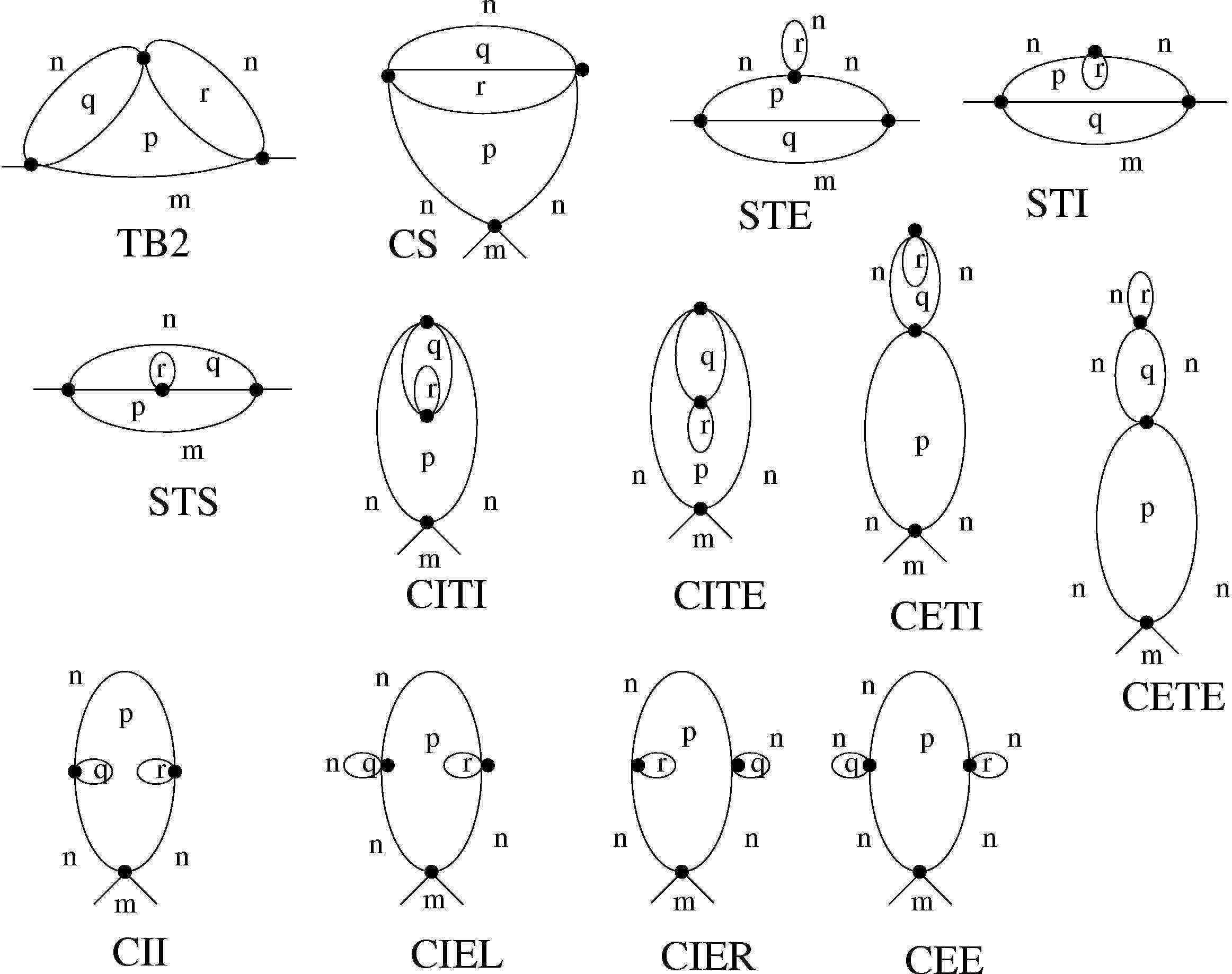}}
\caption{Two Point Graphs at Three Loops}
\label{2p3l}
\end{figure}

\begin{figure}
\centerline{\includegraphics[width=10cm]{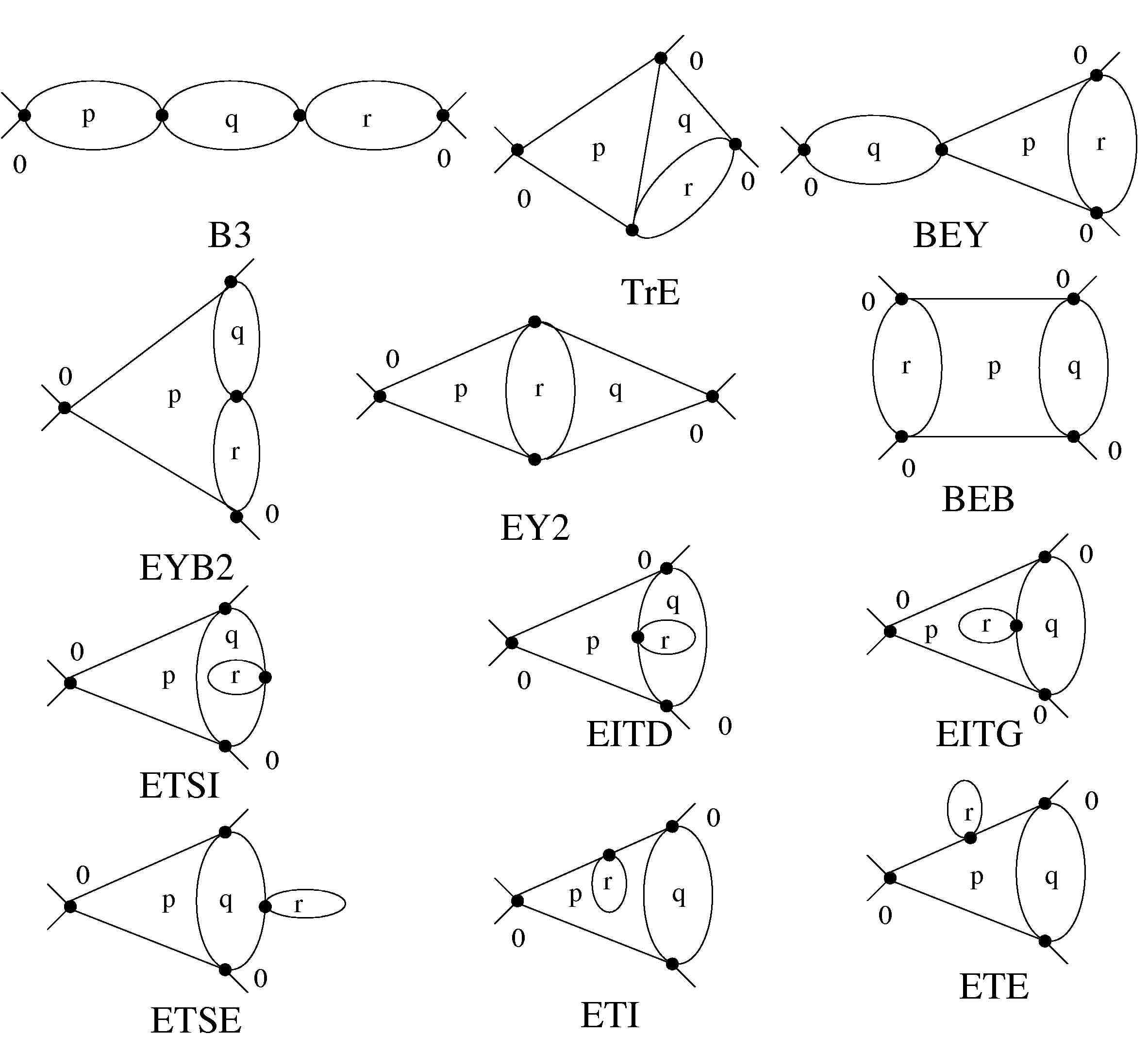}}
\caption{Four Point Graphs at Three Loops, Part I}
\label{4p3l1}\end{figure}

\begin{figure}
\centerline{\includegraphics[width=10cm]{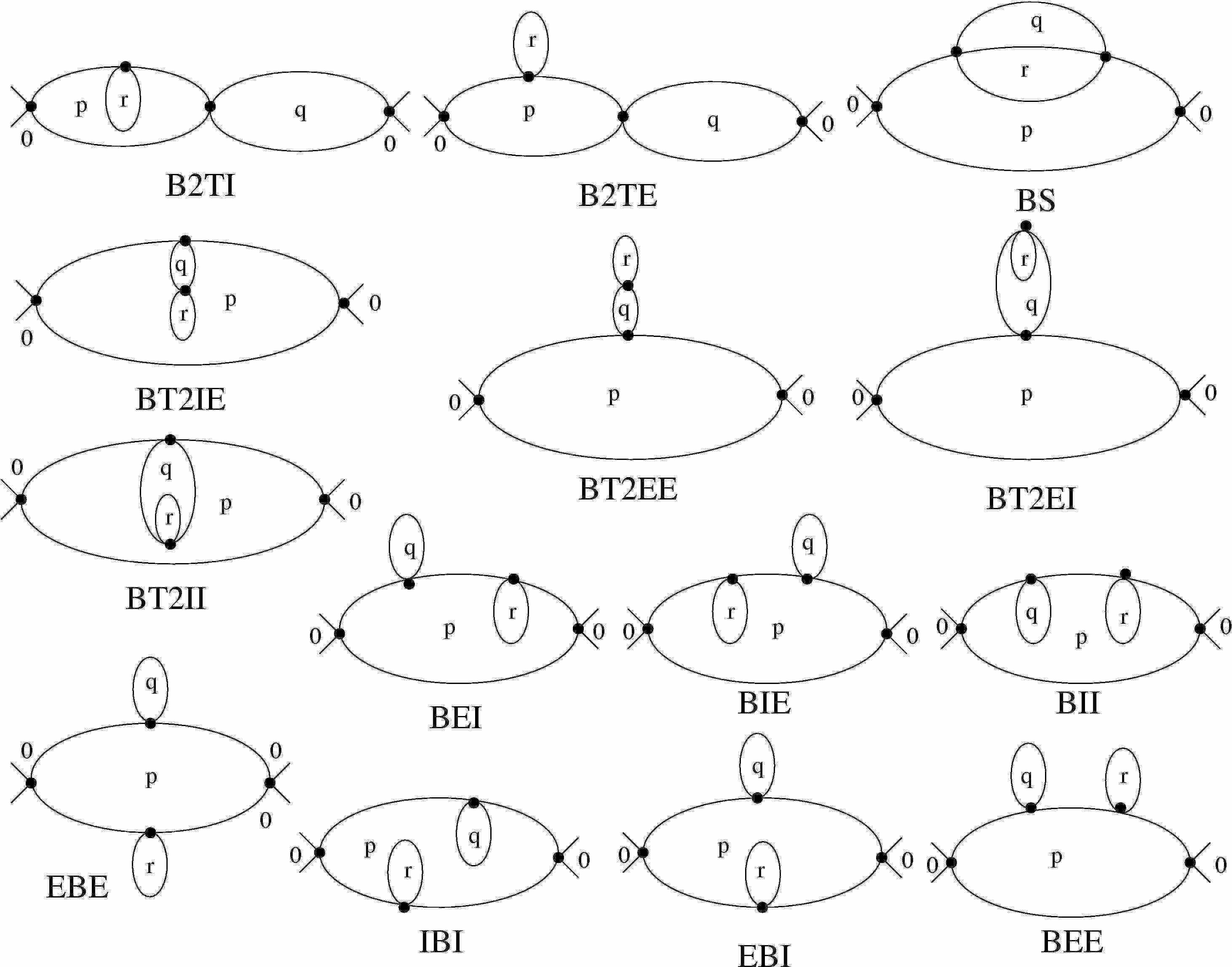}}
\caption{Four Point Graphs at Three Loops, Part II}
\label{4p3l2}\end{figure}


\subsection{The General Ward Identity}

In this section, essentially reproduced from \cite{beta2-06},
we prove a general Ward identity which allows to check
that theorem \ref{killghost} continue to hold at any order
in perturbation theory.

We orient the propagators from a $\bar{\phi}$ to a $\phi$.
For a field $\bar{\phi}_{a b}$ we call the index $a$ a 
{\it left index} and the index, $b$ a {\it right index}. The first (second) index of a $\bar{\phi}$ {\it always} contracts with the second (first) index of a $\phi$.  Consequently for $\phi_{c d}$, 
$c$ is a {\it right index} and $d$ is a {\it left index}.

Let $U=e^{\imath B}$ with $B$ a small hermitian matrix. We consider the ``left" (as it acts only on the left indices) change of variables:
\bea
\phi^U=\phi U;\bar{\phi}^U=U^{\dagger}\bar{\phi} \ .
\eea
There is a similar ``right" change of variables. The variation of the action is, at first order:
\bea
\delta S&=&\phi U X U^{\dagger}\bar{\phi}-\phi X \bar{\phi}\approx
\imath\big{(}\phi B X\bar{\phi}-\phi X B \bar{\phi}\big{)}\nonumber\\
&=&\imath B\big{(}X\bar{\phi}\phi-\bar{\phi}\phi X \big{)}
\eea
and the variation of the external sources is:
\bea
\delta F&=&U^{\dagger}\bar{\phi}\eta-\bar{\phi}\eta+\bar{\eta}\phi U-\bar{\eta}\phi 
        \approx-\imath B \bar{\phi}\eta+\imath\bar{\eta}\phi B\nonumber\\
	&=&\imath B\big{(}-\bar{\phi}\eta+\bar{\eta}\phi{)} .
\eea
We obviously have:
\bea
&&\frac{\delta \ln Z}{\delta B_{b a}}=0=\frac{1}{Z(\bar{\eta},\eta)}\int d\bar{\phi} d\phi
   \big{(}-\frac{\delta S}{\delta B_{b a}}+\frac{\delta F}{\delta B_{b a}}\big{)}e^{-S+F}\nonumber\\
   &&=\frac{1}{Z(\bar{\eta},\eta)}\int d\bar{\phi} d\phi  ~e^{-S+F}
\big{(}-[X \bar{\phi}\phi-\bar{\phi}\phi X]_{a b}+
       [-\bar{\phi}\eta+\bar{\eta}\phi]_{a b}\big{)} \  .
\eea

We now apply $\partial_{\eta}\partial_{\bar{\eta}}|_{\eta=\bar{\eta}=0}$ 
on the above expression. As we have at most two insertions,
we get only the connected components of the correlation functions.
\bea
0=<\partial_{\eta}\partial_{\bar{\eta}}\big{(}
-[X \bar{\phi}\phi-\bar{\phi}\phi X]_{a b}+
       [-\bar{\phi}\eta+\bar{\eta}\phi]_{a b}\big{)}e^{F(\bar{\eta},\eta)} |_0>_c \ ,
\eea
which gives:
\bea
<\frac{\partial(\bar{\eta}\phi)_{a b}}{\partial \bar{\eta}}\frac{\partial(\bar{\phi}\eta)}{\partial \eta}
-\frac{\partial(\bar{\phi}\eta)_{a b}}{\partial \eta}\frac{\partial (\bar{\eta}\phi)}{\partial \bar{\eta}}
- [X \bar{\phi}\phi-\bar{\phi}\phi X]_{a b}
\frac{\partial(\bar{\eta}\phi)}{\partial \bar{\eta}}\frac{\partial (\bar{\phi}\eta)}{\partial\eta}>_c=0 .
\eea
Using the explicit form of $X$ we get:
\bea
(a-b)<[ \bar{\phi}\phi]_{a b}
\frac{\partial(\bar{\eta}\phi)}{\partial \bar{\eta}}\frac{\partial (\bar{\phi}\eta)}{\partial\eta}>_c=
<\frac{\partial(\bar{\eta}\phi)_{a b}}{\partial \bar{\eta}}\frac{\partial(\bar{\phi}\eta)}{\partial \eta}>_c
-<\frac{\partial(\bar{\phi}\eta)_{a b}}{\partial \eta}\frac{\partial (\bar{\eta}\phi)}{\partial \bar{\eta}}> \ ,
\nonumber
\eea
and for $\bar{\eta}_{ \beta \alpha} \eta_{ \nu \mu}$ we get:
\bea
(a-b)<[ \bar{\phi}\phi]_{a b} \phi_{\alpha \beta} 
\bar{\phi}_{\mu \nu }>_c=
<\delta_{a \beta}\phi_{\alpha b} \bar{\phi}_{\mu \nu}>_c
-<\delta _{b \mu }\bar{\phi}_{a \nu} \phi_{\alpha \beta}>_c
\eea

We restrict to terms in the above expressions which are planar with a single external face,
as all others are irrelevant. Such terms have $\alpha=\nu$, $a=\beta$ and $b=\mu$. 
The Ward identity for the $2$ point function reads:
\bea\label{ward2point}
(a-b)<[ \bar{\phi}\phi]_{a b} \phi_{\nu a} 
\bar{\phi}_{b \nu }>_c=
<\phi_{\nu b} \bar{\phi}_{b \nu}>_c
-<\bar{\phi}_{a \nu} \phi_{\nu a}>_c
\eea
(repeated indices are not summed). 

\begin{figure}[hbt]
\centerline{
\includegraphics[width=100mm]{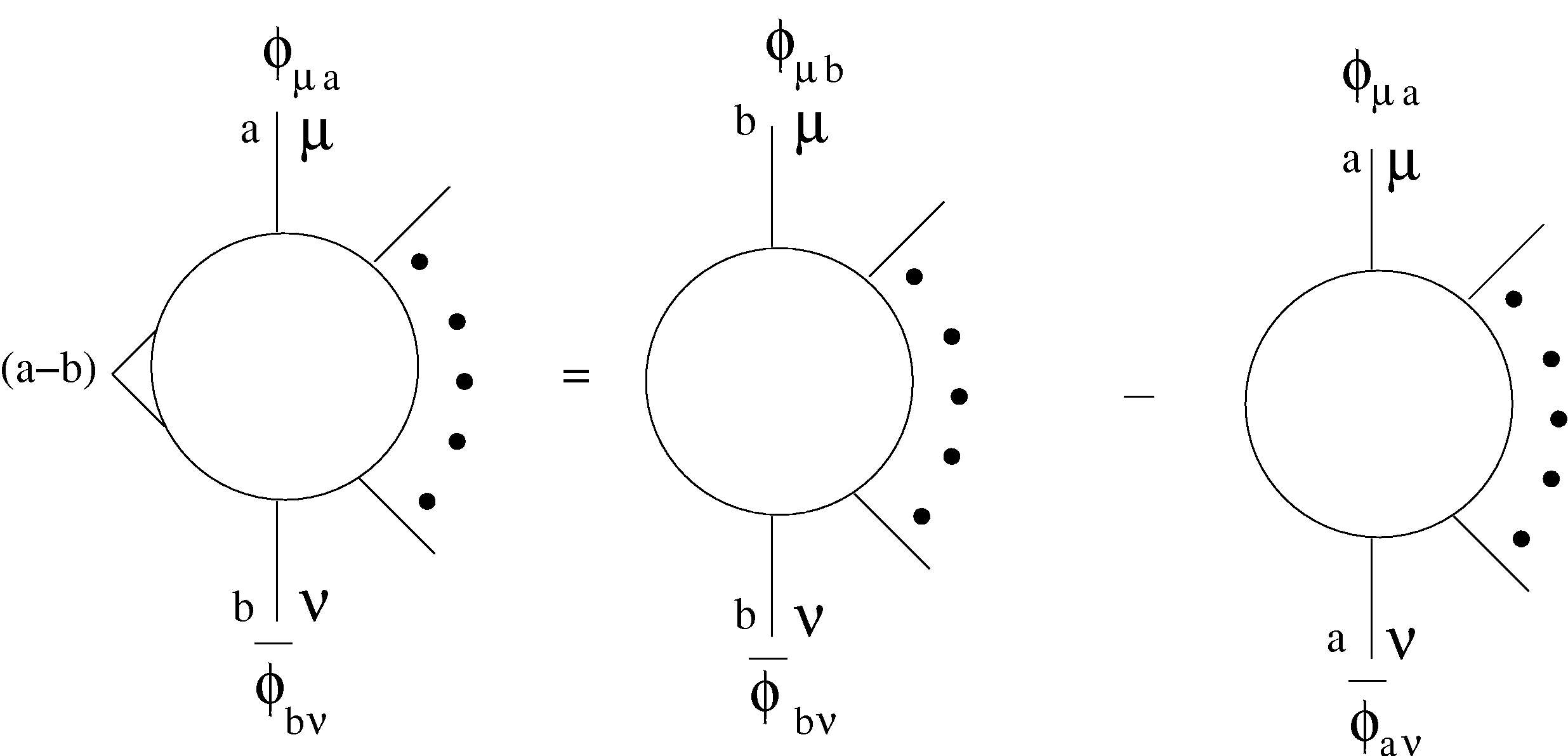}
}
\caption{The Ward identity for a 2p point function with insertion on the left face}\label{fig:Ward}
\end{figure}

Derivating further we get:
\bea
&&(a-b)<[\bar{\phi}\phi]_{a b}\partial_{\bar{\eta}_1}(\bar{\eta}\phi)
\partial_{\eta_1}(\bar{\phi}\eta) \partial_{\bar{\eta}_2}(\bar{\eta}\phi)
\partial_{\eta_2}(\bar{\phi}\eta) >_c=\\
&&<\partial_{\bar{\eta}_1}(\bar{\eta}\phi)
\partial_{\eta_1}(\bar{\phi}\eta)\big{[}
 \partial_{\bar{\eta_2}}
 (\bar{\eta}\phi)_{ab}\partial_{\eta_2}(\bar{\phi}\eta)-\partial_{\eta_2}(\bar{\phi}\eta)_{a b}
 \partial_{\bar{\eta}_2}(\bar{\eta}\phi) \big{]}>_c+1 \leftrightarrow 2 \ .\nonumber
\eea
Take $\bar{\eta}_{1~\beta \alpha}$, $\eta_{1~ \nu\mu}$, $\bar{\eta}_{2~\delta \gamma}$ and $\eta_{2~\sigma \rho}$.
We get:
\bea
&&(a-b)<[\bar{\phi}\phi]_{ab}\phi_{\alpha \beta}\bar{\phi}_{\mu \nu}\phi_{\gamma \delta}
\bar{\phi}_{\rho \sigma}>_c\\
&&=<\phi_{\alpha \beta}\bar{\phi}_{\mu \nu}  \delta_{a \delta}\phi_{\gamma b}\bar{\phi}_{\rho \sigma}>_c
-<\phi_{\alpha \beta}\bar{\phi}_{\mu \nu}\phi_{\gamma \delta}\bar{\phi}_{a \sigma}\delta_{b \rho}>_c+
\nonumber\\
&&<\phi_{\gamma \delta}\bar{\phi}_{\rho \sigma}  \delta_{a \beta}\phi_{\alpha b}\bar{\phi}_{\mu \nu}>_c
-<\phi_{\gamma \delta}\bar{\phi}_{\rho \sigma}\phi_{\alpha \beta}\bar{\phi}_{a \nu}\delta_{b \mu}>_c \ .
\nonumber
\eea
Again neglecting all terms which are not planar with a single external face leads to
\bea\label{ward4point}
(a-b) <\phi_{\alpha a}[\bar{\phi}\phi]_{ab}\bar{\phi}_{b\nu}\phi_{\nu \delta}\bar{\phi}_{\delta \alpha}>_c=
<\phi_{\alpha b}\bar{\phi}_{b \nu}\phi_{\nu \delta}\bar{\phi}_{\delta\alpha}>_c-
<\phi_{\alpha a}\bar{\phi}_{a \nu}\phi_{\nu \delta}\bar{\phi}_{\delta\alpha}>_c \ .
\nonumber
\eea
Clearly there are similar identities for $2p$ point functions for any $p$.

The indices $a$ and $b$ are left indices, so that we have the Ward identity with an insertion on a left face as 
represented in  Fig. \ref{fig:Ward}.
There is a similar Ward identity obtained with the ``right" transformation, consequently with the insertion on a right face.

\subsubsection{Proof of Theorem \ref{killghost}}

We start this section by some definitions:
we will denote $G^{4}(m,n,k,l)$ the connected four point function restricted to the planar one broken face case, where $m,n,k,l$  are the indices of the external face in the correct cyclic order. The first index $m$ always represents a left index.

Similarly, $G^{2}(m,n)$ is the connected planar one broken face two point function with $m,n$ the indices on the external face (also called the {\bf dressed} propagator, see Fig. \ref{fig:propagators}). $G^{2}(m,n)$ and $\Sigma(m,n)$ are related by:
\bea
\label{G2Sigmarelation}
  G^{2}(m,n)=\frac{C_{m n}}{1-C_{m n}\Sigma(m,n)}=\frac{1}{C_{m n}^{-1}-\Sigma(m,n)} \, .
\eea 

\begin{figure}[hbt]
\centerline{
\includegraphics[width=60mm]{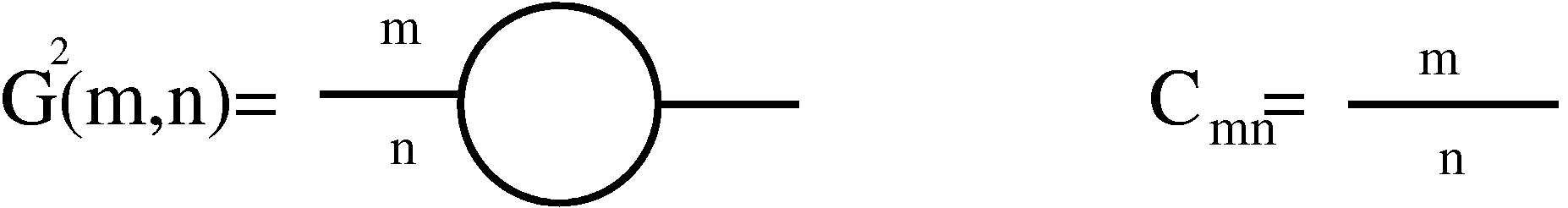}
}
\caption{The {\bf dressed} and the bare propagators}\label{fig:propagators}
\end{figure}

$G_{ins}(a,b;...)$ will denote the planar one broken face 
connected function with one insertion on the left border where the matrix index jumps from $a$ to $b$. With this notations the Ward identity (\ref{ward2point}) writes:
\bea
(a-b) ~ G^{2}_{ins}(a,b;\nu)=G^{2}(b,\nu)-G^{2}(a,\nu)\, .
\eea

All the identities we use, either Ward identities or the Dyson equation of motion
can be written either for the bare  theory or for the theory with complete mass renormalization, which is the one considered in \cite{DisertoriRivasseau2006}. In the first case the parameter $A$ in (\ref{propafixed}) is the bare one, $A_{bare}$
and there is no mass subtraction. In the second case the parameter $A$ in (\ref{propafixed}) 
is $A_{ren}= A_{bare} - \Sigma(0,0)$, and every two point 1PI subgraph is subtracted at 0 external indices\footnote{These mass subtractions need not be rearranged into forests 
since 1PI 2point subgraphs never overlap non trivially.}. $\partial_{L}$ denotes the derivative with respect to a left index and $\partial_{R}$ the one with respect to a right index. When the two derivatives are equal we use 
the generic notation $\partial$. 

Let us prove first the Theorem in the mass-renormalized case, then in the next subsection
in the bare case. Indeed the mass renormalized theory used is free from any quadratic divergences. Remaining logarithmic subdivergences in the ultra violet cutoff can be removed easily by passing to the effective series
as explained in \cite{DisertoriRivasseau2006}. 

\begin{figure}[hbt]
\centerline{
\includegraphics[width=120mm]{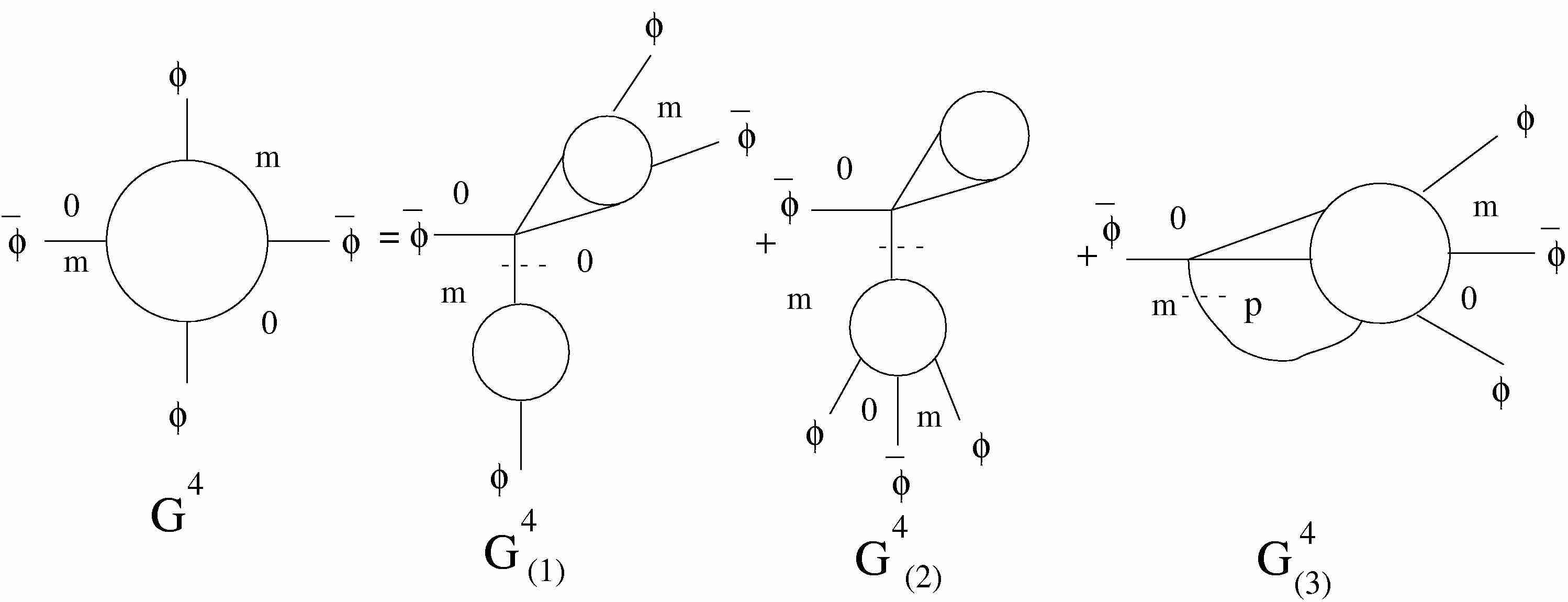}
}
\caption{The Dyson equation}\label{fig:dyson}
\end{figure}

We analyze a four point connected function $G^4(0,m,0,m)$ with index $m \ne 0$ on the right borders. 
This explicit break of left-right symmetry is adapted to our problem.

Consider a $\bar{\phi}$ external line and the first vertex hooked to it. 
Turning right on the $m$ border at this vertex we meet a new line (the slashed line in Fig. \ref{fig:dyson}). The slashed line either separates the graph into two disconnected components ($G^{4}_{(1)}$ and $G^{4}_{(2)}$ in Fig. \ref{fig:dyson}) or not 
($G^{4}_{(3)}$ in Fig. \ref{fig:dyson}). Furthermore, if the slashed line separates the graph into two disconnected components the first vertex may either belong to the four point component ($G^{4}_{(1)}$ in Fig. \ref{fig:dyson}) or to the two point component
($G^{4}_{(2)}$ in Fig. \ref{fig:dyson}). 

We stress that this is a {\it classification} of graphs: the different components depicted in Fig. \ref{fig:dyson} take into account all the combinatoric factors. Furthermore, the setting of the external indices to $0$ on the left borders and $m$ on the right borders distinguishes the $G^{4}_{(1)}$ and $G^{4}_{(2)}$ from their counterparts ``pointing upwards": indeed, the latter are classified in $G^{4}_{(3)}$!

We have thus the Dyson equation:
\bea
\label{Dyson}
 G^4(0,m,0,m)=G^4_{(1)}(0,m,0,m)+G^4_{(2)}(0,m,0,m)+G^4_{(3)}(0,m,0,m)\, .
\eea    

The second term,  $G^{4}_{(2)}$, is zero. Indeed the mass renormalized two point insertion is zero, as it has the external left index set to zero. Note that this is an insertion exclusively on the left border. The simplest case of such an insertion is a
 (left) tadpole. We will (naturally) call a general insertion touching only the left border a ``generalized left tadpole". 

We will prove that $G^{4}_{(1)}+G^{4}_{(3)}$ yields 
$\Gamma^4=\lambda (1-\partial \Sigma)^2$ after amputation of the four external propagators.

We start with $G^{4}_{(1)}$. It is of the form:
\bea
G^4_{(1)}(0,m,0,m)=\lambda C_{0 m} G^{2}(0, m) G^{2}_{ins}(0,0;m)\,.
 \eea

By the Ward identity we have:
\bea
G^{2}_{ins}(0,0;m)&=&\lim_{\zeta\rightarrow 0}G^{2}_{ins}(\zeta ,0;m)=
\lim_{\zeta\rightarrow 0}\frac{G^{2}(0,m)-G^{2}(\zeta,m)}{\zeta}\nonumber\\
&=&-\partial_{L}G^{2}(0,m) \, .
\eea
Using the explicit form of the bare propagator we have $\partial_L C^{-1}_{ab}=\partial_R C^{-1}_{ab}=\partial C^{-1}_{ab}=1$. Reexpressing $G^{2}(0,m)$ by eq.  (\ref{G2Sigmarelation}) we conclude that:
\bea\label{g41}
G^4_{(1)}(0,m,0,m)&=&\lambda
C_{0m}\frac{C_{0m}C^2_{0m}[1-\partial_{L}\Sigma(0,m)]}{[1-C_{0m}\Sigma(0,m)]
(1-C_{0m}\Sigma(0,m))^2}\nonumber\\
&=&\lambda [G^{2}(0,m)]^{4}\frac{C_{0m}}{G^{2}(0,m)}[1-\partial_{L}\Sigma(0,m)]\, .
\eea
The self energy is (again up to irrelevant terms (\cite{c}):
\bea
\label{PropDressed}
\Sigma(m,n)=\Sigma(0,0)+(m+n)\partial\Sigma(0,0) 
\eea 
Therefore up to irrelevant terms ($C^{-1}_{0m}=m+A_{ren}$) we have:
\bea
\label{G2(0,m)}
G^{2}(0,m)=\frac{1}{m+A_{bare}-\Sigma(0,m)}=\frac{1}{m[1-\partial\Sigma(0,0)]+A_{ren}}
\, ,
\eea
and
\bea \label{cdressed}
\frac{C_{0m}}{G^{2}(0,m)}=1-\partial\Sigma(0,0)+\frac{A_{ren}}{m+A_{ren}}\partial\Sigma(0,0) \, .
\eea
Inserting eq. (\ref{cdressed}) in eq. (\ref{g41}) holds:
\bea
\label{g41final}
G^4_{(1)}(0,m,0,m)&=&\lambda [G^{2}(0,m)]^{4}\bigl( 
1-\partial\Sigma(0,0)+\frac{A_{ren}}{m+A_{ren}}\partial\Sigma(0,0)
\bigr) \nonumber\\
&&[1-\partial_{L}\Sigma(0,m)]\, .
\eea

\begin{figure}[hbt]
\centerline{
\includegraphics[width=120mm]{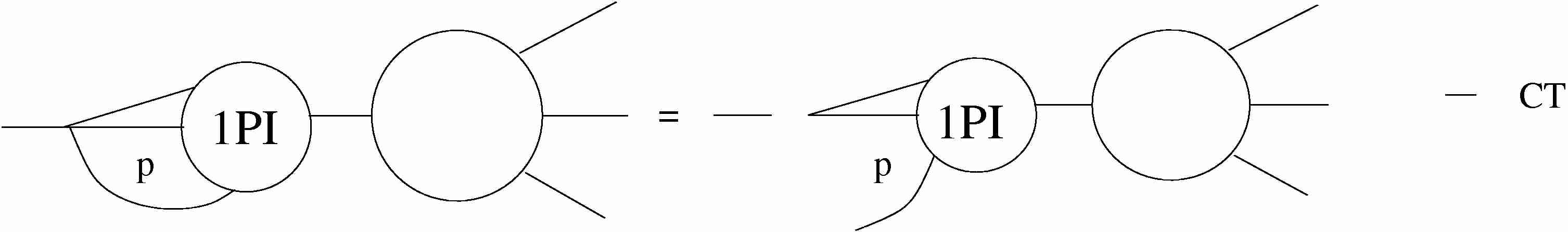}
}
\caption{Two point insertion and opening of the loop with index $p$}\label{fig:insertion}
\end{figure}

For the $G^4_{(3)}(0,m,0,m)$ one starts by ``opening" the face which is ``first on the right''. The summed index of this face is called  $p$ (see Fig. \ref{fig:dyson}).  For bare Green functions this reads:
\bea
\label{opening}
G^{4,bare}_{(3)}(0,m,0,m)=C_{0m}\sum_{ p} G^{4,bare}_{ins}(p,0;m,0,m)\, .
\eea
When passing to mass renormalized Green functions one must be cautious. It is possible that the face $p$ belonged to a  1PI two point insertion in $G^{4}_{(3)}$ (see the left hand side in Fig. \ref{fig:insertion}).
Upon opening the face $p$ this 2 point insertion disappears (see right hand side of Fig. \ref{fig:insertion})! 
When renormalizing, the counterterm  coresponding to this kind of two point insertion will be subtracted on the left hand side of  eq.(\ref{opening}), but not on the right hand side. In the equation for $G^{4}_{(3)}(0,m,0,m)$ one must 
therefore \textit{add its missing counterterm}, so that:
\bea
\label{Open2}
G^4_{(3)}(0,m,0,m)&=& C_{0m}\sum_{p} G^{4}_{ins}(0,p;m,0,m)\nonumber\\
     &-&C_{0m}(CT_{lost})G^{4}(0,m,0,m)\,.
\eea

It is clear that not all 1PI 2 point insertions on the left hand side of Fig. \ref{fig:insertion} will be ``lost" on the right hand side. If the insertion is a ``generalized left tadpole" it is not ``lost" by opening the face $p$ (imagine a tadpole pointing upwards in Fig.\ref{fig:insertion}: clearly it will not be opened by opening the line). 
We will call the 2 point 1PI insertions ``lost" on the right hand side $\Sigma^R(m,n)$. 
Denoting the generalized left tadpole $T^{L}$ we can write (see Fig .\ref{fig:selfenergy}):
\bea
\label{eq:leftright}
  \Sigma(m,n)=T^{L}(m,n)+\Sigma^R(m,n)\, .
\eea
Note that as $T^{L}(m,n)$ is an insertion exclusively on the left border, it does not depend upon the right index $n$. We therefore have $\partial\Sigma(m,n)=\partial_R\Sigma(m,n)=\partial_R\Sigma^R(m,n)$.

\begin{figure}[hbt]
\centerline{
\includegraphics[width=90mm]{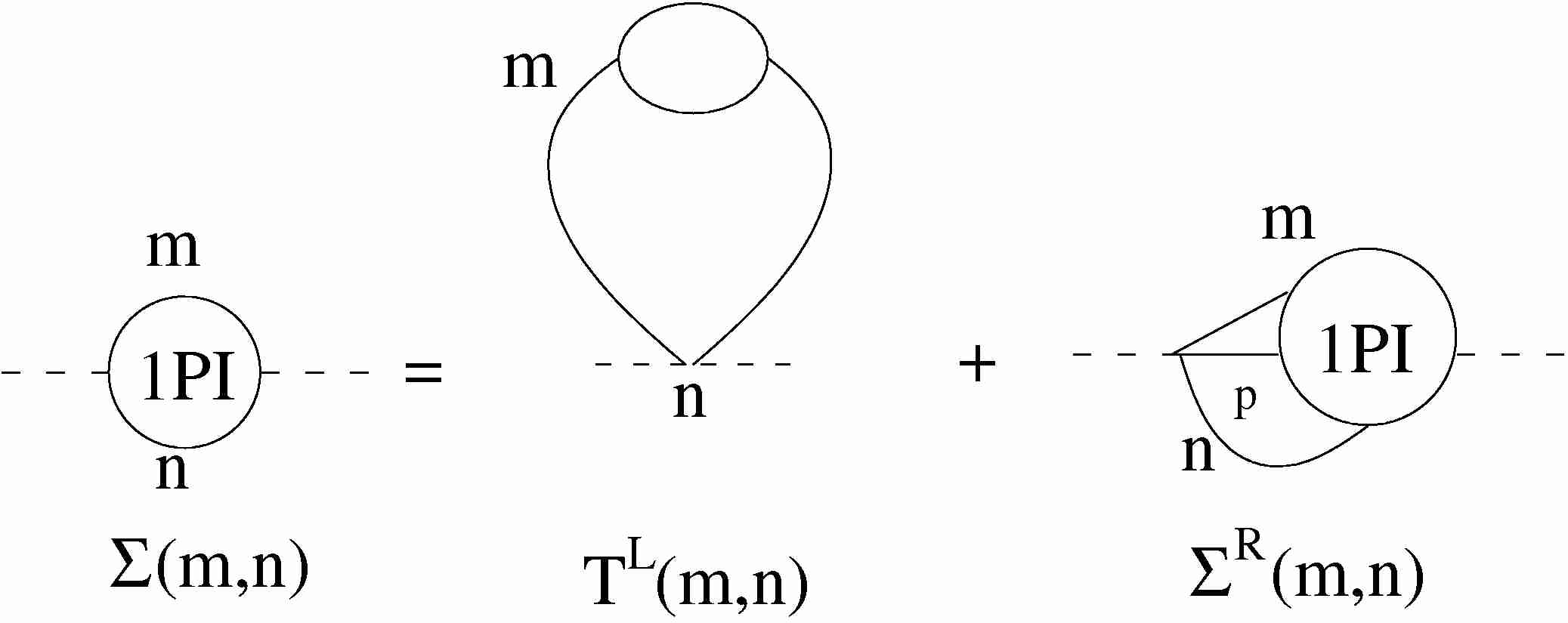}
}
\caption{The self energy}\label{fig:selfenergy}
\end{figure}

The missing mass counterterm writes:
\bea\label{lostct}
CT_{lost}=\Sigma^R(0,0)=\Sigma(0,0)-T^{L}\, .
\eea
In order to evaluate $\Sigma^{R}(0,0)$ we procede
by opening its face $p$ and using the Ward identity (\ref{ward2point}), to obtain:
\bea
\label{S2}
\Sigma^R(0,0)&=&\frac{1}{G^{2}(0,0)}\sum_{p}G^2_{ins}(0,p;0)\nonumber\\
    &=&\frac{1}{G^{2}(0,0)}\sum_{p}\frac{1}{p}[G^2(0,0)-G^2(p,0)]\nonumber\\
    &=&\sum_{p}\frac{1}{p} \biggl(1 -\frac{G^{2}(p,0)}{G^{2}(0,0)}\biggr) \, .
\eea

Using eq. (\ref{Open2}) and eq. (\ref{S2}) we have:
\bea\label{S3}
G^4_{(3)}(0,m,0,m)&=& C_{0m}\sum_{p} G^{4}_{ins}(0,p;m,0,m)\nonumber\\
&-&C_{0m} G^{4}(0,m,0,m) \sum_{p}\frac{1}{p}  
\biggl( 1- \frac{G^{2}(p,0)}{G^2(0,0)}  \biggr)  \, .
\eea

But by the Ward identity (\ref{ward4point}): 
\bea
\label{Ward4}
C_{0m} \sum_{p} G^{4}_{ins}(0,p;m,0,m)=C_{0m} \sum_{p} \frac{1}{p}\biggl( G^{4}(0,m,0,m)-G^{4}(p,m,0,m) \biggr)\, ,
\eea
The second term in eq. (\ref{Ward4}), having at least three denominators 
linear in $p$, is irrelevant \footnote{Any perturbation order of $G^4(p,m,0,m)$ is a polynomial in $\ln(p)$ divided by $p^2$. Therefore the sums over $p$ above are always convergent.}
. Substituting eq. (\ref{Ward4}) in eq . (\ref{S3}) we have:
\bea
\label{G3}
G^4_{(3)}(0,m,0,m) =C_{0m}\frac{G^{4}(0,m,0,m)}{G^2(0,0)}\sum_{p}
\frac{G^{2}(p,0)}{p} \, .
\eea
To conclude we must evaluate the sum in eq. (\ref{G3}). Using eq. (\ref {G2(0,m)}) we have:
\bea
\label{derivee}
\sum_{p}\frac{G^{2}(p,0)}{p}=\sum_{p}\frac{G^{2}(p,0)}{p}\bigl( \frac{1}{G^{2}(0,1)}-\frac{1}{G^{2}(0,0)}\bigr)
\frac{1}{1-\partial\Sigma(0,0)}
\eea

In order to interpret the two terms in the above equation we start by performing the same manipulations as in eq (\ref{S2}) for $\Sigma^R(0,1)$. We get:
\bea
\label{S2new}
\Sigma^R(0,1)&=&\sum_{p}\frac{1}{p} 
\biggl(1 -\frac{G^{2}(p,1)}{G^2(0,1)}\biggr) =\sum_{p}\frac{1}{p} 
\biggl(1 -\frac{G^{2}(p,0)}{G^2(0,1)}\biggr)\, .
\eea
where in the second equality we have neglected an irrelevant term. 

Substituting  eq. (\ref{S2}) and eq. (\ref{S2new}) in eq. (\ref{derivee}) we get:
\bea
\sum_{p}\frac{G^{2}(p,0)}{p}=\frac{\Sigma^R(0,0)-\Sigma^R(0,1)}{1-\partial\Sigma(0,0)}
=-\frac{\partial_{R}\Sigma^R(0,0)}{1-\partial\Sigma(0,0)}
=-\frac{\partial\Sigma(0,0)}{1-\partial\Sigma(0,0)}\, .
\eea
as $\partial_R\Sigma^R=\partial\Sigma$. Hence:
\bea\label{g43}
G^4_{(3)}(0,m,0,m;p)&=&-C_{0m}G^{4}(0,m,0,m)\frac{1}{G^{2}(0,0)}\frac{\partial\Sigma(0,0)}{1-\partial\Sigma(0,0)}
\nonumber\\
&=&-G^{4}(0,m,0,m)
\frac{A_{ren} \; \partial\Sigma(0,0)}{(m+A_{ren}) [1-\partial\Sigma(0,0)]} \ .
\eea
Using (\ref{g41final}) and (\ref{g43}), equation (\ref{Dyson}) rewrites as:
\bea
\label{final}
&&G^4(0,m,0,m)\Big{(}1+
\frac{A_{ren}\; \partial\Sigma(0,0)}{(m+A_{ren}) \; [ 1-\partial\Sigma(0,0)] }\Big{)}
\\
&&=\lambda_{bare} (G^{2}(0,m))^{4}\Big{(}1-\partial\Sigma(0,0)+\frac{A_{ren}}{m+A_{ren}}\partial\Sigma(0,0)\Big{)}
[1-\partial_{L}\Sigma(0,m)]\, .\nonumber
\eea
We multiply (\ref{final}) by $[1-\partial\Sigma(0,0)]$ and amputate four times. As the differences $\Gamma^4(0,m,0,m,)-\Gamma^4(0,0,0,0)$ and $\partial_L\Sigma(0,m)-\partial_L\Sigma(0,0)$ are irrelevant we get:
\bea
\Gamma^{4}(0,0,0,0)=\lambda (1-\partial\Sigma(0,0))^2\, .
\eea
\qed 

\subsubsection{Bare identity}

Let us explain now why the main theorem is also true as an identity between bare functions, without
any renormalization, but with ultraviolet cutoff.

Using the same Ward identities, all the equations go through with only few differences:

- we should no longer add the lost mass counterterm in (\ref{lostct})

- the term $G^{4}_{(2)}$ is no longer zero.

- equation (\ref{cdressed}) and all propagators now involve the bare $A$ parameter.

But these effects compensate. Indeed the bare  $G^{4}_{(2)}$ term is the left generalized
tadpole $\Sigma - \Sigma^R$, hence
\begin{equation} \label{newleft}
G^{4}_{(2)}  (0,m,0,m) = C_{0,m} \bigl(  \Sigma(0,m) - \Sigma^R (0,m) \bigr) G^4(0,m,0,m)\; .
\end{equation}
Equation (\ref{cdressed}) becomes up to irrelevant terms
\bea \label{cdressedbare}
\frac{C^{bare}_{0m}}{G^{2,bare}(0,m)}=1-\partial_{L}\Sigma(0,0)+
\frac{A_{bare}}{m+A_{bare}}\partial_{L}\Sigma(0,0) 
- \frac{1}{m+A_{bare}}\Sigma(0,0) 
\eea
The first  term proportional to $ \Sigma(0,m) $ in (\ref{newleft})  combines with 
the new term in (\ref{cdressedbare}), and the second term proportional to $ \Sigma^R(0,m) $ in (\ref{newleft})
is exactly the former ``lost counterterm" (\ref{lostct}). This proves (\ref{beautiful}) in the bare case.

\subsection{The RG Flow}

It remains to understand better the meaning of the 
Langmann-Szabo symmetry which certainly lies behind this Ward identity.
Of course we also need to develop a non-perturbative or constructive analysis of the theory
to fully confirm the absence of the Landau ghost. If this constructive
analysis confirms the perturbative picture the expected 
non perturbative flow for the effective parameters $\lambda$ and $\Omega$ should be:

\bqa
 \frac{d\la_i}{di} &\simeq a (1- \Omega_i) F(\la_i)\ ,\\  
\frac{d\Omega_i}{di} &\simeq b(1- \Omega_i) G(\la_i) \ ,
\eqa
where $F(\la_i) = \la_i^2 + O(\la_i^3)$, $G(\la_i) = \la_i + O(\la_i^2)$
and $a,b\in \mathbb{R}$ are two constants.
The behavior of this system is qualitatively the same as
the simpler system
\bqa  \frac{d\la_i}{di} &\simeq& a (1- \Omega_i) \la_i^2 \ , \\
\frac{d\Omega_i}{di} &\simeq& b (1- \Omega_i) \la_i \ ,
\eqa
whose solution is
\be
\la_i = \la_0 e^{\frac{a}{b} (\Omega_i-\Omega_0)}\ ,
\ee
with $\Omega_i$ solution of  
\be
b \; i \; \la_0 = \int_{1-\Omega_i}^{1-\Omega_0} e^{\frac{au}{b}} \frac{du}{u}\ ,
\ee
hence going exponentially fast to 1 as $i$ goes to infinity. The
corresponding numerical flow is drawn on Figure \ref{flow}.
\begin{figure}[t]
\centerline{\includegraphics[width=6cm]{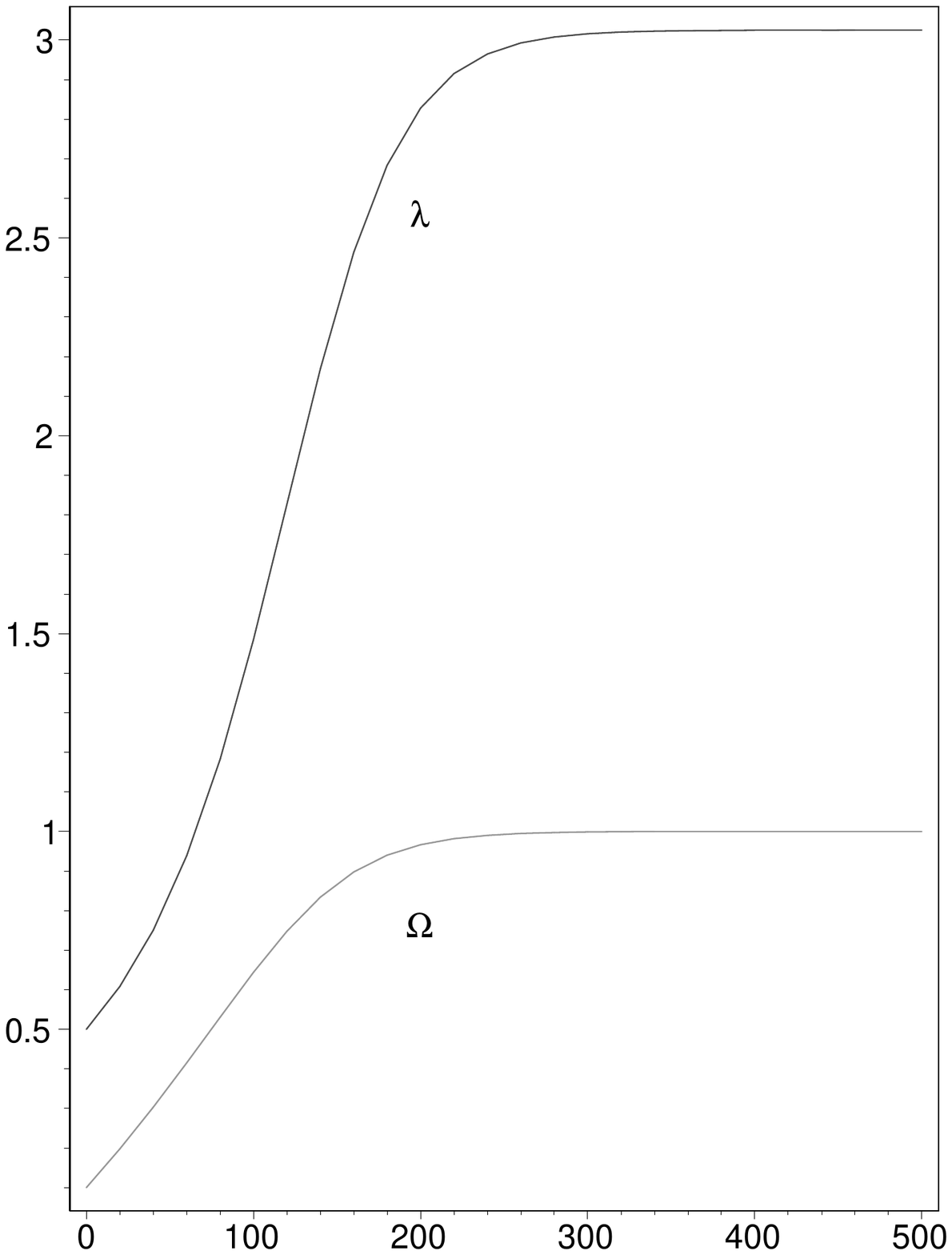}}
\caption{Numerical flow for $\la$ and $\Om$}
\label{flow}\end{figure}

Of course to establish fully rigorously this picture is beyond the
reach of perturbative theorems and requires a constructive analysis.

\section{Propagators on \encv{} space}
\label{sec:boite-outils}

We give here the results we get in \cite{toolbox05}. In this article, we computed the $x$-space and matrix 
basis kernels of operators which generalize the Mehler kernel \eqref{eq:Mehler}. Then we proceeded to 
a study of the scaling behaviors of these kernels in the matrix basis. This work is useful to study the 
\encv{} Gross-Neveu model in the matrix basis.

\subsection{Bosonic kernel}

The following lemma generalizes the Mehler kernel \cite{simon79funct}:
\begin{lemma}
  \label{HinXspace}Let $H$ the operator:
  \begin{equation}
    H=\frac{1}{2}\Big(-\Delta+ \Omega^2x^2-2\imath B(x_0\partial_1-x_1\partial_0)\Big).\label{eq:HinMat}
  \end{equation}
  The $x$-space kernel of $e^{-tH}$ is:
  \begin{equation}
    e^{-tH}(x,x')=\frac{\Omega}{2\pi\sinh\Omega t}e^{-A},\label{eq:propaxboson}
  \end{equation}
  \begin{equation}
    A=\frac{\Omega\cosh\Omega t}{2\sinh\Omega t}(x^2+x'^2)-
    \frac{\Omega\cosh Bt}{\sinh\Omega t}x\cdot x'-\imath
    \frac{\Omega\sinh Bt}{\sinh\Omega t}x\wedge x'.
  \end{equation}
\end{lemma}
\begin{rem}
  The Mehler kernel corresponds to $B=0$. The limit $\Omega=B\to 0$ gives the usual heat kernel.
\end{rem}
\begin{lemma}
  Let $H$ be given by \eqref{eq:HinMat} with $\Omega (B)\to 2\Omega/\theta (2B\theta)$. 
  Its inverse in the matrix basis is:
  \begin{align}
    H^{-1}_{m,m+h;l+h,l}=\frac{\theta}{8\Omega} \int_0^1 d\alpha\,  
\dfrac{(1-\alpha)^{\frac{\mu_0^2 \theta}{8 \Omega}+(\frac{D}{4}-1)}}{  
(1 + C\alpha )^{\frac{D}{2}}} (1-\alpha)^{-\frac{4B}{8\Omega}h}\prod_{s=1}^{\frac{D}{2}} 
G^{(\alpha)}_{m^s, m^s+h^s; l^s + h^s, l^s},\label{eq:propbosonmatrix}
\\
 G^{(\alpha)}_{m, m+h; l + h, l}
= \left(\frac{\sqrt{1-\alpha}}{1+C \alpha} 
\right)^{m+l+h} \sum_{u=\max(0,-h)}^{\min(m,l)}
   {\mathcal A}(m,l,h,u)
\left( \frac{C \alpha (1+\Omega)}{\sqrt{1-\alpha}(1-\Omega)} 
\right)^{m+l-2u},\notag
\end{align}
where ${\mathcal A}(m,l,h,u)=\sqrt{\binom{m}{m-u}
\binom{m+h}{m-u}\binom{l}{l-u}\binom{l+h}{l-u}}$ and $C$ is a function of $\Omega$ : $C(\Omega)=
\frac{(1-\Omega)^2}{4\Omega}$.
\end{lemma}

\subsection{Fermionic kernel}

On the Moyal space, we modified the commutative Gross-Neveu model by adding a $\xts$ term (see 
lemma \ref{xpropa1GN}). We have
\begin{eqnarray}
G(x,y) &=& -\frac{\Omega}{\theta\pi}\int_{0}^{\infty}\frac{dt}{\sinh(2\Ot
t)}\, e^{-\frac{\Ot}{2}\coth(2\Ot t)(x-y)^{2}+\imath\Ot
x\wedge y}
\\ 
&&    \lb\imath\Ot\coth(2\Ot t)(\xs-\ys)+\Omega(\xts-\yts)- \mu \rb
e^{-2\imath\Ot
t\gamma^{0}\gamma^{1}}e^{-t\mu^{2}} \; . \nonumber
\end{eqnarray}
It will be useful to express $G$ in terms of commutators:
\begin{eqnarray}    
G(x,y)  &=&-\frac{\Omega}{\theta\pi}\int_{0}^{\infty}dt\,\lb \imath\Ot\coth(2\Ot
t)\lsb\xs, \Gamma^t  \rsb(x,y) \right.
\nonumber\\
&&
\left. +\Omega\lsb\xts, \Gamma^t \rsb(x,y)  -\mu \Gamma^t (x,y)  \rb
e^{-2\imath\Ot t\gamma^{0}\gamma^{1}}e^{-t\mu^{2}}, 
\label{xfullprop}
\end{eqnarray}
where
\begin{eqnarray}
\Gamma^t (x,y)  &=&
\frac{1}{\sinh(2\Ot t)}\,
e^{-\frac{\Ot}{2}\coth(2\Ot t)(x-y)^{2}+\imath\Ot x\wedge y}
\end{eqnarray}
with $\Ot=\frac{2\Omega}{\theta}$ and $x\wedge y=x^{0}y^{1}-x^{1}y^{0}$.\\

We now give the expression of the Fermionic kernel \eqref{xfullprop} in the matrix basis. The inverse of 
the quadratic form
\begin{equation}
\Delta=p^{2}+\mu^{2}+\frac{4\Omega^{2}}{\theta^2} x ^{2} +\frac{4B}{\theta}L_{2}
\end{equation}
is given by \eqref{eq:propbosonmatrix} in the preceeding section:
\begin{align}
  \Gamma_{m, m+h; l + h, l} 
  &= \frac{\theta}{8\Omega} \int_0^1 d\alpha\,  
  \dfrac{(1-\alpha)^{\frac{\mu^2 \theta}{8 \Omega}-\frac{1}{2}}}{  
    (1 + C\alpha )} 
  \Gamma^{\alpha}_{m, m+h; l + h, l}\,\label{eq:propinit}
\\
  \Gamma^{(\alpha)}_{m, m+h; l + h, l}
  &= \left(\frac{\sqrt{1-\alpha}}{1+C \alpha} 
  \right)^{m+l+h}\left( 1-\alpha\right)^{-\frac{Bh}{2\Omega}} \label{eq:propinit-b}\\
  &
  \sum_{u=0}^{\min(m,l)} {\cal A}(m,l,h,u)\ 
  \left( \frac{C \alpha (1+\Omega)}{\sqrt{1-\alpha}\,(1-\Omega)} 
  \right)^{m+l-2u}.\notag
\end{align}
The Fermionic propagator $G$ (\ref{xfullprop}) in the matrix basis may be deduced from the kernel \eqref
{eq:propinit}. We just set $B= \Omega$, add the missing term with $\gamma^0 \gamma^1$ and compute 
the action of $-\ps-\Omega\xts+\mu$ on $\Gamma$. We must then evaluate $\lsb x^{\nu},\Gamma\rsb$ in 
the matrix basis:
\begin{align}
  \lsb x^{0},\Gamma\rsb_{m,n;k,l}=&2\pi\theta\sqrt\frac{\theta}{8}\lb\sqrt{m+1}
  \Gamma_{m+1,n;k,l}-\sqrt{l}\Gamma_{m,n;k,l-1}+\sqrt{m}\Gamma_{m-1,n;k,l}
\right.\nonumber\\
&-\sqrt{l+1}\Gamma_{m,n;k,l+1}+\sqrt{n+1}\Gamma_{m,n+1;k,l}-\sqrt{k}
\Gamma_{m,n;k-1,l}\nonumber\\
&\left.+\sqrt{n}\Gamma_{m,n-1;k,l}-\sqrt{k+1}
  \Gamma_{m,n;k+1,l}\rb,\label{x0Gamma}\\
  \lsb
  x^{1},\Gamma\rsb_{m,n;k,l}=&2\imath\pi\theta\sqrt\frac{\theta}{8}\lb\sqrt{m+1}
  \Gamma_{m+1,n;k,l}-\sqrt{l}\Gamma_{m,n;k,l-1}-\sqrt{m}
  \Gamma_{m-1,n;k,l} \right.
\nonumber\\
&+\sqrt{l+1}\Gamma_{m,n;k,l+1}
-\sqrt{n+1}\Gamma_{m,n+1;k,l}+\sqrt{k}\Gamma_{m,n;k-1,l}
\nonumber\\
&\left.+\sqrt{n}\Gamma_{m,n-1;k,l}-\sqrt{k+1}\Gamma_{m,n;k+1,l}\rb.
  \label{x1Gamma}
\end{align}
This allows to prove:
\begin{lemma}
Let $G_{m,n;k,l}$ the kernel, in the matrix basis, of the operator\\
$\lbt\ps+\Omega\xts+\mu\rbt^{-1}$. We have:
\begin{align}
G_{m,n;k,l}=& 
-\frac{2\Omega}{\theta^{2}\pi^{2}} \int_{0}^{1} 
d\alpha\, G^{\alpha}_{m,n;k,l},\label{eq:propaFermiomatrix}
\\
G^{\alpha}_{m,n;k,l}=&\lbt\imath\Ot\frac{2-\alpha}{\alpha}\lsb\xs,
\Gamma^{\alpha}\rsb_{m,n;k,l}
+\Omega\lsb\slashed{\tilde{x}},\Gamma^{\alpha}\rsb_{m,n;k,l} - \mu\,\Gamma^{\alpha}_{m,n;k,l}\rbt
\nonumber\\
&\times\lbt\frac{2-\alpha}{2\sqrt{1-\alpha}}
\mathds{1}_{2}-\imath\frac{\alpha}{2\sqrt{1-\alpha}}\gamma^{0}\gamma^{1}
\rbt.\label{eq:matrixfullprop}
\end{align}
where $\Gamma^{\alpha}$ is given by (\ref{eq:propinit-b}) and the commutators by the formulas (\ref
{x0Gamma}) and (\ref{x1Gamma}).
\end{lemma}
The first two terms in the equation (\ref{eq:matrixfullprop}) contain commutators and will be gathered 
under the name $G^{\alpha, {\rm comm}}_{m,n;k,l}$. The last term will be called $G^{\alpha, {\rm mass}}_
{m,n;k,l}$:
\begin{align}\label{commterm}
G^{\alpha, {\rm comm}}_{m,n;k,l}=& \lbt\imath\Ot\frac{2-\alpha}{\alpha}\lsb\xs,
\Gamma^{\alpha}\rsb_{m,n;k,l} +\Omega\lsb\slashed{\tilde{x}},\Gamma^{\alpha}\rsb_{m,n;k,l} \rbt  
\nonumber\\
&\times\lbt\frac{2-\alpha}{2\sqrt{1-\alpha}}
\mathds{1}_{2}-\imath\frac{\alpha}{2\sqrt{1-\alpha}}\gamma^{0}\gamma^{1} \rbt,\\
\notag\\
G^{\alpha, {\rm mass}}_{m,n;k,l}=& - \mu\, \Gamma^{\alpha}_{m,n;k,l}
\times\lbt\frac{2-\alpha}{2\sqrt{1-\alpha}}
\mathds{1}_{2}-\imath\frac{\alpha}{2\sqrt{1-\alpha}}\gamma^{0}\gamma^{1} \rbt.\label{massterm}
\end{align}

\subsection{Bounds}
\label{sec:bornes}

We use the multi-scale analysis to study the behavior of the propagator \eqref{eq:matrixfullprop} and 
revisit more finely the bounds \eqref{th1} to \eqref{thsummax}. In a slice $i$, the propagator is
\begin{align}
  \Gamma^i_{m,m+h,l+h,l} 
  &=\frac{\theta}{8\Omega}  \int_{M^{-2i}}^{M^{-2(i-1)}} d\alpha\; 
  \dfrac{(1-\alpha)^{\frac{\mu_0^2 \theta}{8 \Omega}-\frac{1}{2}}}{  
    (1 + C\alpha )} 
  \Gamma^{(\alpha)}_{m, m+h; l + h, l}\;.
  \label{prop-slice-i}
\end{align}
\begin{eqnarray}
G_{m,n;k,l}&=& \sum_{i=1}^\infty G^i_{m,n;k,l} \ ; \ G^i_{m,n;k,l} = 
-\frac{2\Omega}{\theta^{2}\pi^{2}} \int_{M^{-2i}}^{M^{-2(i-1)}} 
d\alpha\, G^{\alpha}_{m,n;k,l}  \ .
\label{eq:matrixfullpropsliced}
\end{eqnarray}
Let $h= n-m$ and $p=l-m$. Without loss of generality, we assume $h \ges 0 $ and $p\ges 0$. Then the 
smallest index among $m,n,k,l$ is $m$ and the biggest is $k=m+h+p$. We have:
\begin{thm}\label{maintheorem}
Under the assumptions $h =n-m\ges 0$ and $p=l-m \ges 0$, there exists $K,c\in\R_{+}$ ($c$ depends 
on $\Omega$) such that the propagator of the \encv{} Gross-Neveu model in a slice $i$ obeys the bound
\begin{eqnarray}\label{mainbound1}  
\vert G^{i,{\rm comm}}_{m,n;k,l}\vert&\les&   
K M^{-i} \bigg( \chi(\alpha k>1)\frac{\exp \{- \frac{c p ^2  }{1+ kM^{-2i}}
- \frac{ c M^{-2i}}{1+k} (h - \frac{k}{1+C})^2 \}}{(1+\sqrt{ kM^{-2i}}) }  
\nonumber\\
&&+ \min(1,(\alpha k)^{p})e^{- c k M^{-2i} - c  p }\bigg).
\end{eqnarray}
The mass term is slightly different:
\begin{align} \label{mainbound2}  
\vert  G^{i,{\rm mass}}_{m,n;k,l}\vert\les&   
K M^{-2i} \bigg( \chi(\alpha k>1) \frac{\exp \{- \frac{c p ^2  }{1+ kM^{-2i}}
- \frac{ c M^{-2i}}{1+k} (h - \frac{k}{1+C})^2 \}}{1+\sqrt{ kM^{-2i}}}\notag
\\
&+\min(1,(\alpha k)^{p}) e^{- c k M^{-2i} - c  p }\bigg).
\end{align}
\end{thm}
\begin{rem}
  We can redo the same analysis for the $\Phi^{4}$ propagator and get
  \begin{equation}\label{eq:boundphi4}  
    G^i_{m,n;k,l}\les K M^{-2i}\min\lbt 1,(\alpha k)^{p}\rbt e^{-c(M^{-2i}k+p)}
  \end{equation}
which allows to recover the bounds \eqref{th1} to \eqref{thsummax}.
\end{rem}

\subsection{Propagators and renormalizability}
\label{sec:prop-et-renorm}

Let us consider the propagator \eqref{eq:propaFermiomatrix} of the \encv{} Gross-Neveu model. We saw 
in section \ref{sec:bornes} that there exists two regions in the space of indices where the propagator 
behaves very differently. In one of them it behaves as the $\Phi^{4}$ propagator and leads then to the 
same power counting. In the critical region, we have
\begin{align}
  G^{i}\les&K\frac{M^{-i}}{1+\sqrt{ kM^{-2i}}}\,e^{- \frac{c p ^2  }{1+ kM^{-2i}}
    -\frac{ c M^{-2i}}{1+k} (h - \frac{k}{1+C})^2}.
\end{align}
The point is that such a propagator does not allow to sum two reference indices with a unique line. This 
fact was useful in the proof of the power counting of the $\Phi^{4}$ model. This leads to a \emph
{renormalizable} UV/IR mixing.

Let us consider the graph in figure \ref{fig:sunsetj} where the two external lines bear an index $i\gg 1$ 
and the internal one an index $j<i$. The propagator \eqref{eq:propaFermiomatrix} obeys the bound in 
Prop.~\eqref{thsum} which means that it is almost local. We only have to sum over one index per internal 
face.
\begin{figure}[!htbp]
  \begin{center}
    \subfloat[At scale $i$]{{\label{fig:sunseti}}\includegraphics[scale=.7]{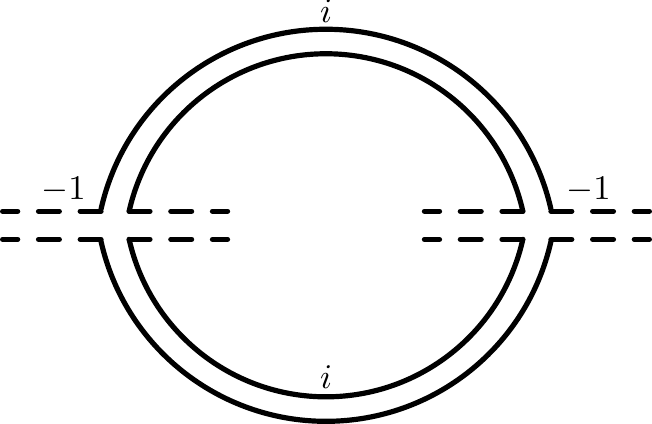}}\qquad
    \subfloat[At scale $j$]{{\label{fig:sunsetj}}\includegraphics[scale=.7]{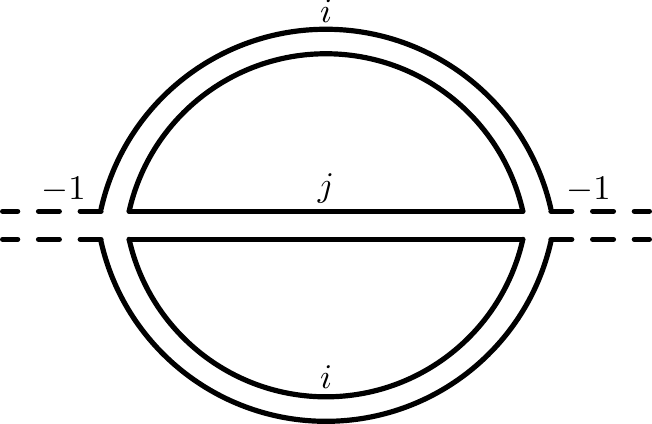}}
  \end{center}
  \caption{Sunset Graph}
  \label{figsunset}
\end{figure}

On the graph of the figure \ref{fig:sunseti}, if the two lines inside are true external ones, the graph has 
two broken faces and there is no index to sum over. Then by using Prop.~\eqref{th1} we get $A_{G}\les 
M^{-2i}$. The sum over $i$ converges and we have the same behavior as the $\Phi^{4}$ theory, that is 
to say the graphs with $B\ges 2$ broken faces are finite. But if these two lines belongs to a line of scale 
$j<i$ (see figure \ref{fig:sunsetj}), the result is different. Indeed, at scale $i$, we recover the graph of 
figure \ref{fig:sunseti}. To maintain the previous result ($M^{-2i}$), we should sum the two indices 
corresponding to the internal faces with the propagator of scale $j$. This is not possible. Instead we 
have:
\begin{align}
\sum_{k,h}M^{-2i-j}\,e^{-M^{-2i}k}\frac{e^{-\frac{ c M^{-2j}}{1+k} 
(h - \frac{k}{1+C})^2}}{1+\sqrt{ kM^{-2j}}}\les KM^{j}.
\end{align}
The sum over $i$ diverges logarithmically. The graph of figure \ref{fig:sunseti} converges if it is linked to 
true external legs et diverges if it is a subgraph of a graph at a lower scale. The power counting depends 
on the scales lower than the lowest scale of the graph. It can't then be factorized into the connected 
components: this is UV/IR mixing.\\

Let's remark that the graph of figure \ref{fig:sunseti} is not renormalizable by a counter-term in the 
Lagrangian. Its logarithmic divergence can't be absorbed in a redefinition of a coupling constant. 
Fortunately the renormalization of the two-point graph of figure \ref{fig:sunsetj} makes the four-point 
subdivergence finite \cite{RenNCGN05}. This makes the \encv{} Gross-Neveu model renormalizable.

\section{Direct space}
\label{sec:direct-space}

We want now to explain how the power counting analysis
can be performed  in direct space, and the ``Moyality'' of the necessary counterterms
can be checked by a Taylor expansion which is a generalization of the one used in direct commutative 
space.

In the commutative case there is translation invariance, hence each propagator depends on a single
difference variable which is short in the ultraviolet regime; in the non-commutative case the propagator
depends both of the difference of end positions, which is again short in the uv regime, but also of the 
sum which is long in the uv regime, considering the explicit form (\ref{eq:Mehler}) of the Mehler kernel.

This distinction between short and long variables is at the basis of the  power counting analysis
in direct space.

\subsection{Short and long variables}

Let $G$ be an arbitrary connected graph. 
The amplitude associated with this graph is in direct space
(with hopefully self-explaining notations):
\begin{align}
A_G=&\int \prod_{v,i=1,...4} dx_{v,i} \prod_l dt_l     \\
& \prod_v \left[ \delta(x_{v,1}-x_{v,2}+x_{v,3}-x_{v,4})e^{\imath
\sum_{i<j}(-1)^{i+j+1}x_{v,i}\theta^{-1} x_{v,j}} \right] \prod_l C_l \; , 
\nonumber  \\
C_l=&  
\frac{\Omega^2}{[2\pi\sinh(\Omega t_l)]^2}e^{-\frac{\Omega}{2}\coth(\Omega 
t_l)(x_{v,i(l)}^{2}+x_{v',i'(l)}^{2})
+\frac{\Omega}{\sinh(\Omega t_l)}x_{v,i(l)} . x_{v',i'(l)}   - \mu_0^2 t_l}\; .\nonumber
\label{amplitude1}
\end{align} 

For each line $l$ of the graph joining positions $x_{v,i(l)}$ and $x_{v',i'(l)}$, 
we choose an orientation and we define 
the  ``short'' variable $u_l=x_{v,i(l)}-x_{v',i'(l)}$ and the 
``long'' variable $v_l=x_{v,i(l)}+x_{v',i'(l)}$.

With these notations, defining $\Omega t_l=\alpha_l$, the propagators in our graph can be 
written as:
\begin{equation}
\int_{0}^{\infty} \prod_l \frac{\Omega d\alpha_l}{[2\pi\sinh(\alpha_l)]^2}
e^{-\frac{\Omega}{4}\coth(\frac{\alpha_l}{2})
{ u_l^2}- \frac{\Omega}{4}\tanh(\frac{\alpha_l}{2})
{ v_l^2}  - \frac{\mu_0^2}{\Omega} \alpha_l}\; .\label{tanhyp}
\end{equation} 

As in matrix space we  can slice each propagator according to the size of its $\alpha$ parameter
and obtain the multiscale representation of each Feynman amplitude:

\begin{align}
A_G=& \sum_{\mu}  A_{G,\mu}\quad, \quad A_{G,\mu} = \int \prod_{v,i=1,...4} dx_{v,i} \prod_l C_{l}^{i_
{\mu}(l)}(u_{l},v_{l})\label{amplitude2}\\ 
&\prod_v \left[ \delta(x_{v,1}-x_{v,2}+x_{v,3}-x_{v,4})e^{\imath
\sum_{i<j}(-1)^{i+j+1}x_{v,i}\theta^{-1} x_{v,j}} \right]\nonumber\\
 C^i (u,v) =& \int_{M^{-2i}}^{M^{-2(i-1)}} 
\frac{\Omega d\alpha}{[2\pi\sinh(\alpha)]^2}
e^{-\frac{\Omega}{4}\coth(\frac{\alpha}{2})
{ u^2}- \frac{\Omega}{4}\tanh(\frac{\alpha}{2})
{ v^2}  - \frac{\mu_0^2}{\Omega} \alpha}\; ,
\end{align} 
where $\mu$ runs over scales attributions $\{i_{\mu}(l) \}$ for each line $l$ of the graph, 
and the sliced propagator $C^i$ in slice $i \in {\mathbb N}$
obeys the crude bound:
\begin{lemma} For some constants $K$ (large) and $c$ (small):
\begin{equation}\label{eq:propbound-phi4}
C^i (u,v) \les K M^{2i}e^{-c [ M^{i}\Vert u \Vert + M^{-i}\Vert v\Vert ] }
\end{equation} 
(which a posteriori justifies the terminology of ``long'' and ```short'' variables).
\end{lemma}

The proof is elementary.

\subsection{Routing, Filk moves}
\label{sec:routing-filk-moves}
\subsubsection{Oriented graphs}

We pick a tree $T$ of lines of the graph, hence connecting all vertices, pick with a root vertex and build 
an \textit{orientation} of all the lines of the graph in an inductive way. Starting from an arbitrary 
orientation of a field at the 
root of the tree, we climb in the tree and at each vertex of the tree 
we impose cyclic order to alternate entering and exiting tree lines and loop half-lines, as in figure \ref
{otree}.
\begin{figure}[!htb]
  \centering
  \subfloat[Orientation of a tree]{\label{otree}\includegraphics[scale=0.7]{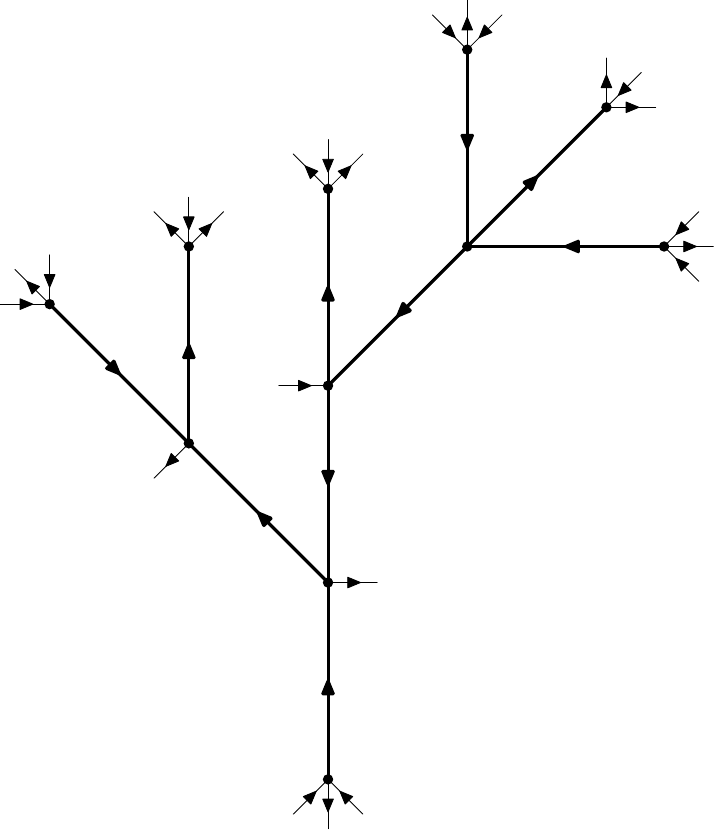}}\qquad
  \subfloat[A non-orientable graph]{\label{nono}\includegraphics[scale=0.6]{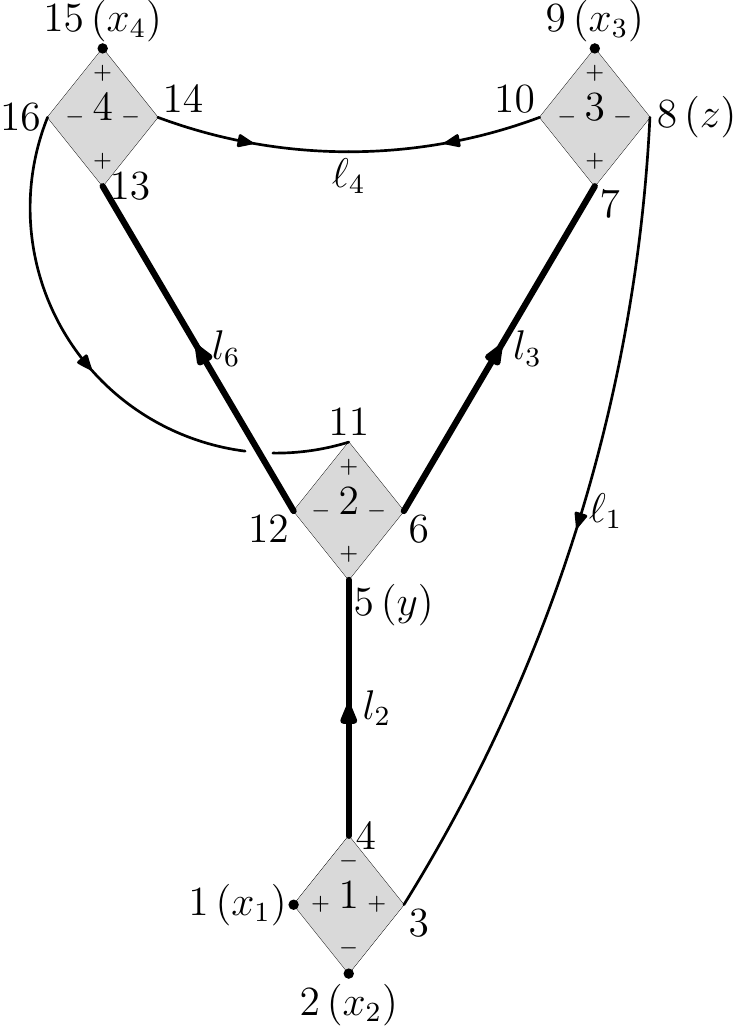}}
  \caption{Orientation}
  \label{fig:orientation}
\end{figure}
Then we look at the loop lines. If every loop lines consist in the contraction of an 
entering and an exiting line, the graph is called orientable. Otherwise we call it non-orientable
as in figure \ref{nono}.

\subsubsection{Position routing}

There are $n$ $\delta$ functions in an amplitude with $n$ vertices,
hence $n$ linear equations for the $4n$ positions,
one for each vertex. The {\it position routing} associated to the tree $T$ 
solves this system by passing to another equivalent system of $n$ linear equations,
one for each branch of the tree. This is a triangular change of variables,
of Jacobian 1. This equivalent system is obtained by summing the arguments of the 
$\delta$ functions of the vertices in each branch. This change of variables is exactly the 
$x$-space analog of the resolution of momentum conservation called \textit{momentum routing}
in the standard physics literature of commutative field theory, except that 
one should now take care of the additional $\pm$ cyclic signs.

One can prove \cite{xphi4-05} that the rank of the system of $\delta$ functions 
in an amplitude with $n$ vertices is
\begin{itemize}
\item $n-1$ if the graph is orientable
\item $n$ if the graph is non-orientable
\end{itemize}

The position routing change of variables is summarized by the following lemma:

\begin{lemma}[Position Routing]
We have, calling $I_G$ the remaining integrand in (\ref{amplitude2}):
\begin{align}
A_G =& \int \Big[ \prod_v  \big[ \delta(x_{v,1}-x_{v,2}+x_{v,3}-x_{v,4})\big] \, \Big]\;
I_G(\{x_{v,i}  \}  )   \\
=& \int \prod_{b}
\delta \left(   \sum_{l\in T_b \cup L_b }u_{l} + \sum_{l\in L_{b,+}}v_{l}-\sum_{l\in L_{b,-}}v_{l}
+\sum_{f\in X_b}\epsilon(f) x_f \right) I_G(\{x_{v,i}  \}), \nonumber 
\end{align} 
where $\epsilon(f)$ is $\pm 1$ depending on whether the field $f$ enters or exits the branch.
\end{lemma}

We can now use the system of delta functions to eliminate variables. It is of course better to eliminate
long variables as their integration costs a factor $M^{4i}$ whereas the integration of a short
variable brings $M^{-4i}$. Rough power counting, neglecting all oscillations of the vertices
leads therefore, in the case of an orientable graph with $N$ external fields, $n$ internal vertices
and $l= 2n - N/2$ internal lines at scale $i$ to:
\begin{itemize}
\item a factor $M^{2i(2n - N/2)}$ coming from the $M^{2i}$ factors for each line of scale $i$
  in (\ref{eq:propbound-phi4}),
\item a factor $M^{-4i(2n - N/2)}$ for  the $l = 2n - N/2$ short variables integrations,
\item a factor $M^{4i (n - N/2 +1)}$ for the long variables after eliminating $n-1$ of them 
using the delta functions.
\end{itemize}
The total factor is therefore $M^{-(N-4)i}$, the ordinary scaling of $\phi^4_4$, which means that
only two and four point subgraphs ($N \les 4)$ diverge when $i$ has to be summed.

In the non-orientable case, we can eliminate one additional long variable since the rank of
the system of delta functions is larger by one unit! Therefore we get a power counting bound
$M^{-Ni}$, which proves that  only {\it orientable} graphs may diverge.

In fact we of course know that not all {\it orientable}  two and four point subgraphs diverge
but only the planar ones with a single external face. 
(It is easy to check that all such planar graphs are indeed orientable).

Since only these planar subgraphs with a single external face can be renormalized by Moyal 
counterterms, we need to prove that orientable, non-planar graphs or orientable planar graphs with 
several external faces have in fact a better power  counting than this crude estimate. 
This can be done only by exploiting their vertices oscillations. 
We explain now how to do this with minimal effort.

\subsubsection{Filk moves and rosettes}

Following Filk \cite{Filk1996dm}, we can contract all lines of a spanning tree $T$
and reduce $G$ to a single vertex with ``tadpole loops'' called a ``rosette graph''.
This rosette is a cycle (which is the border of the former tree) bearing loops lines on it (see figure \ref
{fig:rosette}):
\begin{figure}[!htb]
  \centering
    \includegraphics[scale=0.4]{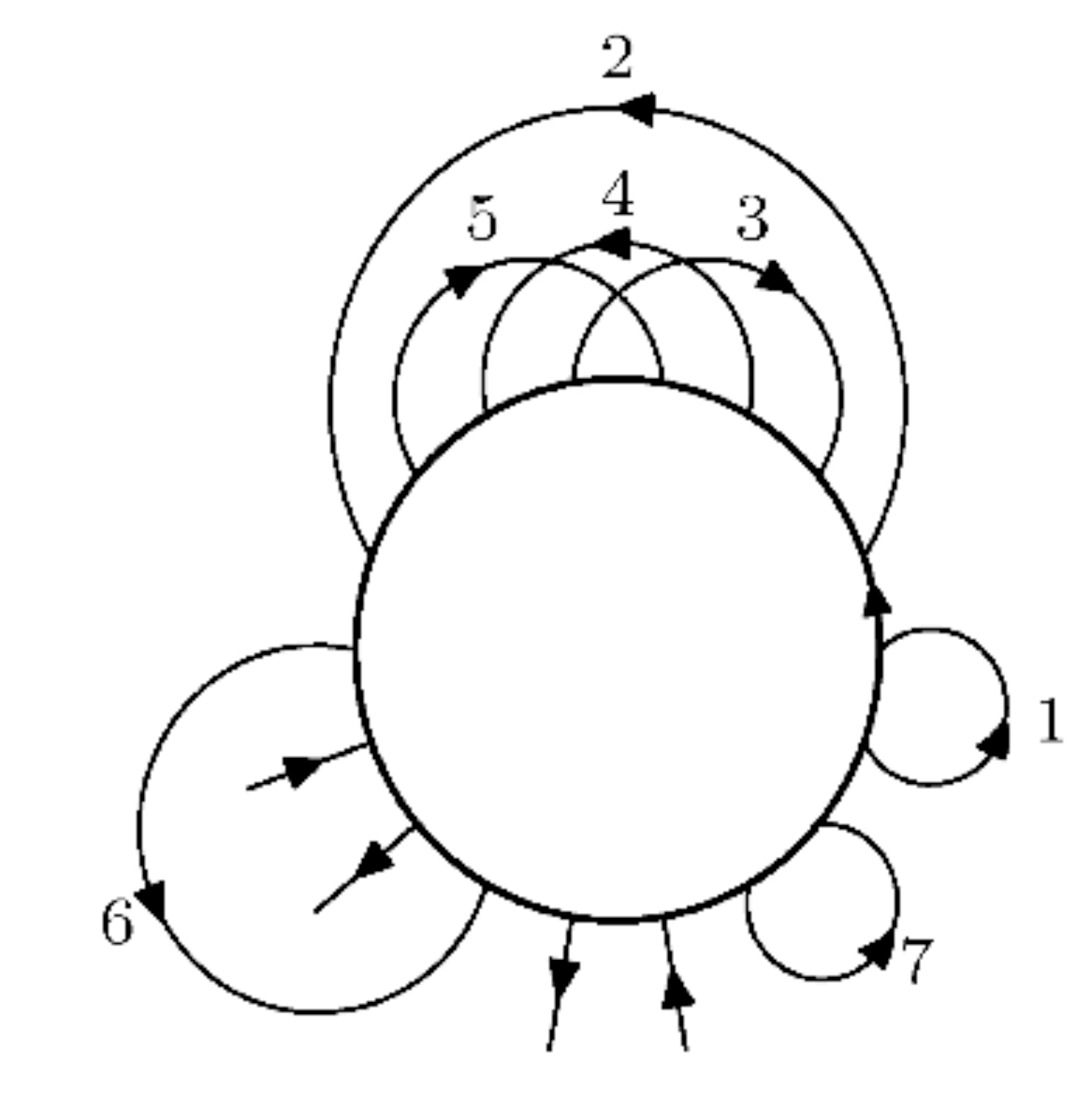}
    \caption{A rosette}
    \label{fig:rosette}
\end{figure}
Remark that the rosette can also be considered 
as a big vertex, with $r=2n+2$ fields, on which $N$
are external fields with external variables $x$ and $2n+2-N$ are loop fields for the corresponding
$n+1-N/2$ loops. When the graph is orientable, the rosette is also orientable,
which means that turning around the rosette the lines alternatively enter and exit.
These lines correspond to the contraction of the fields on the border of the tree $T$
before the Filk contraction, also called the ``first Filk move''.

\subsubsection{Rosette factor}

We start from the root and turn around the tree in the trigonometrical sense. We number
separately all the fields as $1,\dots,2n+2$ and all 
the tree lines as $1,\dots,n-1$ in the order they are met.
\begin{lemma}
The rosette contribution after a complete first Filk reduction is exactly:
\begin{equation}
\delta(v_1-v_2+\dots-v_{2n+2}+\sum_{l\in T}u_l)
e^{i VQV + iURU + iUSV}
\end{equation}
where the $v$ variables are the long or external variables
of the rosette, counted with their signs,
and the quadratic oscillations for these variables is
\begin{equation}
VQV= \sum_{0\les i<j\les r}(-1)^{i+j+1}v_i\theta^{-1} v_j 
\end{equation}
\end{lemma}

We have now to analyze in detail this quadratic oscillation of the remaining long loop variables 
since it is essential to improve power counting. We can neglect the
secondary oscillations $URU$ and $USV$ which imply short variables.

The second Filk reduction \cite{Filk1996dm} further simplifies the rosette factor by erasing the
loops of the rosette which do not cross any other loops or arch over external fields. 
It can be shown that the loops which disappear in this operation correspond to those 
long variables who do not appear in the quadratic form $Q$.

Using the remaining {\it oscillating factors}
one can prove that non-planar graphs with genus larger than one or with more than
one external face {\it do not diverge}.

The basic mechanism to improve the
power counting of a single non-planar subgraph is the following:
\begin{align}\label{gainoscill}
&\int dw_1dw_2 e^{-M^{-2i_1}w_1^2-M^{-2i_2}w_2^2
- iw_1\theta^{-1}w_2+w_1 . E_1(x,u)+w_2 E_2(x,u)}
\nonumber\\
=& \int dw'_1dw'_2 e^{-M^{-2i_1}(w_1')^2
-M^{-2i_2}(w'_2)^2 +iw'_1\theta^{-1}w'_2 + (u,x)Q(u,x)}
\nonumber\\
=&  K  M^{4i_1} \int dw'_2
e^{- (M^{2i_1}+ M^{-2i_2})(w'_2)^2 }=
K M^{4i_1}M^{-4i_2} \; .
\end{align} 
In these equations we used for simplicity $M^{-2i}$ 
instead of the correct but more complicated factor $(\Omega /4) \tanh (\alpha /2 )$
(of course this does not change the argument) and we performed
a unitary linear change of variables $w'_1 = w_1 + \ell_1 (x, u)$, $w'_2 = w_2 + \ell_2 (x, u)$
to compute the oscillating $w'_1$ integral. The gain in (\ref{gainoscill}) is
$M^{-8i_2}$, which is the difference between $M^{-4i_2}$ and 
the normal factor $M^{4i_2}$ that the $w_2$ integral would have cost if
we had done it with the regular $e^{-M^{-2i_2}w_2^2}$ factor for long variables. 
To maximize this gain we can assume $i_1 \les i_2$.

This basic argument must then be generalized to each non-planar 
subgraph in the multiscale analysis, which is possible.

Finally it remains to consider the case of subgraphs which are planar orientable
but with more than one external face. In that case there are no crossing loops in the rosette but
there must be at least one loop line arching over a non trivial subset
of external legs (see e.g. line $6$ in figure \ref{fig:rosette}). We have then a non trivial integration over at 
least one external variable, called $x$, of at least one long loop variable called $w$. This ``external'' $x$ 
variable without the oscillation improvement 
would be integrated with a test function of scale 1 (if it is a true external line of scale $1$)
or better (if it is a higher long loop variable)\footnote{Since the loop line arches 
over a non trivial (i.e. neither full nor empty) subset
of external legs of the rosette, the variable $x$ cannot be the full combination 
of external variables in the ``root'' $\delta$ function.}. But we get now
\begin{align}\label{gainoscillb}
&\int dx dw e^{-M^{-2i}w^2
- iw\theta^{-1}x  +w.E_1(x',u)}
\nonumber\\
=&  K  M^{4i} \int dx 
e^{-M^{+2i} x^2 }=
K' \ ,
\end{align} 
so that a factor $M^{4i}$ in the former bound becomes $\cO(1)$ hence is improved by $M^{-4i}$.

In this way we can reduce the convergence of the multiscale analysis to the
problem of renormalization of planar two- and four-point subgraphs 
with a single external face, which we treat in the next section.

Remark that the power counting obtained in this way is still not optimal. To get the same level 
of precision than with the matrix base requires e.g. to display $g$ independent improvements
of the type (\ref{gainoscill}) for a graph of genus $g$. This is doable but basically requires
a reduction of the quadratic form $Q$ for single-faced rosette (also called ``hyperrosette'')
into $g$ standard symplectic blocks
through the so-called ``third Filk move'' introduced in \cite{gurauhypersyman}.
We return to this question in section \ref{hyperbo}.

\subsection{Renormalization}

\subsubsection{Four-point function}

Consider the amplitude of a four-point graph $G$ which in the multiscale expansion
has all its internal scales higher than its four external scales.

The idea is that one should compare its amplitude to a similar amplitude with a
``Moyal factor'' $\exp\Big(
2\imath\theta^{-1}\lbt x_{1}\wedge x_{2}+x_{3}\wedge x_{4}\rbt\Big)\delta(\Delta)$
factorized in front, where $\Delta= x_1- x_{2}+x_{3} - x_{4}$.
But precisely because the graph is planar with a single external face
we understand that the external positions $x$ only couple to \textit{short variables} $U$ 
of the internal amplitudes through the 
global delta function and the oscillations. Hence we can break this coupling 
by a systematic Taylor expansion to first order. This separates a piece proportional to 
``Moyal factor'', then absorbed into the effective coupling constant, and a remainder 
which has at least one additional small factor which gives him improved power counting.

This is done by expressing the amplitude for a graph with $N=4$, $g = 0$ and $B = 1$ as:
\begin{align}
A(G)(x_1, x_2, x_3, x_4) =&\int  {\exp\Big(
2\imath\theta^{-1}\lbt x_{1}\wedge x_{2}+x_{3}\wedge x_{4}\rbt\Big)}
\prod_{\ell\in T^{i}_{k} }  du_{\ell} \ C_{\ell}(u_\ell, U_\ell, V_\ell) 
\nonumber\\
&
\bigg[ \prod_{l \in G^{i}_{k} \, \ l \not \in T}  du_{l}  d v_{l} C_{l}(u_l, v_l) \bigg]
\ e^{\imath URU+\imath USV}
\\
&\hspace{-3cm}
\Bigg\{ {\delta(\Delta)}
+ \int_{0}^{1}dt\bigg[ \mathfrak{U}\cdot \nabla \delta(\Delta+t\mathfrak{U})
+\delta(\Delta+t\mathfrak{U}) [\imath XQU  + {\mathfrak R}' (t)]  \bigg] 
e^{\imath tXQU + {\mathfrak R}(t)}  \Bigg\} \ .\nonumber
\end{align}
where
$C_{\ell}(u_\ell, U_\ell, V_\ell) $ is the propagator taken at $X_\ell=0$, $\mathfrak{U}= \sum_\ell u_\ell$
and ${\mathfrak R}(t)$ is a correcting term involving $\tanh \alpha_\ell [X.X + X.(U+V)]$.

The {first term} is of the initial $\int Tr \phi \star \phi \star \phi \star \phi $ form.
The rest no longer diverges, since the $U$ and ${\mathfrak R}$
provide the necessary small factors.

\subsubsection{Two-point function}

Following the same strategy we have to Taylor-expand the coupling
between external variables and $U$ factors in two point planar graphs with a single external face 
to \textit{third order} and some non-trivial symmetrization of the terms
according to the two external arguments to cancel some odd contributions.
The corresponding factorized relevant and marginal contributions 
can be then shown to give rise only to

\begin{itemize}
\item A mass counterterm,
\item A wave function counterterm,
\item An {harmonic potential counterterm}.
\end{itemize}
and the remainder has convergent power counting.
This concludes the construction of the effective expansion in this direct space multiscale analysis.

Again the BPHZ theorem itself for the renormalized expansion follows by developing the
counterterms still hidden in the effective couplings and its finiteness
follows from the standard classification of forests. 
See however the remarks 
at the end of section \ref{sec:vari-indep}.

Since the bound (\ref{eq:propbound-phi4}) works for any $\Omega \ne 0$,
an additional bonus of the $x$-space
method is that it proves renormalizability of the model for any $\Omega$ in $]0,1]$\footnote{The case $
\Omega$ in $[1,+\infty[$ is irrelevant since it can be rewritten by LS duality as an equivalent
model with $\Omega$ in $]0,1]$.},
whether the matrix method proved it only for $\Omega$ in $]0.5,1]$.

\subsubsection{The Langmann-Szabo-Zarembo model}  

It is a four-dimensional theory of a Bosonic complex field defined by the action
\begin{align}
S=&\int \frac{1}{2} \bar \phi (- D^\mu D_\mu + \Omega^2 x^2  )\phi + \lambda 
\bar \phi \star \phi\star\bar \phi \star \phi
\end{align}
where $D^\mu= \imath\partial_\mu+B_{\mu \nu}x^\nu$ is the covariant derivative in a magnetic field $B$.

The interaction $\bar \phi \star \phi\star\bar \phi \star \phi$
ensures that perturbation theory contains only orientable graphs.
For $\Omega >0$ the $x$-space propagator still decays as in the ordinary $\phi^4_4$ case 
and the model has been shown renormalizable by an easy extension of the methods of the previous 
section \cite{xphi4-05}.

However at $\Omega =0$, there is no longer any harmonic potential in addition 
to the covariant derivatives and the bounds are lost. 
We call models in this category \emph{covariant}.

\subsubsection{Covariant models}  

Consider the $x$-kernel of the operator
\begin{align}
H^{-1}&=\lbt p^2 + \Omega^2\xt^2 - 2\imath B\lbt x^0p_1-x^1p_0 \rbt \rbt^{-1}\\
H^{-1}(x,y)&=\frac{\Ot}{8\pi}\int_0^\infty \frac{dt}{\sinh(2\Ot t)}\, \exp\lbt -\frac{\Ot}{2}\frac{\cosh(2Bt)}{\sinh(2
\Ot t)} (x-y)^2\right.\\
&\hspace{1.5cm}\left. -\frac{\Ot}{2}{\frac{\cosh(2\Ot t)-\cosh(2Bt)}{\sinh(2\Ot t)}(x^2+y^2)} \right.\\
&\hspace{1.5cm}\left. +2\imath\Ot{\frac{\sinh(2Bt)}{\sinh(2\Ot t)}x\wedge y} \rbt\text{\hspace{0.5cm}with }
\Ot=\frac{2\Omega}{\theta}
\end{align}
The Gross-Neveu model or the covariant Langmann-Szabo-Zarembo models correspond to 
the case $B=\Ot$.
In these models there is no longer any confining decay for the ``long variables''
but only an oscillation:
\begin{align}
Q^{-1}=H^{-1}&=\frac{\Ot}{8\pi}\int_0^\infty \frac{dt}{\sinh(2\Ot t)}\, \exp\lbt -\frac{\Ot}{2}\coth(2\Ot t)(x-y)^2 +
{2\imath\Ot x\wedge y}\rbt\label{eq:CriticalPropa}
\end{align}

The construction of these covariant models is more difficult, since
sufficiently many oscillations must be proven independent before power counting
can be established. The prototype paper which solved this problem is \cite{RenNCGN05},
which we briefly summarize now.

The main technical difficulty of the covariant models is the absence of decreasing functions for the long $v$ 
variables in the propagator replaced by an oscillation, see (\ref{eq:CriticalPropa}). Note that these 
decreasing functions are in principle created by integration over the $u$ variables\footnote{In all the 
following we restrict ourselves to the dimension $2$.}:
\begin{align}
  \int du\,e^{-\frac{\Ot}{2}\coth(2\Ot t)u^{2}+\imath u\wed v}=&K\tanh(2\Ot t)\,e^{-k\tanh(2\Ot t)v^{2}}.
\end{align}
But to perform all these Gaussian integrations for a general graph is a difficult task (see \cite{RivTan}) 
and is in fact not necessary for a BPHZ theorem. We can instead exploit the vertices and propagators 
oscillations to get rational decreasing functions in some linear combinations of the long $v$ variables. 
The difficulty is then to prove that all these linear combinations are independent and hence allow to 
integrate over all the $v$ variables. To solve this problem we need the exact expression of the total 
oscillation in terms of the short and long variables. This consists in a generalization of the Filk's work 
\cite{Filk1996dm}. This has been done in \cite{RenNCGN05}. Once the oscillations are proven 
independent, one can just use the same arguments than in the $\Phi^{4}$ case (see section \ref
{sec:routing-filk-moves}) to compute an upper bound for the power counting:
\begin{lemma}[Power counting $\GN$]\label{lem:compt-puissGN}
  Let $G$ a connected orientable graph. For all $\Omega\in\lsb 0,1\right)$, there exists $K\in\R_{+}$ such 
that its amputated amplitude $A_{G}$ integrated over test functions is bounded by
  \begin{align}
    \labs A_{G}\rabs\les&K^{n}M^{-\frac 12\omega(G)}\label{eq:compt-bound}\\
    \text{with } \omega(G)=&
    \begin{cases}
      N-4&\text{if ($N=2$ or $N\ges 6$) and $g=0$,}\\
      &\text{if $N=4$, $g=0$ and $B=1$,}\\
      &\text{if $G$ is critical,}\\
      N&\text{if $N=4$, $g=0$, $B=2$ and $G$ non-critical,}\\
      N+4&\text{if $g\ges 1$.}
    \end{cases}
  \end{align}
\end{lemma}
As in the non-commutative $\Phi^{4}$ case, only the planar graphs are divergent. But the behavior of 
the graphs with more than one broken face is different. Note that we already discussed such a feature in 
the matrix basis (see section \ref{sec:prop-et-renorm}). In the multiscale framework, the Feynman 
diagrams are endowed with a scale attribution which gives each line a scale index. The only subgraphs 
we meet in this setting have all their internal scales higher than their external ones. Then a subgraph $G
$ of scale $i$ is called \emph{critical} if it has $N=4, g=0, B=2$ and that the two ``external'' points in the 
second broken face are only linked by a single line of scale $j<i$. The typical example is the graph of 
figure \ref{fig:sunseti}. In this case, the subgraph is logarithmically divergent whereas it is convergent in 
the $\Phi^{4}$ model. Let us now show roughly how it happens in the case of figure \ref{fig:sunseti} but 
now in $x$-space.

The same arguments than in the $\Phi^{4}$ model prove that the integrations over the internal points of 
the graph \ref{fig:sunseti} lead to a logarithmic divergence which means that $A_{G^{i}}\simeq\cO(1)$ 
in the multiscale framework. But remind that there is a remaining oscillation between a long variable of 
this graph and the external points in the second broken face of the form $v\wed(x-y)$. But $v$ is of order 
$M^{i}$ which leads to a decreasing function implementing $x-y$ of order $M^{-i}$. If these points are 
true external ones, they are integrated over test functions of norm $1$. Then thanks to the additional 
decreasing function for $x-y$ we gain a factor $M^{-2i}$ which makes the graph convergent. But if $x$ 
and $y$ are linked by a single line of scale $j<i$ (as in figure \ref{fig:sunsetj}), instead of test functions 
we have a propagator between $x$ and $y$. This one behaves like (see \eqref{eq:CriticalPropa}):
\begin{align}
  C^{j}(x,y)\simeq&M^{j}\,e^{-M^{2j}(x-y)^{2}+\imath x\wed y}.  
\end{align}
The integration over $x-y$ instead of giving $M^{-2j}$ gives $M^{-2i}$ thanks to the oscillation $v\wed(x-
y)$. Then we have gained a good factor $M^{-2(i-j)}$. But the oscillation in the propagator $x\wed y$ 
now gives $x+y\simeq M^{2i}$ instead of $M^{2j}$ and the integration over $x+y$ cancels the 
preceeding gain. The critical component of figure \ref{fig:sunseti} is logarithmically divergent.

This kind of argument can be repeated and refined for more general graphs to prove that this problem 
appears only when the external points of the auxiliary broken faces are linked only by a \emph{single} 
lower line \cite{RenNCGN05}. This phenomenon can be seen as a mixing between scales. Indeed the 
power counting of a given subgraph now depends on the graphs at lower scales. This was not the case 
in the commutative realm. Fortunately this mixing doesn't prevent renormalization. Note that whereas the 
critical subgraphs are not renormalizable by a vertex-like counterterm, they are regularized by the 
renormalization of the two-point function at scale $j$. The proof of this point relies heavily on the fact that 
there is only one line of lower scale.

Let us conclude this section by mentioning the flows of the covariant models. One very interesting feature 
of the non-commutative $\Phi^{4}$ model is the boundedness of its flows and even the vanishing of its 
beta function for a special value of its bare parameters \cite
{GrWu04-2,DisertoriRivasseau2006,beta2-06}. Note that its commutative counterpart (the usual $\phi^
{4}$ model on $\R^{4}$) is asymptotically free in the infrared and has then an unbounded flow. It turns 
out that the flow of the covariant models are not regularized by the non-commutativity. The one-loop 
computation of the beta functions of the \encv{} Gross-Neveu model \cite{betaGN1Loop} shows that it is 
asymptotically free in the ultraviolet region as in the commutative case. 

\section{Parametric Representation}

\subsection{Ordinary Symanzik polynomials}

In ordinary commutative field theory, Symanzik's polynomials are obtained
after integration over internal position variables. 
The amplitude of an 
amputated graph $G$ with external momenta $p$ is, up to a normalization,
in space-time dimension $D$:
\begin{align}
A_G (p) =& \delta(\sum p)\int_0^{\infty} 
\frac{e^{- V_G(p,\alpha)/U_G (\alpha) }}{U_G (\alpha)^{D/2}} 
\prod_l  ( e^{-m^2 \alpha_l} d\alpha_l )\ .\label{symanzik} 
\end{align}
The first and second Symanzik polynomials $U_G$ and $V_G$ are
\begin{subequations}
  \begin{align}
    U_G =& \sum_T \prod_{l \not \in T} \alpha_l \ ,\label{symanzik1}\\
    V_G =& \sum_{T_2} \prod_{l \not \in T_2} \alpha_l  (\sum_{i \in E(T_2)} p_i)^2 \ , \label{symanzik2}
  \end{align}
\end{subequations}
where the first sum is over spanning trees $T$ of $G$
and the second sum  is over two trees $T_2$, i.e. forests separating the graph
in exactly two connected components $E(T_2)$ and $F(T_2)$; the corresponding
Euclidean invariant $ (\sum_{i \in E(T_2)} p_i)^2$ is, by momentum conservation, also
equal to $ (\sum_{i \in F(T_2)} p_i)^2$.

There are many interesting features in the parametric representation:

- It is more compact than direct or momentum space for dimension $D>2$, hence
it is adapted to numerical computations.

- The dimension $D$ appears now as a simple parameter. This allows to make it non integer or even complex,
at least in perturbation theory. This opens the road to the definition of dimensional
regularization and renormalization, which respect the symmetries of gauge theories.
This technique was the key to the first proof of the renormalizability
of non-Abelian gauge theories \cite{tHVe}.

- The form of the first and second Symanzik show an explicit \emph{positivity} and \emph{democracy}
between trees (or two-trees): each of them appears with positive and equal coefficients.

- The locality of the counterterms is still visible (although less obvious
than in direct space). It corresponds to the 
factorization of $U_G$ into $U_{S}U_{G/S}$ plus smaller terms
under scaling of all the parameters of a subgraph $S$, because the leading terms are the trees
whose restriction to $S$ are subtrees of $S$.
One could remark that this factorization also plays a key
role in the constructive RG analysis and multiscale bounds of the theory \cite{Riv1}.

In the next two subsections we shall derive the analogs of the
corresponding statements in NCVQFT. But before
that let us give a brief proof of formulas (\ref{symanzik}). The proof
of (\ref{symanzik2}) is similar.

Formula (\ref{symanzik}) is equivalent to the computation of the determinant, namely that
of the quadratic form gathering the heat kernels of all the internal lines in position space,
when we integrate over all vertices \emph{save one}. The role of this saved vertex
is crucial because otherwise the determinant of the quadratic form vanishes, i.e.
the computation becomes infinite by translation invariance.

But the same determinants and problems already arose a century before Feynman graphs
in the XIX century theory of electric circuits, where wires play the role\
of propagators and the conservation of currents at each node of the circuit play the role
of conservation of momenta or translation invariance. In fact the parametric representation follows from
the tree matrix theorem of Kirchoff \cite{Kirchoff}, 
which is a key result of combinatorial theory which in its simplest form
may be stated as:

\begin{thm}[Tree Matrix Theorem]\label{treematrix}
Let $A$ be an $n$ by $n$ matrix such that 
\be \label{sumnulle}
\sum_{i=1}^n A_{ij} = 0 \ \ \forall j \ . 
\ee
Obviously $\det A = 0$. But let $A^{11}$ be the matrix $A$ with line 1 and column 1 deleted.
Then 
\be \det A^{11} = \sum_{T} \prod_{\ell \in T} A_{i_{\ell},j_{\ell}} ,
\ee
where the sum runs over all directed trees on $\{1, ... , n \}$, directed away from root 1. 
\end{thm}
This theorem is a particular case of
a more general result that can compute any minor of a matrix
as a graphical sum over forests
and more \cite{Abdesselam}.

To deduce (\ref{symanzik}) from that theorem one defines $A_{ii}$ as the coordination of
the graph at vertex $i$ and $A_{ij}$ as $-l(ij)$ where $l(ij)$ is the number of lines from
vertex $i$c to vertex $j$. The line 1 and column 1 deleted correspond e.g. to fix
the first vertex 1 at the origin to break translation invariance.

We include now a proof of this Theorem using Grassmann variables
derived from \cite{Abdesselam}, because
this proof was essential for us to find the correct non commutative generalization
of the parametric representation. 
Recall that Grasmann variables anticommute
\be
\ch_i \ch_j +
\ch_j \ch_i = 0
\label{anti1}
\ee
hence in particular $\chi_i^2 = 0$, and that the Grassmann rules of integration are
\be
\int d\ch  = 0 \  ; \ 
\int \ch_d\ch  = 1.
\label{anti2}
\ee
Therefore we have:
\begin{lemma}
Consider a set of $2n$ independent Grasmann variables 
\be
{\Br\psi}_1, ... {\Br\psi}_n,  \psi_1 , ... , \psi_n
\ee
and the integration measure
\be
{\rm d}{\Br\psi} {\rm d}\psi =
{\rm d}{\Br\psi}_1, ... {\rm d}{\Br\psi}_n,  {\rm d}\psi_1 , ... , {\rm d}\psi_n
\ee
The bar is there for convenience, but it is not complex conjugation. Prove that
for any matrix $A$,
\be
\det A =
\int {\rm d}{\Br\psi} {\rm d}\psi
e^{-{\Br\psi}A\psi} \ .
\ee

More generally,
if $p$ is an integer $0\le p\le m$, and
$I=\{i_1,\ldots,i_p\}$,
$J=\{j_1,\ldots,j_p\}$ are two ordered subsets with $p$ elements
$i_1<\cdots<i_p$ and
$j_1<\cdots<j_p$, if also $A^{I,J}$
denotes the $(n-p)\times(n-p)$ matrix obtained
by erasing the rows of $A$ with index in $I$ and
the columns of $A$ with index in $J$, then
\be
\int {\rm d}{\Br\psi} {\rm d}\psi 
\ (\psi_J {\Br \psi}_I)
e^{-{\Br\psi}A\psi}
=
(-1)^{\Si I+\Si J}{\rm det}(A^{I,J})
\label{multicramer}
\ee
where $(\psi_J {\Br \psi}_I)\eqdef
\psi_{j_1}{\Br\psi}_{i_1}
\psi_{j_2}{\Br\psi}_{i_2}\ldots
\psi_{j_p}{\Br\psi}_{i_p}$,
$\Si I\eqdef i_1+\cdots+i_p$ and likewise for $\Si J$.
\label{exercisegrass}
\end{lemma}
We return now to

\noindent{\bf Proof of Theorem \ref{treematrix}:}
We use Grassmann variables to write the determinant of a matrix with one line and one raw deleted
as a Grassmann integral with two corresponding sources:
\be
\det A^{11}=
\int ({\rm d}{\Br\psi} {\rm d}\psi)
\ (\psi_1 {\Br \psi}_1)
e^{-{\Br\psi}A\psi}
\ee
The trick is to use (\ref{sumnulle}) to write
\be
{\Br \psi}A\psi=
\sum_{i,j=1}^n
({\Br\psi}_i-{\Br\psi}_j)A_{ij}\psi_j
\ee
Let, for any $j$, $1\le j\le n$, $B_j\eqdef\sum_{i=1}^n A_{ij}$, one
then obtains by Lemma \ref{exercisegrass}:
\be
\det A^{11} =
\int {\rm d}{\Br\psi} {\rm d}\psi
\ (\psi_1 {\Br \psi}_1)
\exp\lp
-\sum_{i,j=1}^n A_{ij}({\Br\psi}_i-{\Br\psi}_j)\psi_j
\rp
\ee
\be
=
\int {\rm d}{\Br\psi} {\rm d}\psi
\ (\psi_1 {\Br \psi}_1)
\left[
\prod_{i,j=1}^n
\lp 1-A_{ij}({\Br\psi}_i-{\Br\psi}_j)\psi_j \rp
\right]
\ee
by the Pauli exclusion principle. We now expand to get
\be
\det A^{11} =
\sum_{\cG}
\lp
\prod_{\ell=(i,j)\in\cG}(-A_{ij})
\rp
\Om_{\cG}
\ee
where $\cG$ is {\em any} subset of $[n]\times[n]$, and we used the notation
\be
\Om_{\cG}\eqdef
\int {\rm d}{\Br\psi} {\rm d}\psi
\ (\psi_1 {\Br \psi}_1)
\lp
\prod_{(i,j)\in\cG}
\left[ ({\Br\psi}_i-{\Br\psi}_j)\psi_j \right]
\rp
\ee

The theorem will now follow from the following
\begin{lemma} 
$\Om_{\cG}=0$
unless the graph $\cG$
is a tree directed away from 1 in which case
$\Om_{\cG}=1$.
\end{lemma}

\noindent{\bf Proof:}
Trivially, if $(i,i)$ belongs to $\cG$, then the integrand of
$\Om_{\cG}$ contains a factor ${\Br\psi}_i-{\Br\psi}_i=0$ and
therefore $\Om_{\cG}$ vanishes. 

But the crucial observation is that if 
there is a loop in $\cG$ then again $\Om_{\cG}=0$.
This is because then the integrand of $\Om_{\cF,\cR}$ contains the factor
\be
{\Br\psi}_{\ta(k)}-{\Br\psi}_{\ta(1)}=
({\Br\psi}_{\ta(k)}-{\Br\psi}_{\ta(k-1)})+\cdots+
({\Br\psi}_{\ta(2)}-{\Br\psi}_{\ta(1)})
\ee
Now, upon inserting this telescoping expansion of the factor
${\Br\psi}_{\ta(k)}-{\Br\psi}_{\ta(1)}$ into the integrand of 
$\Om_{\cF,\cR}$, the latter breaks into a sum of $(k-1)$ products.
For each of these products, there exists an $\al\in\ZZ/k\ZZ$
such that the factor $({\Br\psi}_{\ta(\al)}-{\Br\psi}_{\ta(\al-1)})$
appears {\em twice} : once with the $+$ sign from the telescopic
expansion of $({\Br\psi}_{\ta(k)}-{\Br\psi}_{\ta(1)})$, and once more
with a $+$ (resp. $-$) sign if $(\ta(\al),\ta(\al-1))$
(resp. $(\ta(\al-1),\ta(\al))$) belongs to $\cF$.
Again, the Pauli exclusion principle entails that $\Om_{\cG}=0$.

Now  every connected component of $\cG$ must contain 
1, otherwise there is no way to saturate the $d\psi_1$ integration.

This means that $\cG$ has to be a directed tree on $\{1,... n\}$.
It remains only to see now that $\cG$ has to be directed away from 1,
which is not too difficult.
\endproof

Now Theorem \ref{treematrix} follows immediately.
\endproof

\subsection{Non-commutative hyperbolic polynomials, the non-covariant case}

\label{hyperbo}

Since the Mehler kernel is still quadratic in position space it is possible
to also integrate explicitly all positions to reduce Feynman amplitudes
of e.g. non-commutative $\Phi^{\star 4}_4$ purely to parametric formulas, but of course
the analogs of Symanzik polynomials are now hyperbolic polynomials which encode 
the richer information about ribbon graphs. These polynomials
were first computed in \cite{gurauhypersyman}
in the case of the non-covariant vulcanized $\Phi^{\star 4}_4$ theory. The computation
relies essentially on a Grassmann variable analysis of Pfaffians 
which generalizes the tree matrix theorem of the previous section.

Defining the antisymmetric matrix $\sigma$ as
\begin{align}
\sigma=&\begin{pmatrix} \sigma_2 & 0 \\ 0 & \sigma_2 \end{pmatrix} \mbox{ with}\\
\sigma_2=&\begin{pmatrix} 0 & -i \\ i & 0 \end{pmatrix}
\end{align}
\noi
the $\delta-$functions appearing in the vertex contribution can be
rewritten as an integral over some new variables $p_V$. We refer to these
variables as to {\it hypermomenta}. Note that one associates such
a hypermomenta $p_V$ to any vertex $V$ {\it via} the relation
\begin{align}
\label{pbar1}
\delta(x_1^V -x_2^V+x_3^V-x_4^V ) =& \int  \frac{d p'_V}{(2 \pi)^4}
e^{ip'_V(x_1^V-x_2^V+x_3^V-x_4^V)}\nonumber\\
=&\int  \frac{d p_V}{(2 \pi)^4}
e^{p_V \sigma (x_1^V-x_2^V+x_3^V-x_4^V)} \ .
\end{align}
\noi

Consider a particular ribbon graph $G$.
Specializing to dimension 4 and choosing a particular root vertex $\bar{V}$ of the graph,
one can write the Feynman amplitude for $G$ in the condensed way
\begin{align}
\label{a-condens}
{\cal A}_G =& \int \prod_{\ell}\big[\frac{1-t_{\ell}^2}{t_{\ell}}\big]^{2} d\alpha_{\ell} \int d x d p e^{-\frac
{\Omega}{2} X G X^t}
\end{align}
where $t_\ell = \tanh\frac{\alpha_\ell}{2}$, 
$X$ summarizes all positions and hypermomenta and $G$ is a certain quadratic form. If we call 
$x_e$ and $p_{\bar{V}}$ the external variables we can decompose $G$ according to an internal 
quadratic form $Q$, an external one $M$ and a coupling part $P$ so that 
\begin{align}
\label{defX}
X =& \begin{pmatrix}
x_e & p_{\bar{V}} & u & v & p\\
\end{pmatrix} \ \ , \ \  G= \begin{pmatrix} M & P \\ P^{t} & Q \\
\end{pmatrix}\ ,
\end{align}
Performing the gaussian integration over all internal variables one obtains:
\begin{align}
\label{aQ}
{\cal A}_G  =& \int \big[\frac{1-t^2}{t}\big]^{2} d\alpha\frac{1}{\sqrt{\det Q}}
e^{-\frac{\Ot}{2}
\begin{pmatrix} x_e & \bar{p} \\
\end{pmatrix} [M-P Q^{-1}P^{t}]
\begin{pmatrix} x_e \\ \bar{p} \\
\end{pmatrix}}\ .
\end{align}
\noi
This form allows to define the polynomials $HU_{G, \bar{v}}$ and $HV_{G, \bar{v}}$, analogs 
of the Symanzik polynomials $U$ and $V$ of the commutative case (see \eqref{symanzik}). 
They are defined by
\begin{align}
\label{polv}
{\cal A}_{{\bar V}}  (\{x_e\},\;  p_{\bar v}) =& K'  \int_{0}^{\infty} \prod_l  [ d\alpha_l (1-t_l^2)^{2} ]
HU_{G, \bar{v}} ( t )^{-2}   
e^{-  \frac {HV_{G, \bar{v}} ( t , x_e , p_{\bar v})}{HU_{G, \bar{v}} ( t )}}.
\end{align}
\noi
They are polynomials in the set of variables $t_\ell$ ($\ell =1,\ldots, L$), the hyperbolic tangent of the
half-angle of the parameters $\alpha_\ell$.

Using now \eqref{aQ} and \eqref{polv} the polynomial $HU_{G,\bar{v}}$ writes
\begin{align}
\label{hugvq}
HU_{\bar{v}}=&(\det Q)^\frac14 \prod_{\ell=1}^L t_\ell
\end{align}

The main results (\cite{gurauhypersyman}) are

\begin{itemize}
 \item The polynomials $HU_{G, \bar{v}}$ and $HV_{G, \bar{v}}$ 
have a strong \emph{positivity property}. Roughly speaking they are sums of monomials
with positive integer coefficients. This positive integer property 
comes from the fact that each such coefficient is the square of a Pfaffian 
with integer entries,

\item {Leading terms} can be identified in a given ``Hepp sector'', 
at least for \textit{orientable graphs}.
A Hepp sector is a complete ordering of the $t$ parameters.
These leading terms which can be shown strictly positive in $HU_{G, \bar{v}}$ correspond to super-trees
which are the disjoint union of a tree in the direct graph and a tree in the dual graph.
Hypertrees in a graph with $n$ vertices and $F$ faces have therefore $n+F-2$ lines.
(Any connected graph has hypertrees, and under reduction of the hypertree, the graph becomes
a hyperrosette). Similarly one can identify ``super-two-trees'' $HV_{G, \bar{v}}$ 
which govern the leading behavior of $HV_{G, \bar{v}}$ in any Hepp sector.
\end{itemize}

From the second property, one can deduce the \textit{exact power counting} of any orientable
ribbon graph of the theory, just as in the matrix base.

Let us now borrow from \cite{gurauhypersyman} some examples of these hyperbolic polynomials. 
We put $s =(4\theta\Omega)^{-1}$.
\begin{figure}[htb] 
  \centering
  \includegraphics{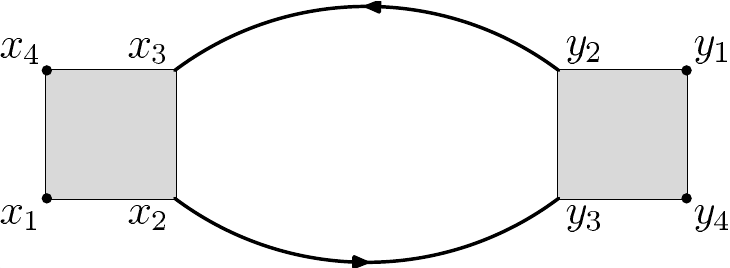}
  \caption{The bubble graph}
  \label{figex1}
\end{figure}
For the bubble graph of figure \ref{figex1}:
\begin{align}
HU_{G,v}=&(1+4s^2)(t_1+t_2+t_1^2t_2+t_1t_2^2)\,,\nonumber\\
HV_{G,v}=&t_2^2\Big{[}p_2+2s(x_4-x_1)\Big{]}^2+t_1t_2\Big{[}2p_2^2+(1+16s^4)(x_1-x_4)^2
                \Big{]}\,,\nonumber\\
                &+t_1^2\Big{[}p_2+2s(x_1-x_4)\Big{]}^2\nonumber\\
\end{align}

For the sunshine graph fig.~\ref{figex2}:
\begin{figure}[htb] 
  \centering
  \includegraphics{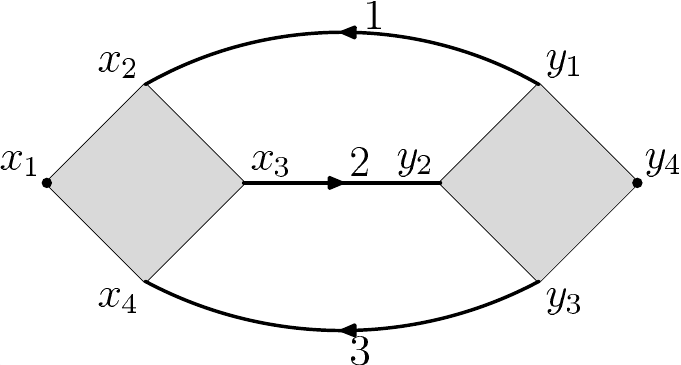}
  \caption{The Sunshine graph}
  \label{figex2}
\end{figure}
\begin{align}
  HU_{G,v}=&\Big{[} t_1t_2+t_1t_3+t_2t_3+t_1^2t_2t_3+t_1t_2^2t_3+t_1t_2t_3^2\Big{]}
  (1+8s^2+16s^4)\nonumber\\
  &+16s^2(t_2^2+t_1^2t_3^2)\, ,\nonumber\\
\end{align}

For the non-planar sunshine graph (see fig.~\ref{figex3}) we have:
\begin{figure}[htb] 
  \centering
  \includegraphics{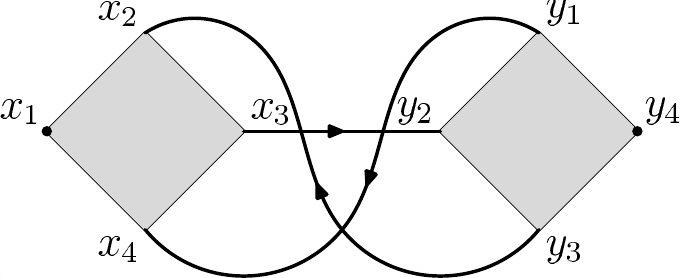}
  \caption{The non-planar sunshine graph}
  \label{figex3}
\end{figure}
\begin{align}
  HU_{G,v}=&\Big{[} t_1t_2+t_1t_3+t_2t_3+t_1^2t_2t_3+t_1t_2^2t_3+t_1t_2t_3^2\Big{]}
  (1+8s^2+16s^4)\nonumber\\
  &+4s^2\Big{[}1+t_1^2+t_2^2+t_1^2t_2^2+t_3^2+t_1^2t_3^2+t_2^2t_3^2+
  t_1^2t_2^2t_3^2\Big{]}\,,\nonumber
\end{align}
We note the improvement in the genus with respect to its planar counterparts.

For the broken bubble graph (see fig. \ref{figex4}) we have:
\begin{figure}[htb] 
  \centering
  \includegraphics{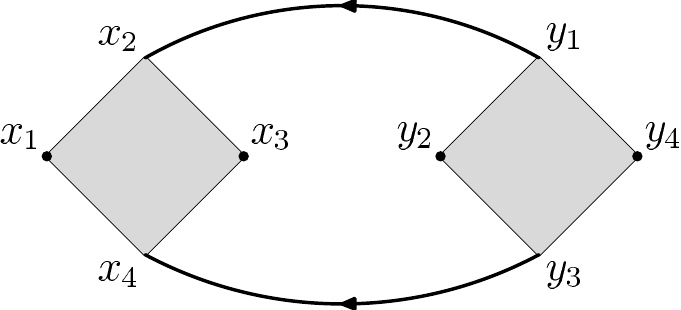}
  \caption{The broken bubble graph}
  \label{figex4}
\end{figure}
\begin{align}
  HU_{G,v}=&(1+4s^2)(t_1+t_2+t_1^2t_2+t_1t_2^2)\,,\nonumber\\
  HV_{G,v}=& t_2^2 \Big{[}4s^2(x_1+y_2)^2+(p_2-2s(x_3+y_4))^2\Big{]}+t_1^2\Big{[}p_2
  +2s(x_3-y_4) \Big{]}^2\,,\nonumber\\
  &+t_1t_2\Big{[}8s^2y_2^2+2(p_2-2sy_4)^2+(x_1+x_3)^2+16s^4(x_1-x_3)^2\Big{]}\nonumber\\
  &+t_1^2t_2^24s^2(x_1-y_2)^2\,,\nonumber
\end{align}
Note that $HU_{G,v}$ is identical to the one of the bubble with only one broken face. 
The power counting improvement comes from the broken face and can be seen only in $HV_{G,v}$. 

\begin{figure}[htb] 
  \centering
  \includegraphics{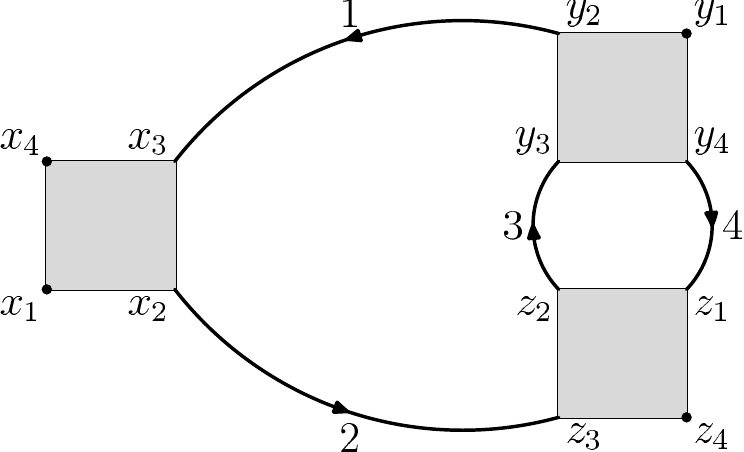}
  \caption{The half-eye graph}
  \label{figeye}
\end{figure}
Finally, for the half-eye graph (see Fig. \ref{figeye}), we start by defining:
\begin{align}
  A_{24}=&t_1t_3+t_1t_3t_2^2+t_1t_3t_4^2+t_1t_3t_2^2t_4^2\,.
\end{align}
The $HU_{G,v}$ polynomial with fixed hypermomentum corresponding to the vertex with two external 
legs is:
\begin{align}\label{hueye1}
  HU_{G,v_1}=&(A_{24}+A_{14}+A_{23}+A_{13}+A_{12})(1+8s^2+16s^4)\nonumber\\
  &+t_1t_2t_3t_4(8+16s^2+256s^4)+4t_1t_2t_3^2+4t_1t_2t_4^2\nonumber\\
  &+16s^2(t_3^2+t_2^2t_4^2+t_1^2t_4^2+t_1^2t_2^2t_3^2)\nonumber\\
  &+64s^4(t_1t_2t_3^2+t_1t_2t_4^2)\,,
\end{align}
whereas with another fixed hypermomentum we get:
\begin{align}
HU_{G,v_2}=&(A_{24}+A_{14}+A_{23}+A_{13}+A_{12})(1+8s^2+16s^4)\nonumber\\
&+t_1t_2t_3t_4(4+32s^2+64s^4)+32s^2t_1t_2t_3^2+32s^2t_1t_2t_4^2\nonumber\\
&+16s^2(t_3^2+t_1^2t_4^2+t_2^2t_4^2+t_1^2t_2^3t_3^2)\,.\label{hueye2}
\end{align}

Note that the leading terms are identical and the choice of the root perturbs only the non-leading ones. 
Moreover note the presence of the $t_3^2$ term. Its presence can be understood by the fact that in the 
sector $t_1,t_2,t_4>t_3$ the subgraph formed by the 
lines $1,2,4$ has two broken faces. This is the sign of 
a power counting improvement due to the additional broken face in that sector. 
To exploit it, we have just to integrate over the variables of line $3$
in that sector, using the second polynomial $HV_{G',v}$ for the triangle subgraph $G'$ 
made of lines $1,2,4$.

\subsection{Non-commutative hyperbolic polynomials, the covariant case}

In the \textit{covariant case} the diagonal coefficients on the long variables disappear
but there are new antisymmetric terms proportional to $\Omega$ due to 
the propagator oscillations. 

It is possible to reproduce easily the positivity theorem of the previous non-covariant
case, because we still have sums of squares of Pfaffians. But identifying the leading terms
of the polynomials under a rescaling associating to a subgraph is more difficult.
It is easy to see that for transcendental values of $\Omega$, the desired leading terms
cannot vanish because that would correspond to 
$\Omega$ being the root of a polynomial with integer coefficients. But 
power counting under a transcendentality condition is not very satisfying,
especially because continuous RG flows also necessarily cross non transcendental points.

But thanks to a slightly more difficult analysis inspired
by \cite{vignes-tourneret06:PhD} and which involve
a kind of new fourth Filk move, it is possible to prove that 
except again for some special cases of four point graphs with two 
broken faces, the power counting goes through at $\Omega <1$.

The corresponding analysis together with many examples 
are given in \cite{RivTan}. 

The covariant case at $\Omega =1$, also called the \emph{self-dual} covariant  case
 is very interesting, because it may be the most relevant for
the study of e.g. the quantum Hall effect. 
Apparently it corresponds to a very degenerate
non renormalizable situation because even the four point function
has non logarithmic divergences as can be seen easily in the 
matrix basis, where the propagator is now either $1/(2m+A)$ or $1/(2n+A)$
depending on the sign of the ``magnetic field" $\Omega$ . But 
there is a huge gauge invariance and we feel that the Ward identities of section
\ref{Noghost} should allow renormalization of the theory even in that case.

Let us also recall that the parametric representation 
can be used to derive the dimensional regularization of the
theory, hence perturbative quantum field theory
on non-integer-dimensional Moyal space, and the associated 
dimensional renormalization which may be useful
for renormalizing non commutative gauge theories
\cite{GurTan}.

\section{Conclusion}

Non-commutative QFT seemed initially to have non-renormalizable divergencies,
due to UV/IR mixing. But following the Grosse-Wulkenhaar breakthrough, 
there has been recent rapid progress in our understanding of renormalizable QFT on Moyal spaces.
We can already propose a preliminary classification of these models into 
different categories, according to the behavior of their propagators:
\begin{itemize}
\item ordinary models at $0 < \Omega < 1$ such as $\Phi^{\star 4}_4$ 
(which has non-orientable graphs) or $(\bar\phi\phi)^2$ models 
(which has none). Their
propagator, roughly $(p^2 + \Omega^2 \tilde x^2 + A)^{-1}$ is LS covariant and has good decay both in 
matrix space (\ref{th1}-\ref{thsummax}) and direct space (\ref{tanhyp}). They have non-logarithmic mass 
divergencies and definitely require ``vulcanization'' i.e. the $\Omega$ term.
\item self-dual models at $\Omega = 1$ in which the propagator is LS 
invariant.
Their propagator is even better. In the matrix base 
it is diagonal, e.g. of the form $G_{m,n}= (m+n + A)^{-1}$, where
$A$ is a constant. The supermodels seem generically 
ultraviolet fixed points of the ordinary models, at which non-trivial Ward identities
force the vanishing of the beta function. The flow of $\Omega$ to the $\Omega = 1$ fixed
point is very fast (exponentially fast in RG steps).
\item covariant models such as orientable versions of LSZ or Gross-Neveu (and presumably 
orientable gauge theories
of various kind: Yang-Mills, Chern-Simons...). They may have only logarithmic divergencies and 
apparently no perturbative UV/IR mixing. However the vulcanized version 
still appears the most generic framework for their treatment.
The propagator is then roughly $(p^2 + \Omega^2 \tilde x^2 + 2\Omega \tilde x \wedge p)^{-1}$.
In matrix space this propagator shows definitely a weaker decay (\ref{mainbound1})
than for the ordinary models, because of the presence of a non-trivial saddle point.
In direct space the propagator no longer decays with respect to the long variables, 
but only oscillates. Nevertheless the main lesson is that 
in matrix space the weaker decay can still be used; and in
$x$ space the oscillations can never be completely killed by the vertices
oscillations. Hence these models retain therefore essentially the power counting 
of the ordinary models, up to some nasty details concerning the
four-point subgraphs with two external faces.
Ultimately, thanks to a little conspiration in which the
four-point subgraphs with two external faces are renormalized by the mass renormalization,
the covariant models remain renormalizable. This is the main message of \cite{RenNCGN05,vignes-tourneret06:PhD}.
\item self-dual covariant models which are of the previous type but at $\Omega = 1$. 
Their propagator in the 
matrix base is diagonal and depends only on one index $m$
(e.g. always the left side of the ribbon). It is of the form $G_{m,n}= (m + A)^{-1}$.
In $x$ space the propagator oscillates in a way that often
exactly compensates the vertices oscillations. These models have definitely
worse power counting than in the ordinary case, with e.g. quadratically divergent 
four point-graphs (if sharp cut-offs are used). Nevertheless Ward identities
can presumably still be used to show that they can still be renormalized. This 
probably requires a much larger conspiration to generalize
the Ward identities of the supermodels.
\end{itemize}
Notice that the status of non-orientable covariant theories is not yet clarified.

Parametric representation can be derived in the non-commutative case. It implies 
hyperbolic generalizations of the Symanzik polynomials which condense the 
information about the rich topological structure of a ribbon graph.
Using this representation, dimensional regularization and dimensional renormalization 
should extend to the non-commutative framework.

Remark that trees, which are the building blocks of the Symanzik polynomials, are also 
at the heart of (commutative) constructive theory, whose philosophy
could be roughly summarized as ``You shall use trees\footnote{These trees
may be either true trees of the graphs in the Fermionic case or trees associated to 
cluster or Mayer expansions in the Bosonic case, but this distinction is not essential.}, 
but you shall \textit{not} develop their loops or else you shall diverge''.
It is quite natural to conjecture that hypertrees, which are the natural non-commutative
objects intrinsic to a ribbon graph, should play a key combinatoric role
in the yet to develop non-commutative constructive field theory.

In conclusion we have barely started to scratch the world of 
renormalizable QFT on non-commutative spaces.
The little we see through the narrow window now open
is extremely tantalizing. There {exists renormalizable NCQFTs} e.g. $\Phi^{\star 4}$ on ${\mathbb R}^4_\theta$, 
Gross-Neveu on ${\mathbb R}^2_\theta$ and they 
enjoy better properties than their commutative counterparts,
since they have no Landau ghosts. The constructive program looks \emph{easier}
on non commutative geometries than on commutative ones.
Non-commutative non relativistic field theories with a chemical potential
seem the right formalism for a study {ab initio} of various problems in 
presence of a magnetic field, and in particular of the quantum Hall effect.
The correct scaling and RG theory of this effect presumably requires to build
a very singular theory (of the self-dual covariant type) because of the huge
degeneracy of the Landau levels. To understand this theory and 
the gauge theories on non-commutative spaces
seem the most obvious challenges ahead of us.
An exciting possibility is that the
non-commutativity of space time which killed the Landau ghost
might be also a good substitute to supersymmetry  
by taming ultraviolet flows
without requiring any new particles.


\end{document}